\documentclass[10pt]{article}

\usepackage{euscript}
\usepackage{etoolbox}
\usepackage{comment}
\newtoggle{colt}
\togglefalse{colt}
\newtoggle{icml}
\togglefalse{icml}
\newcommand{\colt}[1]{\iftoggle{colt}{#1}{}}
\newcommand{\arxiv}[1]{\iftoggle{colt}{}{#1}}

\newcommand{\icml}[1]{\iftoggle{icml}{#1}{}}

\arxiv{
\usepackage[letterpaper, left=1in, right=1in, top=1in,bottom=1in]{geometry}
  \usepackage{parskip}
  \usepackage[colorlinks=true, linkcolor=blue!70!black, citecolor=blue!70!black,urlcolor=black,breaklinks=true]{hyperref}
  \usepackage[dvipsnames]{xcolor}
}

\icml{
  \usepackage{parskip}
    \usepackage[colorlinks=true, linkcolor=blue!70!black, citecolor=blue!70!black,urlcolor=black,breaklinks=true, bookmarks=true]{hyperref}
}

\PassOptionsToPackage{hypertexnames=false}{hyperref}  %

\usepackage{amsmath}
\usepackage{microtype}
\usepackage{hhline}

\arxiv{
\usepackage{amsthm}
}

\usepackage{bbm}
\usepackage{amsfonts}
\usepackage{amssymb}
\usepackage[nameinlink,capitalize]{cleveref}
\usepackage{autonum}

\usepackage{mathtools}

\usepackage{xargs}

\arxiv{
\usepackage{algorithm}

\usepackage{natbib}
\bibliographystyle{plainnat}
\bibpunct{(}{)}{;}{a}{,}{,}
}

\usepackage{xpatch}

\AddToHook{env/lemma/begin}
  {\crefalias{theorem}{lemma}}

\AddToHook{env/proposition/begin}
  {\crefalias{theorem}{proposition}}

\AddToHook{env/corollary/begin}
  {\crefalias{theorem}{corollary}}

\AddToHook{env/definition/begin}
  {\crefalias{theorem}{definition}}

\arxiv{
\theoremstyle{plain}
\newtheorem{theorem}{Theorem}[section]
\newtheorem{proposition}[theorem]{Proposition}
\newtheorem{lemma}[theorem]{Lemma}
\newtheorem{corollary}[theorem]{Corollary}
\newtheorem{remark}{Remark}

\newtheorem{definition}{Definition}
\newtheorem{assumption}{Assumption}

\theoremstyle{definition}

}
\colt{
\newtheorem{assumption}{Assumption}

}

\newcommand{\pfref}[1]{Proof of \cref{#1}}

\renewcommand{\eqref}[1]{\texorpdfstring{\hyperref[#1]{(\ref*{#1})}}{(\ref*{#1})}}

\crefformat{equation}{#2Eq.\,(#1)#3}
\Crefformat{equation}{#2Eq.\,(#1)#3}
\Crefformat{figure}{#2Figure~#1#3}
\Crefformat{assumption}{#2Assumption~#1#3}
\Crefname{assumption}{Assumption}{Assumptions}
\crefname{fact}{Fact}{Facts}
\Crefformat{figure}{#2Figure #1#3}
\Crefformat{assumption}{#2Assumption #1#3}
\crefname{lemma}{Lemma}{Lemmas}
\Crefname{lemma}{Lemma}{Lemmas}

\usepackage{crossreftools}
\arxiv{
\makeatletter
\renewenvironment{proof}[1][Proof]%
{%
  \par\noindent{\bfseries\upshape {#1.}\ }%
}%
{\qed\newline}
\makeatother
}
\colt{
\makeatletter
\renewenvironment{proof}[1][Proof]%
{%
  \par\noindent{\bfseries\upshape {#1.}\ }%
}%
{\jmlrQED}
\makeatother
}

\xpatchcmd{\proof}{\itshape}{\normalfont\proofnameformat}{}{}
\newcommand{\proofnameformat}{\bfseries}

\usepackage[most]{tcolorbox}

\tcbset {
  base/.style={
    arc=0mm, 
    bottomtitle=0.5mm,
    boxrule=0mm,
    colbacktitle=!10!white, 
    coltitle=black, 
    fonttitle=\bfseries, 
    left=2.5mm,
    leftrule=1mm,
    right=3.5mm,
    title={#1},
    toptitle=0.75mm, 
  }
}

\newtcolorbox{mainbox}[1]{
  colframe=blue!10!black,
  colbacktitle=blue!50!black!30!white,
  colback=blue!2!white,
  enhanced,
  fonttitle=\bfseries,
  attach boxed title to top left={yshift=-2.5mm},
  boxed title style={size=small,colframe=blue!40!black,colback=blue!40!black},
  title={\small\textcolor{white}{\textsc{#1}}}
}

\newtcolorbox{minbox}[1]{
  colframe=blue!10!black,
  colbacktitle=blue!50!black!30!white,
  colback=blue!2!white,
  enhanced,
  fonttitle=\bfseries,
}

\DeclarePairedDelimiter{\abs}{\lvert}{\rvert} %
\DeclarePairedDelimiter{\brk}{[}{]}
\DeclarePairedDelimiter{\crl}{\{}{\}}
\DeclarePairedDelimiter{\prn}{(}{)}
\DeclarePairedDelimiter{\nrm}{\|}{\|}
\DeclarePairedDelimiter{\tri}{\langle}{\rangle}

\DeclarePairedDelimiter{\ceil}{\lceil}{\rceil}
\DeclarePairedDelimiter{\floor}{\lfloor}{\rfloor}

\def\ddefloop#1{\ifx\ddefloop#1\else\ddef{#1}\expandafter\ddefloop\fi}
\def\ddef#1{\expandafter\def\csname bb#1\endcsname{\ensuremath{\mathbb{#1}}}}
\ddefloop ABCDEFGHIJKLMNOPQRSTUVWXYZ\ddefloop
\def\ddefloop#1{\ifx\ddefloop#1\else\ddef{#1}\expandafter\ddefloop\fi}
\def\ddef#1{\expandafter\def\csname b#1\endcsname{\ensuremath{\mathbf{#1}}}}
\ddefloop ABCDEFGHIJKLMNOPQRSTUVWXYZ\ddefloop
\def\ddef#1{\expandafter\def\csname sf#1\endcsname{\ensuremath{\mathsf{#1}}}}
\ddefloop ABCDEFGHIJKLMNOPQRSTUVWXYZ\ddefloop
\def\ddef#1{\expandafter\def\csname c#1\endcsname{\ensuremath{\mathcal{#1}}}}
\ddefloop ABCDEFGHIJKLMNOPQRSTUVWXYZ\ddefloop
\def\ddef#1{\expandafter\def\csname h#1\endcsname{\ensuremath{\widehat{#1}}}}
\ddefloop ABCDEFGHIJKLMNOPQRSTUVWXYZ\ddefloop
\def\ddef#1{\expandafter\def\csname hc#1\endcsname{\ensuremath{\widehat{\mathcal{#1}}}}}
\ddefloop ABCDEFGHIJKLMNOPQRSTUVWXYZ\ddefloop
\def\ddef#1{\expandafter\def\csname t#1\endcsname{\ensuremath{\widetilde{#1}}}}
\ddefloop ABCDEFGHIJKLMNOPQRSTUVWXYZ\ddefloop
\def\ddef#1{\expandafter\def\csname tc#1\endcsname{\ensuremath{\widetilde{\mathcal{#1}}}}}
\ddefloop ABCDEFGHIJKLMNOPQRSTUVWXYZ\ddefloop
\def\ddef#1{\expandafter\def\csname #1#1\endcsname{\ensuremath{\mathbb{#1}}}}
\ddefloop ABCDEFGHIJKLMNOPQRSTUVWXYZ\ddefloop
\def\ddef#1{\expandafter\def\csname #1\endcsname{\ensuremath{\mathbb{#1}}}}
\ddefloop ABCDEFGHIJKLMNOPQRSTUVWXYZ\ddefloop
\def\ddef#1{\expandafter\def\csname D#1\endcsname{\ensuremath{\Delta(\mathcal{#1})}}}
\ddefloop ABCDEFGHIJKLMNOPQRSTUVWXYZ\ddefloop
\def\ddefloop#1{\ifx\ddefloop#1\else\ddef{#1}\expandafter\ddefloop\fi}
\def\ddef#1{\expandafter\def\csname scr#1\endcsname{\ensuremath{\mathscr{#1}}}}
\ddefloop ABCDEFGHIJKLMNOPQRSTUVWXYZ\ddefloop

\DeclareMathOperator{\En}{\mathbb{E}}

\DeclareMathOperator*{\argmin}{arg\,min} %

\newcommand{\mc}[1]{\mathcal{#1}}

\newcommand{\wt}[1]{\widetilde{#1}}
\newcommand{\wh}[1]{\widehat{#1}}
\newcommand{\wb}[1]{\widebar{#1}}

\renewcommand{\M}{M}

\NewDocumentCommand{\divprn}{s m}{%
    \IfBooleanTF{#1}{\bigl(#2\bigr)}{(#2)}%
}
\newcommand{\usedivprn}[2]{\IfBooleanTF{#1}{\divprn*{#2}}{\divprn{#2}}}

\NewDocumentCommand{\kl}{s m m}{D_{\mathsf{KL}}\usedivprn{#1}{#2\,\|\,#3}}
\NewDocumentCommand{\Dkl}{s m m}{D_{\mathsf{KL}}\usedivprn{#1}{#2\,\|\,#3}}
\NewDocumentCommand{\Dce}{s m m}{D_{\mathsf{CE}}\usedivprn{#1}{#2\,\|\,#3}}
\NewDocumentCommand{\Dcov}{s O{{\M}} m m}{\cE_{#2}\usedivprn{#1}{#3\,\|\,#4}}
\NewDocumentCommand{\Dcmp}{s O{{\M}} m m}{\mathsf{d}_{#2}\usedivprn{#1}{#3\,\|\,#4}}
\newcommand{\Dklshort}{D_{\mathsf{KL}}}
\NewDocumentCommand{\Dhel}{s m m}{D_{\mathsf{H}}\usedivprn{#1}{#2,#3}}
\NewDocumentCommand{\Dgen}{s m m}{D\usedivprn{#1}{#2\dmid{}#3}}
\NewDocumentCommand{\Dgent}{s m m}{D\ind{t}\usedivprn{#1}{#2\dmid{}#3}}
\NewDocumentCommand{\Dhels}{s m m}{D^{2}_{\mathsf{H}}\usedivprn{#1}{#2,#3}}
\NewDocumentCommand{\Ddel}{s m m}{D_{\Delta}\usedivprn{#1}{#2,#3}}

\NewDocumentCommand{\Dchis}{s m m}{D_{\chi^2}\usedivprn{#1}{#2\dmid{}#3}}

\NewDocumentCommand{\Dtris}{s m m}{D^{2}_{\Delta}\usedivprn{#1}{#2,#3}}
\NewDocumentCommand{\Dtri}{s m m}{D_{\Delta}\usedivprn{#1}{#2,#3}}
\NewDocumentCommand{\Dtv}{s m m}{D_{\mathsf{TV}}\usedivprn{#1}{#2,#3}}
\NewDocumentCommand{\DTV}{s m}{D_{\mathsf{TV}}\usedivprn{#1}{#2}}
\NewDocumentCommand{\Dtvs}{s m m}{D^2_{\mathsf{TV}}\usedivprn{#1}{#2,#3}}
\NewDocumentCommand{\Dbl}{s m m}{D_{\mathsf{BL}}\usedivprn{#1}{#2,#3}}
\NewDocumentCommand{\dkl}{s m m}{d_{\mathsf{KL}}\usedivprn{#1}{#2\,\|\,#3}}

\NewDocumentCommand{\Dren}{s O{\lambda} m m}{\mathsf{D}_{#2}\usedivprn{#1}{#3\,\|\,#4}}
\NewDocumentCommand{\Ren}{s O{\lambda} m m}{\mathsf{R}_{#2}\usedivprn{#1}{#3\,\|\,#4}}
\NewDocumentCommand{\Rsys}{s O{\lambda} m m}{\wb{\mathsf{R}}_{#2}\usedivprn{#1}{#3\,\|\,#4}}
\NewDocumentCommand{\Dsys}{s O{\lambda} m m}{\bar{\mathsf{D}}_{#2}\usedivprn{#1}{#3\,\|\,#4}}

\newcommand{\FI}[2]{\mathsf{FI}\prn*{#1\,\|\,#2}}

\newcommand{\ind}{\mathbbm{1}}    %
\newcommand{\eps}{\epsilon}

\newcommand{\pclip}[2][\B]{\mathsf{Clip}_{#1}(#2)}
\newcommand{\trunc}[2][\B]{\tau_{#1}(#2)}

\newcommandx{\gam}[3][1=x,2=z]{\gamma_{#2,#3}(#1)}
\newcommandx{\gamp}[3][1=x,2=z]{\dot\gamma_{#2,#3}(#1)}

\newcommandx{\gamz}[4][1=x,2=z,3=\lr,4=\xz]{\gamma_{#2,#3,#4}(#1)}
\newcommandx{\gamzp}[4][1=x,2=z,3=\lr,4=\xz]{\dot\gamma_{#2,#3,#4}(#1)}

\newcommand{\xz}{x_0}

\newcommand{\lr}{r}

\renewcommand{\B}{B}

\newcommand{\nrmF}[1]{\nrm{#1}_{\mathrm{F}}}
\newcommand{\tr}{\mathrm{tr}}

\newcommand{\Cov}{\mathrm{Cov}}

\newcommand{\CLSI}{C_{\mathsf{LSI}}}
\newcommand{\CPI}{C_{\mathsf{PI}}}

\newcommand{\Poi}{\mathsf{Poisson}}

\newcommand{\Bbar}{\wb{B}}

\newcommand{\nuhat}{\wh{\nu}}

\newcommand{\ud}{\mathrm d}

\newcommand{\note}[1]{{\textcolor{red}{[#1]}}}

\newcommand{\matt}[1]{\note{MATT: #1}}

\newcommand{\betaF}{\beta_{\mathsf{H}}}
\newcommand{\kappaF}{\kappa_{\mathsf{H}}}

\newcommand{\target}{\pi}

\newcommand{\Zbar}{\wb{Z}}
\newcommand{\nubar}{\wb{\nu}}

\newcommand{\cPt}[1]{\mathscr{P}\subs{#1}}

\renewcommand{\d}{\mathrm{d}}
\newcommand{\dt}{\d t}
\newcommand{\dBt}{\d B_t}
\newcommand{\That}{\widehat{\T}}
\newcommand{\bPhat}{\widehat{\bP}}
\newcommand{\bPstar}{\bP^\star}

\newcommand{\wstar}{w^\star}

\newcommand{\isq}{^{-1/2}}
\newcommand{\msf}[1]{{\mathsf{#1}}}
\newcommand{\inner}[1]{{\langle{#1}\rangle}}
\newcommand{\mmid}{\,\lVert\,}
\newcommand{\Breg}{{\mathsf{Breg}}}
\newcommand{\Tbar}{\wb{\T}}

\newcommand{\leqsim}{\approxleq}

\DeclarePairedDelimiter{\nrmop}{\|}{\|_{\mathrm{op}}}

\newcommand{\whp}[1][\delta]{with probability at least $1-#1$}

\newcommand{\Unif}{\mathsf{Unif}}

\newcommand{\var}{\mathrm{Var}}
\newcommand{\Var}{\var}

\newcommand{\subs}[1]{_{{\scriptscriptstyle#1}}}

\newcommand{\pihat}{\wh{\pi}}

\newcommand{\Rbar}{\wb{R}}

\newcommand{\approxleq}{\lesssim}

\newcommand{\Id}{I}

\renewcommand{\ind}[1]{^{{\scriptscriptstyle#1}}}

\newcommand{\indic}{\mathbb{I}}

\newcommand{\poly}{\mathrm{poly}}
\newcommand{\polylog}{\mathrm{polylog}}

\newcommand{\dmid}{\;\|\;}

\newcommand{\unif}{\mathsf{Unif}}

\newcommand{\deq}{\coloneqq}

\newcommand{\muhat}{\widehat{\mu}}

\def\multiset#1#2{\ensuremath{\left(\kern-.3em\left(\genfrac{}{}{0pt}{}{#1}{#2}\right)\kern-.3em\right)}}

\newcommand{\what}{\wh{w}}

\newcommand{\phat}{\wh{p}}

\newcommand{\norm}[1]{\left \lVert #1 \right \rVert}

\newcommand{\st}{\star}

\newcommand{\pstar}{p^\st}

\input{widebar}

\usepackage[utf8]{inputenc} %
\usepackage[T1]{fontenc}    %
\usepackage{url}            %
\usepackage{booktabs}       %
\usepackage{amsfonts}       %
\usepackage{nicefrac}       %
\usepackage{microtype}      %
\usepackage{makecell}
\usepackage{enumitem}
\usepackage{breakcites}
\usepackage{mathrsfs}
\usepackage[most]{tcolorbox}

\usepackage[normalem]{ulem}

\usepackage{algorithm}
\usepackage{verbatim}
\usepackage{algorithmic}

\icml{

}

\arxiv{
\newcommand{\jmlrQED}{\qed}
}

\usepackage{multicol}
\usepackage{colortbl}
\usepackage{setspace}
\usepackage{transparent}
\usepackage{upgreek}

\usepackage{inconsolata}
\usepackage[scaled=.90]{helvet}
\usepackage{xspace}

\icml{
\setlist[enumerate]{leftmargin=*}
\setlist[itemize]{leftmargin=*}
}

\usepackage{graphicx}
\icml{\usepackage{subfigure}}

\usepackage[suppress]{color-edits}
\addauthor{fc}{blue}
\addauthor{sc}{orange}
\addauthor{cd}{pink}
\addauthor{sr}{red}

\usepackage{parskip}

\let\oldparagraph\paragraph

\renewcommand{\paragraph}[1]{\oldparagraph{#1.}}

\newcommand{\normal}[1]{\mathsf{N}\prn*{#1}}

\makeatletter
\g@addto@macro\appendix{%
  \crefalias{section}{appendixsection}%
  \crefalias{subsection}{appendixsubsection}%
  \crefalias{subsubsection}{appendixsubsubsection}%
}
\makeatother
\crefname{appendixsection}{Appendix}{Appendices}
\Crefname{appendixsection}{Appendix}{Appendices}
\crefname{appendixsubsection}{Appendix}{Appendices}
\Crefname{appendixsubsection}{Appendix}{Appendices}
\crefname{appendixsubsubsection}{Appendix}{Appendices}
\Crefname{appendixsubsubsection}{Appendix}{Appendices}
\crefname{assumption}{Assumption}{Assumptions}
\Crefname{assumption}{Assumption}{Assumptions}

\arxiv{
\title{Exact simulation of diffusions and improved algorithms for log-concave sampling}

\author{
   Fan Chen \\ {\small MIT} \\ {\small \texttt{fanchen@mit.edu}} 
   \and Sinho Chewi \\ {\small Yale University} \\ {\small \texttt{sinho.chewi@yale.edu}} 
   \and Alexander Rakhlin \\ {\small MIT} \\ {\small \texttt{rakhlin@mit.edu}}
   \and Matthew S.\ Zhang \\ {\small MIT} \\ {\small \texttt{shuns436@mit.edu}} 
 }
}

\newtcolorbox{resultbox}{
  enhanced,
  width=0.96\textwidth,
  colback=white,
  colframe=black,
  boxrule=0.8pt,
  arc=2mm,
  left=6mm,
  right=6mm,
  top=4mm,
  bottom=4mm,
  before=\begin{center},
  after=\end{center},
  before upper={
  \centering
  \setlength{\baselineskip}{1.25\baselineskip}
  \parindent=0pt
}
}

\begin{document}

\maketitle

\begin{abstract}%
    We study exact simulation of diffusions via rejection sampling on path space using unbiased estimators of the density ratio obtained from Girsanov's theorem.
    When applied to the underdamped Langevin diffusion, it yields an algorithm for sampling from a strongly log-concave and log-smooth distribution with condition number $\kappa$, in dimension $d$, to accuracy $\varepsilon$ in R\'enyi divergence, in $\widetilde O(\kappa^{2/3} d^{1/3}\,\polylog(1/\varepsilon))$ queries.
    Under a third derivative bound, the dimension dependence improves to $d^{1/5}$.
    This improves substantially over the prior state-of-the-art complexity of $\widetilde O(\kappa d^{1/2}\,\polylog(1/\varepsilon))$ for the Metropolis-adjusted Langevin algorithm, and over the $d^{1/4}$ dimension dependence of Metropolized Hamiltonian Monte Carlo under the same third derivative bound.
    We also present applications to the mirror Langevin diffusion, and for obtaining Fisher information bounds in the non-log-concave case.
\end{abstract}

\section{Introduction}

Drawing samples efficiently from a high-dimensional distribution $\pi \propto \exp(-V)$ over $\R^d$ is a core algorithmic task that underlies Bayesian inference, machine learning, and scientific computing.
Recent progress in the field has largely been driven by deep and fruitful connections with optimization~\citep{JKO, Wib18SamplingOpt, Chewi26Book}, which have inspired new algorithms and increasingly refined convergence guarantees.
These developments point toward a broader goal: just as optimization is built upon a mature theory of complexity, with convex functions as the canonical benchmark class, we seek an analogous theory for sampling, with \emph{log-concave distributions} playing the corresponding role.

To fix the ideas, suppose that $V$ is well-conditioned: $0 \prec \alpha I \preceq \nabla^2 V \preceq \beta I$.
In a standard setting, the goal of sampling is to design an algorithm which uses first-order queries to $V$ and outputs a random variable whose law is $\varepsilon$-close to $\pi$ in, say, total variation distance.
We want to quantify the number of queries in terms of the dimension $d$, the condition number $\kappa$, and the target accuracy $\varepsilon$.
Despite the fundamental nature of this question, the optimal complexity remains unknown.

In this paper, we develop new algorithms for this problem that significantly improve upon the state-of-the-art in several important ways.
To contextualize our results, we first review the literature.

\subsection{Background}

\paragraph{Toward acceleration for sampling}
For the analogous problem of smooth convex optimization, one of the most celebrated results is the existence of \emph{accelerated} methods, as first shown in~\citet{Nes1983Accel}.
Namely, in the well-conditioned setting, whereas the iteration complexity of gradient descent---with constant step size, see~\citet{AltPar25Hedging} for recent developments---incurs linear dependence $\Omega(\kappa)$ on the condition number, there are accelerated methods which achieve $O(\sqrt\kappa)$ dependence, a square root speed-up.

In sampling, and in continuous time, the analogue of gradient descent is the \emph{Langevin diffusion} (LD), and the analogue of Nesterov's accelerated gradient method is the \emph{kinetic} or \emph{underdamped Langevin diffusion} (ULD); see \cref{ssec:uld}.
Although the analogies are formally clear, especially in light of the continuous-time Nesterov ODEs derived in~\citet{SuBoyCan16Nesterov}, proving the sharp accelerated convergence of~\eqref{eq:ULD} was a longstanding open problem until the breakthrough work of~\citet{CaoLuWan23Underdamped}.
Subsequently,~\citet{EbeLor24Lifts} cast their result into the framework of second-order lifts, which relates the acceleration question to classical work in the Markov chain literature on non-reversible lifts~\citep{CheLovPak1999Lifting, DiaHolNea00Nonreversible}. In this context, the square root speed-up is interpreted as a \emph{diffusive-to-ballistic} speed-up.
The original result of~\citet{CaoLuWan23Underdamped} was indeed accelerated but only gave exponential decay in $L^2$, which incurs poor dependence on the initialization.
Even more recently, this issue was resolved through the work of~\citet{Lu26SharpEntropy}, which established accelerated convergence in relative entropy.
See \cref{ssec:related} for further discussion.

These encouraging developments suggest that the moment is opportune for realizing acceleration in sampling.
However, translating the continuous-time accelerated rates into algorithmic gains is a challenging question in and of itself.
Thus far, the only progress in this direction has been the work of~\citet{altschuler2026sc4}.
Via the framework of shifted composition developed in~\citet{AltChe25SCI, AltChe26SCIII}, and combined with the result of~\citet{Lu26SharpEntropy}, it implies that a certain discretization of~\eqref{eq:ULD} achieves a complexity of $\widetilde O(\kappa^{5/6} d^{1/3}/\varepsilon^{2/3})$.
This result achieves sublinear $o(\kappa)$ dependence\footnote{\label{foot:1}We note that, similarly to the existence of cutting-plane methods in optimization, there are samplers for general log-concave distributions which incur $\poly(d,\log\kappa, \log(1/\varepsilon))$ dependence~\citep[see, e.g.,][]{Kook26Thesis}. Hence, a more precise formulation of the acceleration question is to achieve $\sqrt \kappa$ dependence simultaneously with ``good'' (e.g., sublinear) dimension dependence.}, but the approach of direct discretization seems inherently limited: for example, as discussed next, it yields a \emph{low-accuracy} $\poly(1/\varepsilon)$ guarantee.

\paragraph{High-accuracy vs.\ low-accuracy sampling}
Unlike in the setting of optimization, discretization incurs bias: the stationary distribution of a diffusion is generally not preserved by discretization.
Controlling the asymptotic bias generally requires choosing the step size as a function of $\varepsilon$, eventually incurring $\poly(1/\varepsilon)$ iteration complexity.
Typically, this is addressed by incorporating a Metropolis--Hastings filter~\citep{Met+1953MCMC, Has1970MCMC}.
Doing so, however, introduces new considerations, such as implementation of the filter---which necessitates ``simple'' proposal kernels---and control of the acceptance probability.

Applying the filter to the Euler--Maruyama discretization of the Langevin diffusion yields the Metropolis-adjusted Langevin algorithm (MALA), whose mixing time is $\widetilde O(\kappa d^{1/2}\, \polylog(1/\varepsilon))$ from a warm start~\citep{Chewi+21MALA, WuSchChe22MALA}.
In~\citet{AltChe24Warm}, it was shown how to algorithmically obtain such a warm start with comparable complexity, leading to an overall $\widetilde O(\kappa d^{1/2}\,\polylog(1/\varepsilon))$ guarantee.
This rate is matched by a piecewise deterministic Markov process called the zigzag sampler~\citep{LuWan22Zigzag}, and also via the proximal sampler algorithm of~\citet{LeeSheTia21RGO} which reduces sampling to a sequence of localized distributions, which in turn can be implemented via approximate rejection sampling~\citep{FanYuaChe23ImprovedProx, chen2026high}.
This is the current state-of-the-art complexity for high-accuracy samplers.
In particular, despite the improved dependence on $\varepsilon$, the complexity is worse in terms of both $d$ and $\kappa$.

Since MALA is obtained by correcting a discretization of the Langevin diffusion, it is natural to expect that faster algorithms arise by replacing LD with~\eqref{eq:ULD}.
This is closely related to the Metropolized Hamiltonian Monte Carlo (MHMC) algorithm, which is one of the most widely used samplers in practice; indeed, MHMC and its adaptive variants are the workhorse implementations of Markov chain Monte Carlo in many popular software packages such as Stan and TensorFlow.
However, satisfactory guarantees for MHMC have only been established under a higher-order smoothness assumption.
Namely, assume that $\nabla^2 V$ is $\betaF$-Lipschitz in the Frobenius norm, and write $\kappaF \deq \betaF^{2/3}/\alpha$ for the associated condition number.
Then, the mixing time result of~\citet{CheGatJia26MHMC} together with the algorithmic warm start result of~\citet{zhang2026algorithmic} yield a sampler with complexity $\widetilde O((\kappa + \kappaF)\,d^{1/4}\,\polylog(1/\varepsilon))$ in this setting.

\paragraph{Concurrent work}
The concurrent work of~\citet{LuLuo26Zigzag} obtains $\widetilde O(\kappa^{1/2} d^2\,\polylog(1/\varepsilon))$ complexity for the bouncy particle sampler.
In particular, the dependence on $\kappa$ is fully accelerated, albeit with a worsened dimension dependence.
Moreover, due to the obstruction presented in~\citet{MonWan26EntropicPDMP}, it is not possible to improve the dimension dependence by establishing a stronger entropic decay.
Nevertheless, it is a promising future direction to further study piecewise deterministic Markov processes and to compare their complexity guarantees with the ones we have obtained here.

\subsection{Contributions}

In this paper, we propose new algorithms for sampling based on a different approach. Rather than correcting a discrete-time step via a Metropolis--Hastings filter, we correct the law of an entire continuous-time path via rejection sampling. Namely, given a reference path measure $\mathbf P$, we propose a tractable path measure $\mathbf Q$. We then correct this measure using the density ratio formula supplied by Girsanov's theorem. In principle, this yields an exact unbiased scheme for generating sample paths from $\mathbf P$. However, the Girsanov reweighting factor is a functional of the entire path, so this is not immediately implementable.

Our principal observation is that it is nevertheless possible to build an \emph{unbiased estimator} of the density ratio using $O(1)$ queries---provided that the simulation time is chosen appropriately---which suffices for accurate simulation. The mechanism shares broad similarities with the first-order rejection sampling (FORS) method introduced recently in~\cite{chen2026high, chen2026stoc}. We also note that the idea of exact simulation of diffusions was considered in prior works such as~\citet{beskos2005exact}; see \cref{ssec:related} for more discussion.
However, to our knowledge, we are the first to apply this idea in a multi-dimensional context, and to obtain explicit complexity guarantees which improve upon the state-of-the-art.

We apply this scheme to the underdamped Langevin diffusion, which is defined by the SDE
\begin{align}\tag{$\msf{ULD}$}\label{eq:ULD}
    \d X_t = P_t \, \d t\,, \qquad \d P_t = \bigl(- \gamma P_t - \nabla V(X_t)\bigr) \, \d t + \sqrt{2\gamma} \, \d B_t\,.
\end{align}
Note that the SDE is posed in phase space, where $x$ denotes the position coordinate and $p$ the momentum coordinate.
As shown recently in~\citet{Lu26SharpEntropy}, the law of this SDE converges at an accelerated rate in KL divergence---and even R\'enyi divergences, see~\citet{LiLu26SpaceTimeLSI}---to its stationary distribution, which we again denote by $\pi$ in an abuse of notation. Explicitly, $\pi(x,p) \propto \exp(-V(x) - \|p\|^2/2)$, so that the position marginal is our distribution of interest. By various proposals and carefully analyzing the variance of the density ratio estimators, we arrive at the following results.

First, in the case where we only assume two-sided bounds on $\nabla^2 V$, we use the exponential Euler discretization as a proposal to obtain:
\begin{resultbox}
\textbf{Result 1:} Assuming $0 \prec \alpha I \preceq \nabla^2 V \preceq \beta I$, with $\kappa \deq \beta/\alpha$, there is an algorithm that achieves $\Ren[q]{\cdot}{\pi} \le \varepsilon$ using $\widetilde O_q(\kappa^{2/3} d^{1/3}\,\polylog(1/\varepsilon))$ queries in expectation.
\end{resultbox}
Our result improves upon both the $\kappa^{5/6}$ condition number dependence of~\citet{altschuler2026sc4}, and also the $d^{1/2}$ dimension dependence of all known high-accuracy samplers.
Thus, our guarantee is \emph{simultaneously} state-of-the-art with respect to all parameters, at least among algorithms with $O(d)$ dimension dependence (see Footnote~\ref{foot:1}).

Under a third-order derivative bound, we use a more sophisticated proposal that makes Hessian-vector queries.
\begin{resultbox}
\textbf{Result 2:} Assuming $0 \prec \alpha I \preceq \nabla^2 V \preceq \beta I$ and that $\nabla^2 V$ is $\betaF$-Lipschitz in the Frobenius norm, with $\kappa \deq \beta/\alpha$ and $\kappaF \deq \betaF^{2/3}/\alpha$, there is an algorithm that achieves $\Ren[q]{\cdot}{\pi} \le \varepsilon$ using $\widetilde O_q\bigl(\bigl(\kappa^{1/2}+\kappaF^{3/5}d^{1/5}+\kappa^{1/3}(\kappaF+\kappaF^{3/11})d^{2/11}\bigr)\polylog(1/\varepsilon)\bigr)$ queries in expectation.
\end{resultbox}
When $\kappaF = O(1)$, the dimension dependence of our algorithm is $d^{1/5}$.
In particular, this provably outperforms the previous state-of-the-art, MHMC, which by dimensional scaling arguments incurs $d^{1/4}$ complexity even for Gaussian targets.
On the other hand, our algorithm is adaptive in the sense that $\kappaF = 0$ for Gaussians, in which case only the $\kappa^{1/2}$ term survives.
This yields a $\widetilde O_q(\kappa^{1/2}\,\polylog(d, 1/\varepsilon))$ guarantee for sampling from Gaussians, which is near-optimal over this class~\citep{Chewi+23QueryLower}.

Finally, we also highlight two additional applications of potential interest.

\begin{enumerate}
    \item If we instead assume that $V$ is \emph{relatively} convex, Lipschitz, and smooth with respect to a self-concordant mirror map $\phi$, we can apply the diffusion rejection schema above to the mirror Langevin diffusion~\citep{chewi2020exponential}, assuming exact simulation of the mirror Brownian motion. This gives an algorithm which operates in the same setting and under the same assumptions of~\citet{ahn2021efficient}, but improves their low-accuracy complexity guarantee to a high-accuracy one.
    See \cref{sec:mirror}.

    \item Assuming that $V$ is $\beta$-smooth but possibly non-convex, the sampling analogue of approximate first-order stationarity is to generate a sample with small relative Fisher information: $\FI{\cdot}{\pi} \le \varepsilon^2$.
    Recently,~\cite{chewi2026complexity} used the proximal sampler to reduce this task to high-accuracy \emph{log-concave} sampling.
    As we have improved the complexity of the latter, this leads to a new state-of-the-art $\widetilde O(\beta d^{1/3} K_0/\varepsilon^2)$ complexity bound for the former, where $K_0$ is the initial KL divergence to $\pi$.
    Here, the improvement is dimensional: $d^{1/2} \to d^{1/3}$.
    See \cref{sec:fisher}.
    
\end{enumerate}

Overall, we believe that these results are surprising, and amply demonstrate the power of this algorithmic framework.
It is an interesting avenue for future research to develop further applications via this methodology, such as the simulation of general SDEs, and moreover for testing these methods numerically.

\subsection{Related work}\label{ssec:related}

\paragraph{Log-concave sampling}
The problem of sampling from distributions over $\R^d$ using Markov chains has an extensive history, see~\cite{brooks2011handbook}. More recent literature centers on the non-asymptotic complexity of such schemes in high dimensions following an optimization-inspired query complexity model. See~\cite{Chewi26Book} for a general exposition and for further references to the literature.

\paragraph{Accelerated hypocoercive convergence}
The study of convergence rates for degenerate processes, including~\eqref{eq:ULD} as a model case, is known as \emph{hypocoercivity}; see the monograph~\citet{Vil09Hypo}.
The first sharp rate was established in $L^2$ in~\citet{CaoLuWan23Underdamped}, building on the space-time approach of~\citet{Alb+24KineticFP}.
This result was cast in the language of second-order lifts in~\citet{EbeLor24Lifts}, and there have been many subsequent extensions~\citep[e.g.,][]{LuWan22PDMP, Bri23Kinetic, EbeLor24Divergence, Ebe+25FlowPI, Lehec25Kinetic, FanLiLu26SharpL2}.
The first sharp rate in entropy was given recently in~\citet{Lu26SharpEntropy} and extended to hypocoercive hypercontractivity~\citep{LiLu26SpaceTimeLSI}, to non-linear equations~\citep{Mon26Accel}, and to PDMPs~\citep{MonWan26EntropicPDMP}.
Recently,~\citet{Mit+26AccelRHMC} gave a different approach to acceleration for randomized Hamiltonian Monte Carlo.

\paragraph{High-accuracy sampling}
MALA has had a long history of study, with limiting arguments suggesting the dimensional scaling $d^{1/3}$ at stationarity---at least under higher-order smoothness~\citep{roberts1998optimal}.
Within recent literature, the non-asymptotic $d^{1/2}$ complexity was obtained in a line of works~\citep{dwivedi2019log, chen2020fast, LeeSheTia20MALA, Chewi+21MALA, WuSchChe22MALA, chen2023simple}, culminating in the algorithmic warm start result of~\citet{AltChe24Warm}.
On the other hand \cite{Chewi+21MALA} showed that under bounds on $\nabla^2 V$ alone, the $d^{1/2}$ complexity for MALA is tight.

Another route to $d^{1/2}$ high-accuracy sampling is via the proximal sampler~\citep{LeeSheTia21RGO}, together with approximate rejection sampling~\citep{FanYuaChe23ImprovedProx}.
In~\citet{chen2026high}, it was shown that this guarantee can be achieved using queries to $\nabla V$ alone (and not to $V$ itself), and that this is robust to noise in the stochastic gradients~\citep{chen2026stoc}.
The first-order rejection sampling mechanism introduced in these two works provides the basis of the present approach of this paper.
In another direction, we note that the $d^{1/2}$ high-accuracy guarantee can also be extended to composite sampling~\citep{LiuChe26ProxGradSampling}.

As for MHMC, \cite{beskos2013optimal} gave a scaling limit with $d^{1/4}$ dependence at stationarity. Non-asymptotically, a $d^{11/12}$ rate was proven in~\cite{chen2020fast}, while the $d^{1/4}$ rate was proven in~\citet{CheGatJia26MHMC}, both from a warm start. An algorithmic warm start was obtained in~\cite{zhang2026algorithmic}. These works all required $\nabla^3 V$ to be appropriately bounded, and we also adopt that assumption here for our $d^{1/5}$ high-accuracy sampler. %

\paragraph{Mirror Langevin}
The mirror Langevin algorithm has been studied in numerous prior works~\citep{zhang2020wasserstein,chewi2020exponential, ahn2021efficient, jiang2021mirror, li2022mirror, gatmiry2022convergence}. In our work, we most closely follow the analysis of~\cite{ahn2021efficient}.
A Metropolis-adjusted version of mirror Langevin was proposed in~\citet{SriWibWil24MAMLA}, yielding a high-accuracy algorithm; we defer a comparison to \cref{sec:mirror}.

\paragraph{Fisher information complexity}
The study of relative Fisher information bounds as approximate first-order stationarity was initiated in~\cite{Bal+22NonLogConcave}. The algorithmic upper bound was recently improved in~\citet{chewi2026complexity} via a reduction to high-accuracy log-concave sampling. We also mention that lower bounds for this problem were studied in~\cite{chewi2023fisher}.

\paragraph{Exact simulation of diffusions}
The pioneering work of \citet{Wag1988SDE} proposed a method of \emph{exactly} simulating one-dimensional SDEs through rejection sampling, utilizing Girsanov's theorem which expresses the log-density as an integral. Extensions of this idea have been studied in follow-up work on simulating SDEs~\citep[e.g.,][]{Wag1988SDE,beskos2005exact,beskos2006retrospective, Pap11Diffusion,gonccalves2023exact} and also Markov chain Monte Carlo~\citep[etc.]{gonccalves2017barker,vats2022efficient,kakkad2024exact}. In this work, we revisit this idea and apply it to the underdamped Langevin diffusion through the framework of \citet{chen2026high,chen2026stoc}.

\subsection{AI usage}

\scedit{The main idea of the algorithm and its analysis were developed by the authors, building upon the prior works~\citet{chen2026high, chen2026stoc}. LLM tools were used to check for errors, provide the implementation details in~\cref{app:implementation}, and to prepare parts of this manuscript.}

\section{Preliminaries}\label{sec:prelim}

We collect the notation and analytic preliminaries that will be used throughout.  We write points in phase space as $z=(x,p)\in\cZ\deq \RR^d\times \RR^d$.  The position target is $\mu_V(dx)\propto e^{-V(x)}\,dx$, and the invariant distribution of the underdamped Langevin dynamics is
\begin{align}
    \target(dx,dp)\deq \mu_V(dx)\otimes \normal{0,\Id}(dp)\,.
\end{align}
When a divergence is applied to position marginals, $\target$ refers to $\mu_V$; when it is applied to phase-space distributions, $\target$ denotes the product law above.  This convention keeps the notation light.

We assume that $V$ is twice continuously differentiable and $\beta$-smooth in the sense that
\begin{align}\label{eq:beta-smooth}
    -\beta\Id \preceq \nabla^2 V(x)\preceq \beta\Id\,,\qquad \forall x\in\RR^d\,.
\end{align}

\paragraph{Isoperimetric conditions}

\begin{definition}[Poincar\'e]
    A distribution $\pi$ satisfies a \emph{Poincar\'e inequality (PI)} with constant $C$ if for all compactly supported and smooth test functions $h : \R^d\to\R$,
    \begin{align}
        \Var_{X\sim \pi}(h(X)) \le C\,\En_{X\sim\pi}[\nrm{\nabla h(X)}^2]\,.
    \end{align}
    We let $\CPI(\pi)$ be the smallest constant $C$ such that $\pi$ satisfies PI with constant $C$.
\end{definition}

\begin{definition}[Log-Sobolev]
    A distribution $\pi$ satisfies a \emph{log-Sobolev inequality (LSI)} with constant $C$ if for all compactly supported and smooth test functions $h : \R^d \to \R$,
    \begin{align}
        \mathrm{Ent}_{X\sim \pi}(h^2(X))
        \deq \E_{X\sim \pi}\brk[\Big]{h^2(X) \log \frac{h^2(X)}{\E_{X\sim \pi}[h^2(X)]}} \le 2C\,\En_{X\sim\pi}[\nrm{\nabla h(X)}^2]\,.
    \end{align}
    We let $\CLSI(\pi)$ be the smallest constant $C$ such that $\pi$ satisfies LSI with constant $C$.
\end{definition}

It is well known~\citep[e.g.,][]{BGL14} that if $\pi$ is $\alpha$-strongly log-concave (SLC), i.e., $-\log\pi$ is $\alpha$-strongly convex, then it satisfies LSI with constant $1/\alpha$, and if $\pi$ satisfies LSI with constant $1/\alpha$, then it satisfies PI with constant $1/\alpha$.
These represent meaningful enlargements of the class of SLC measures that still permit tractable sampling. For example, unlike SLC, LSI is robust to bounded perturbations of the log-density; unlike LSI\@, PI can capture measures without sub-Gaussian tails (e.g., the two-sided exponential).
See~\citet{Chewi26Book} for further background in the context of sampling.

\paragraph{Oracles}

\begin{definition}[First-order oracle]
Given input $x$, the oracle returns $\nabla f(x)$.
\end{definition}

\begin{definition}[Generalized first-order oracle]\label{def:gen_first_order}
Given input $(x,v)$, the oracle returns $(\nabla f(x), \nabla^2 f(x)\,v)$.
\end{definition}

Note that the generalized first-order oracle can be implemented to arbitrary accuracy via finite differences using the first-order oracle.

\paragraph{Notation}
For any function $f:\RR^d\to\RR$ such that $Z_f\deq \int_{\RR^d} e^{-f(x)}dx<+\infty$, we define $\mu_f$ to be the distribution over $\RR^d$ with density $\mu_f(x)=\frac{1}{Z_f}\,e^{-f(x)}$.

For $\nu \ll \mu$, let $\Ren[\lambda]{\nu}{\mu} \deq \frac{1}{\lambda-1} \log \En_\mu[(\frac{d\nu}{d\mu})^\lambda]$ denote the R\'enyi divergence of order $\lambda > 1$, and let $\Rsys[\lambda]{\nu}{\mu} \deq \max\crl*{\Ren[\lambda]{\nu}{\mu},\Ren[\lambda]{\mu}{\nu}}$ denote the corresponding symmetrized divergence. In addition, for $\lambda\geq1$, we define
\begin{align}
    \Dren{\nu}{\mu} \deq \En_\mu\bigl[(\frac{d\nu}{d\mu})^\lambda\bigr]-1\,,\qquad
    \Dsys{\nu}{\mu}=\max\crl{\Dren{\nu}{\mu},\Dren{\mu}{\nu}}\,.
\end{align}
We further consider the divergences
\begin{align}
    \Dtv{\nu}{\mu}
    &\deq \sup_A{\abs{\nu(A)-\mu(A)}}\,,\qquad
    \Dkl{\nu}{\mu}
    \deq \En_\nu\brk[\Big]{\log\frac{\d\nu}{\d\mu}}\,,\qquad
    \Dchis{\nu}{\mu}
    \deq \En_\mu\brk[\Big]{\prn[\Big]{\frac{\d\nu}{\d\mu}-1}^2}\,.
\end{align}
In~\cref{sec:fisher}, we will also make use of the relative Fisher information $\FI{\nu}{\mu} \deq \En_\nu\brk*{\nrm{\nabla\log(\d\nu/\d\mu)}^2}$.

For an interval $I=[a,b]$, define $\pclip[I]{u}\deq \max\{a,\min\{b,u\}\}$. For $B>0$, we abbreviate $\pclip[B]{u}\deq \pclip[{[-B,B]}]{u}$ and define $\trunc[B]{u}\deq (|u|-B)_+$.

\subsection{Underdamped Langevin dynamics}\label{ssec:uld}

We will use the kinetic, or underdamped, Langevin diffusion on $\cZ$ with friction parameter $\gamma>0$, given in~\eqref{eq:ULD}.
The measure $\target=\mu_V\otimes \normal{0,\Id}$ is invariant for this diffusion.  
For any distribution $\nu$ over $\cZ$, we define $\cPt{T}\nu$ to be the distribution of $(X_T,P_T)$ under~\eqref{eq:ULD} with $(X_0,P_0)\sim \nu$.

We make use of the following results, due to~\citet{LiLu26SpaceTimeLSI, CaoLuWan23Underdamped} respectively.

\begin{theorem}[Accelerated R\'enyi decay]\label{thm:ST-LSI}
Assume that $\mu_V$ is log-concave and satisfies LSI with constant $1/\alpha$. There are universal constants $c,C>0$ such that, if $\gamma \asymp\sqrt\alpha$, for every $q>1$ and every phase-space law $\nu\ll\target$,
\begin{align}
    \Ren*[q]{\cPt{t}\nu}{\target}
    \leq Cq\exp\prn*{-c\gamma t}\,\Ren[q]{\nu}{\target}\,,\qquad t\geq C/\gamma\,.
\end{align}
Consequently, for $K\gtrsim (T\gamma)^{-1}\log(q\Ren[q]{\nu}{\target}/\delta)$,
\begin{align}
    \Ren*[q]{(\cPt{T})^K\nu}{\target}\leq \delta\,.
\end{align}
\end{theorem}

\begin{theorem}[Accelerated $L^2$ decay]\label{thm:ST-PI}
Assume that $\mu_V$ is log-concave and satisfies PI with constant $1/\alpha$. There are universal constants $c,C>0$ such that, if $\gamma\asymp \sqrt\alpha$, then for every phase-space law $\nu\ll \pi$,
\begin{align}
    \Dchis*{\cPt{t}\nu}{\target}
    \leq C\exp\prn*{-c\gamma t}\,\Dchis{\nu}{\target}\,,\qquad t\geq 0\,.
\end{align}
\end{theorem}

\subsection{FORS}

In this work, we use first-order rejection sampling (FORS)~\citep{chen2026high} to simulate diffusion processes.

Consider a proposal distribution $\bQ$ over a \emph{continuous-time} stochastic process $\bZ=(Z_t)_{t\in[0,T]}$.
Given a tilt function $w(\bZ)$, the goal of \cref{alg:fors} is to produce a sample from $\bPstar$, defined via $\frac{\d\bPstar}{\d\bQ}(\bZ)\propto e^{\wstar(\bZ)}$, without access to the \emph{value} of $\wstar(\bZ)$. Specifically, in our instantiation, we always construct a function $W(\xi;\bZ)$ depending on only $O(1)$ entries of $\bZ$ such that $\wstar(\bZ)\approx \En_{\xi\sim \Xi }[W(\xi;\bZ)]$.

\begin{algorithm}
\caption{First-order rejection sampling (FORS)}\label{alg:fors}
\begin{algorithmic}
\STATE \textbf{Input:} Parameter $B > 0$, proposal distribution $\bQ$, estimator function $W$, and sampling distribution $\Xi$ for the auxiliary variable $\xi$.
\FOR{$i=1,2,3,\dotsc$}
\STATE Sample $J\sim \Poi(2\B)$.
\STATE Sample $\xi_1,\cdots,\xi_J\sim \Xi$.
\STATE Sample $(W(\xi_1;\bZ),\cdots,W(\xi_J;\bZ),Z_T)$ following $\bZ\sim \bQ$.
\STATE Output $Z_T$ with probability $\prod_{j=1}^J \pclip[{[0,1]}]{\frac{B+W(\xi_j;\bZ)}{2B}}$.
\ENDFOR
\end{algorithmic}
\end{algorithm}

The following guarantee is from~\citet{chen2026high}.

\begin{theorem}[FORS guarantee]\label{thm:fors}
    \cref{alg:fors} outputs a random point $Z_T$ whose distribution $\phat$ is the marginal of $\bPhat$ such that
    \begin{align}
        \frac{\d \bPhat}{\d \bQ}\propto e^{\what(\bZ)}\,, \qquad
        \what(\bZ)\deq \En_{\xi\sim \Xi} \pclip[B]{W(\xi;\bZ)}\,.
    \end{align}
    The number of queries through $W$ is bounded, with probability at least $1-\delta$, by $3Be^{2B}\log(2/\delta)$.
\end{theorem}

\begin{remark}\label{rmk:FORS-efficient}
Note that in \cref{alg:fors}, if $W(\xi;\cdot)$ depends on only $O(1)$ coordinates, then, with high probability, we need to sample only a set $\cT$ of $O(B)$ coordinates of $\bZ$. In particular, after sorting $\cT$, the problem reduces to sampling from the transition kernel $\bQ(Z_{t+h}=\cdot\mid Z_t)$ for certain given pairs $(t,h)$. This may be easier than constructing an entire trajectory $\bZ$, as in the setting of \cref{sec:ULA2}.
\end{remark}

\paragraph{Error analysis}

The following lemma controls the error of the FORS kernel in terms of exponential moments.

\begin{lemma}[Error of FORS]\label{lem:FORS-error}
Suppose that $q\geq2$, $\bPstar$ is a path measure such that $\frac{\d \bPstar}{\d \bQ}\propto e^{\wstar(\bZ)}$, and let $\pstar$ be the marginal of $Z_T$ under $\bPstar$. Then
\begin{align}
    (1+\Dsys*[q]{\pstar}{\phat})^2\leq \En_{\bPstar}\exp\prn*{ 4q\,\abs{\wstar(\bZ)-w(\bZ)} }\cdot \En_{\bZ\sim \bPstar,\xi\sim \Xi}\exp\prn*{ 4q\,\trunc[B]{W(\xi;\bZ)} }\,,
\end{align} 
where we denote $w(\bZ)\deq \En_{\xi\sim\Xi} W(\xi;\bZ)$. 
If $\wstar=w$, then the sharper bound
\begin{align}
    1+\Dsys*[q]{\pstar}{\phat}
    \leq \En_{\bZ\sim \bPstar,\xi\sim\Xi}
    \exp\prn*{2q\,\trunc[B]{W(\xi;\bZ)}}
\end{align} 
holds.
\end{lemma}

\begin{proof}
Applying \cref{lem:Renyi-to-diff} on path space with common base $\bQ$,
$f=-\wstar$, and $g=-\what$, and then applying data processing,
\begin{align}
    \Dsys*[q]{\pstar}{\phat}
    \leq \Dsys*[q]{\bPstar}{\bPhat}
    \leq \En_{\bPstar} \exp\prn*{2q\,\abs{\wstar(\bZ)-\what(\bZ)}}-1\,,
\end{align}
where we recall $\what(\bZ)\deq \En_{\xi\sim \Xi} \pclip[B]{W(\xi;\bZ)}$. Then, we can decompose
\begin{align}
    \abs{\wstar(\bZ)-\what(\bZ)}
    \leq&~ \abs{\wstar(\bZ)-w(\bZ)}+\abs{\En_{\xi\sim \Xi} \pclip[B]{W(\xi;\bZ)}-\En_{\xi\sim \Xi} W(\xi;\bZ)} \\
    \leq&~ \abs{\wstar(\bZ)-w(\bZ)}+\En_{\xi\sim \Xi} \trunc[B]{W(\xi;\bZ)}\,.
\end{align}
The desired upper bound then follows from Jensen's and the Cauchy--Schwarz inequalities. If $\wstar=w$, the same argument without Cauchy--Schwarz gives the sharper bound.
\end{proof}

\section{ULD without additional smoothness}\label{sec:ULA}

In this section, we describe a simple proposal which nevertheless suffices to establish Result 1.
This also elucidates the mechanism of the diffusion sampling approach, which will lay the foundation for the more complicated proposal in \cref{sec:ULA2}.

\subsection{Proposal}

Fix any starting point $z=(x,p)$.
Consider the following exponential Euler proposal with $X_0=x$ and $P_0=p$:
\begin{align}\label{eq:ULA-base}
    \d X_t=P_t\,\dt, \qquad \d P_t=-(\nabla V(X_0)+\gamma P_t)\,\dt+\sqrt{2\gamma}\,\dBt\,.
\end{align}
Let $\bQ\deq \bQ_z$ be the law of $\bZ=(X_t,P_t)_{t\in[0,T]}$. Let $\mu(x)=\frac{1}{\sqrt{2\gamma}}\,(\nabla V(x)-\nabla V(x_0))$. By Girsanov's theorem (under mild assumptions on $T$; see \cref{rmk:Girsanov}), there exists a probability distribution $\bP\deq \bP_z$ defined via
\begin{align}
    \frac{\d \bP}{\d \bQ}=\exp\prn[\bigg]{ -\int_{0}^T \tri{\mu(X_t),\,\dBt}-\frac12\int_0^T \nrm{\mu(X_t)}^2\,\dt  }\,,
\end{align}
such that the process $\Bbar_t=B_t+\int_0^t \mu(X_s)\,\d s$ is a Brownian motion under $\bP$, and the process $\bZ=(X_t,P_t)_{t\in[0,T]}$ satisfies
\begin{align}\label{eq:ULA-shift}
    \d X_t=P_t\,\dt, \qquad \d P_t=-(\nabla V(X_t)+\gamma P_t)\,\dt+\sqrt{2\gamma}\,\d \Bbar_t\,.
\end{align}

\scedit{According to the FORS framework, we need to design an unbiased estimator for $\log \frac{\d \bP}{\d \bQ}$. We can obtain an unbiased estimator for the second integral $\frac12\int_0^T \nrm{\mu(X_t)}^2\,\dt$ simply by evaluating the integrand at a random time in $[0,T]$, so we focus on the first term $I\deq \int_{0}^T \tri{\mu(X_t),\,\dBt}$.} Consider an integer $N\geq 1$ and let $h=\frac{T}{N}$. Note that
\begin{align}
    I\approx \sum_{j=0}^{N-1} \tri{\mu(X_{jh}), B_{(j+1)h}-B_{jh} }
    =&~\sum_{j=1}^{N} \tri{ B_{T}-B_{jh}, \mu(X_{jh})-\mu(X_{(j-1)h}) } \\
    =&~ \frac{T}{N}\,\sum_{j=1}^{N} \tri*{ B_{T}-B_{jh}, \frac{\nabla V(X_{jh})-\nabla V(X_{(j-1)h})}{\sqrt{2\gamma}\,h} }\,.
\end{align}
Then, we can consider the following estimator. In the FORS construction, take $\Xi=\Unif([0,T])$ and write its draw as $t$.
For any $t\in[(j-1)h,jh)$, we set
\begin{align}
    -W(t;\bZ)=\frac{T}{\sqrt{2\gamma}\,h}\,\tri{B_T-B_{jh}, \nabla V(X_{jh})-\nabla V(X_{(j-1)h})}+\frac{T}{2}\,\nrm{\mu(X_t)}^2\,.
\end{align}

With this estimator, FORS reweights the exponential Euler proposal using a discretized approximation to the exact ULD path law; the resulting error is analyzed below.

\subsection{Analysis}

Let $\T$ and $\That$ denote the exact and FORS-corrected transition kernels, respectively: $\T_z$ and $\That_z$ are the marginal distributions of $Z_T$ under $\bP_z$ and $\bPhat_z$.
We denote $R(z)\deq \nrm{p}^2+\beta^{-1}\,\nrm{\nabla V(x)}^2+d$ for $z=(x,p)$.
The key step of the analysis is to analyze the following R\'enyi error.

\begin{proposition}[R\'enyi errors of FORS]\label{prop:ULA-error-prop}
Suppose that $q\geq2$, $\delta\in(0,\frac12]$, $B\in[1,2]$, and $\gamma T\leq1$. Furthermore, suppose that
\begin{align}\label{eq:ULA-T-cond}
    T^3\leq \frac{\gamma}{C\beta^2 R(z)\,(q+\log(1/\delta))}\,, \qquad T^2\leq \frac{1}{C\beta\,(q+\log(1/\delta))}\,,
\end{align}
and 
\begin{align}\label{eq:ULA-N-cond}
    h=\frac{T}{N}\leq
    \frac{\delta\gamma}{Cq\beta R(z)\,\prn*{d^{-1/2}+\beta T^2}}\,.
\end{align}
Then, $\Dsys*[q]{\That_z}{\T_z}\leq \delta$.
\end{proposition}

This allows us to establish the following guarantees, which provide $\widetilde O(d^{1/3})$ query complexity bounds in the strongly convex and smooth setting.

\scedit{
\begin{theorem}[High-accuracy simulation of ULD]\label{thm:ULA-main}
Suppose that $V$ is convex and $\beta$-smooth. Let $\delta\in(0,\frac12]$ and $p\in(0,1]$. The following guarantees hold.
\begin{enumerate}[label=\textup{(\roman*)}]
    \item Suppose that $\mu_V$ satisfies PI with constant $1/\alpha$, and let $\gamma\asymp\sqrt\alpha$ and $\kappa\deq \beta/\alpha$. Set $\Delta_2\deq d+\Ren[2]{\nu}{\target}$.
    If $K\in\mathbb N$ satisfies
    \begin{align}
        K\gtrsim\frac{1}{T\gamma}\,\bigl(\Ren[2]{\nu}{\target}+\log(1/\delta)\bigr) \qquad\text{and}\qquad
        \frac{1}{T\gamma}
        \gg \kappa^{2/3}\Delta_2^{1/3}\log^{2/3}(K/\delta)\,,
    \end{align}
    then the output distribution $\nuhat$ of FORS satisfies
    \begin{align}
        \Ren*[3/2]{\nuhat}{\target}\leq\delta\,.
    \end{align}

    \item Suppose that $\mu_V$ satisfies LSI with constant $1/\alpha$, let $\gamma\asymp\sqrt\alpha$ and $\kappa\deq\beta/\alpha$, and let $q\geq2$. Set $\Delta_q\deq d+\Ren[2q]{\nu}{\target}$.
    If $K\in\mathbb N$ satisfies
    \begin{align}
        K\gtrsim\frac{1}{T\gamma}\log\Bigl(\frac{q\Delta_q}{\delta}\Bigr)\qquad\text{and}\qquad
        \frac{1}{T\gamma}
        \gg \kappa^{2/3}\Delta_q^{1/3}\,\bigl(q+\log(Kq/\delta)\bigr)^{2/3}\,,
    \end{align}
    then the output distribution $\nuhat$ of FORS satisfies
    \begin{align}
        \Ren*[q]{\nuhat}{\target}\leq\delta\,.
    \end{align}
\end{enumerate}
In both cases, with probability at least $1-p$, the number of first-order queries is bounded by $O(K+\log(1/p))$, with a universal implicit constant.
\end{theorem}
}

\screplace{\matt{We should track $q$ dependence if we are going to bound R\'enyi-$q$; otherwise we might as well say R\'enyi-$3$.}}{}

The proof is presented in \cref{app:ULA}.

The exponential Euler proposal, although simple, is fundamentally limited; see \cref{app:exponential-euler-limitation}.
This motivates the more refined proposal developed in \cref{sec:ULA2}.

\section{ULD with additional smoothness}\label{sec:ULA2}

In this section, we describe a refined proposal which better takes advantage of higher-order smoothness, allowing us to achieve Result 2.

\subsection{Proposal}

\scedit{The proposal in \cref{sec:ULA} is based on replacing the non-linear term $\nabla V(X_t)$ in~\eqref{eq:ULD} with $\nabla V(X_0)$. A higher-order approximation is obtained by replacing $\nabla V(X_t)$ with, say, $\nabla V(X_0) + \nabla^2 V(X_0)\,(X_t - X_0)$. A downside of this proposal is that implementing it requires expensive matrix operations (e.g., matrix exponentials) involving $\nabla^2 V(X_0)$. Thus, we present an alternative higher-order proposal which is inspired by Picard iteration.}

Fix an integer $L\geq 1$ and a starting point $z=(x,p)$.
Let $\bQ$ be the distribution of the following linearized SDE:
\begin{align}\label{eq:ULA2-iterative}
    &\d X^\ell_t=P^\ell_t\,\dt, \qquad \d P^\ell_t=-(\nabla V(X_0)+\nabla^2 V(X_0)(X^{\ell-1}_t-X_0)+\gamma P^\ell_t)\,\dt+\sqrt{2\gamma}\,\dBt\,,
\end{align}
with $(X^0_t,P^0_t)\equiv(x,0)$, $X^\ell_0=x$, $P^\ell_0=p$ for $\ell\geq 1$, and $(X_t,P_t)=(X^L_t,P^L_t)$. Note that \cref{eq:ULA-base} corresponds to $L=1$.

\Cref{app:refined-uld-implementation} gives an explicit implementation of the refined ULD proposal which only uses the generalized first-order oracle (\cref{def:gen_first_order}).

\paragraph{FORS instantiation}
Let
\begin{align}
    \mu_t=\frac{1}{\sqrt{2\gamma}}\,\brk*{ \nabla V(X_t)-\nabla V(X_0)-\nabla^2 V(X_0)\,(X^{L-1}_t-X_0) }\,.
\end{align}
Then, the ULD path measure $\bP$ satisfies
\begin{align}
    \frac{\d \bP}{\d \bQ}=\exp\prn[\bigg]{ -\int_{0}^T \tri{\mu_t,\,\dBt}-\frac12\int_0^T \nrm{\mu_t}^2\,\dt  }\,,
\end{align}
such that the process $\Bbar_t=B_t+\int_0^t \mu_s\,\d s$ is a Brownian motion under $\bP$.

We write the stochastic integral term as
\begin{align}
    I\deq \int_{0}^T \tri{\mu_t,\,\dBt}
    =\int_{0}^T \tri{\partial_t \mu_t, B_T-B_t}\,\dt\,,
\end{align}
where 
\begin{align}
    \partial_t \mu_t=\frac{1}{\sqrt{2\gamma}}\,\brk*{ \nabla^2 V(X_t)\,P_t-\nabla^2 V(X_0)\,P^{L-1}_t }\,.
\end{align}
Therefore, we can use the following estimator:
\begin{align}
    W(t;\bZ)\deq -T\,\tri{\partial_t \mu_t, B_T-B_t}-\frac{T}{2}\,\nrm{\mu_t}^2\,.
\end{align}
The Girsanov density is justified by an argument similar to 
\cref{rmk:Girsanov}.

\subsection{Analysis}

We introduce the following standard assumption~\citep{zhang2026algorithmic}.
\begin{assumption}[Lipschitz Hessian under Frobenius norm]\label{asmp:Lip-Hess}
$V$ is three times continuously differentiable, and there is a constant $\betaF>0$ such that, for all $x,y\in\RR^d$,
\begin{align}
    \nrmF{\nabla^2 V(x)-\nabla^2 V(y)}\leq \betaF\,\nrm{x-y}\,.
\end{align}
Equivalently, the third derivative tensor satisfies $\nrm{\nabla^3 V(x)}_{\{1,2\},\{3\}}\leq \betaF$ for any $x\in\R^d$, where
\begin{align}
    \nrm{T}_{\{1,2\},\{3\}}
    \deq \sup_{\nrmF{A}\leq 1,\,\nrm{u}\leq 1}\sum_{i,j,k=1}^d T_{ijk}A_{ij}u_k\,.
\end{align}
\end{assumption}
We define $\kappaF=\CLSI\betaF^{2/3}$. 
Similar to prior analyses such as~\citet{zhang2026algorithmic}, we first present an error bound under a $\Delta$-warm initialization.
This is due to the use of change-of-measure arguments in order to apply the refined probabilistic tail inequalities implied by~\cref{asmp:Lip-Hess}.

\begin{theorem}[Simulation of ULD under higher-order smoothness]\label{thm:ULA2-warm}
Suppose that $V$ is convex and $\beta$-smooth, that $\pi$ satisfies PI with constant $1/\alpha$ and $\gamma \asymp \sqrt\alpha$, and that \cref{asmp:Lip-Hess} holds. Suppose that initial distribution $\nu$ satisfies $\Ren[2]{\nu}{\target}\leq \Delta$, $q\geq 2$ is a constant, $K,L\in\mathbb N$, $U\geq 1$, $\delta\in(0,\frac12]$, and $\Delta\geq 1$. Define $\iota=\log(K\kappa d\Delta/\delta)$ and $\wt{d}=d+\Delta$. If $K\leq \frac{U}{\gamma T}$ and
\begin{align}\label{eq:ULA2-T-cond-warm}
    \frac{1}{T\gamma}\gg&~ \kappaF^{3/5}\,(d^{1/5}+\Delta^{2/5})\,\iota^{1/3}+\kappa^{\frac{1}{2}+\frac{1}{4L-1}}\,(\wt{d}\,\iota^3)^{\frac{1}{4L-1}} +\kappa^{\frac12}\,\iota^{1/2} \\
    &~+\kappa^{\frac{3}{11}}\kappaF^{\frac{3}{11}}\,\wt{d}^{\frac{2}{11}}\,\iota^{\frac{1}{11}}+\kappa^{\frac13}\kappaF^{\frac12}\,\wt{d}^{\frac16}\,U^{\frac13}\iota
    +\min\crl*{\kappa^{\frac35}\wt{d}^{\frac15}\iota^{\frac25},\,\kappa^{\frac17}\kappaF^{\frac37}\prn{\wt{d}\Delta}^{\frac17}\iota^{\frac27}}\,,
\end{align}
then running FORS as in \cref{thm:trunc-alg} for $K$ iterations gives a distribution $\nuhat$ with $\Ren*[q]{\nuhat}{\cPt{KT}\nu}\leq \delta$.
\end{theorem}

Note that due to the $\Delta^{2/5}$ term, the overall dimension dependence is $d^{1/5}$ only under a ``crude'' warm start $\Delta = O(\sqrt d)$.
Similarly to~\citet{zhang2026algorithmic}, we are able to show that FORS itself can produce this crude warm start (see~\cref{prop:ULA2-warm-gen}).
A two-stage analysis then yields the following query complexity bound starting from a general initial distribution with $O(d\log \kappa)$ R\'enyi divergence, which can be achieved in the strongly log-concave case with a suitable Gaussian initialization.

\begin{theorem}[Simulation of ULD from a cold start]\label{thm:ULA2-cold}
Suppose that $V$ is convex and $\beta$-smooth, that $\pi$ satisfies LSI with constant $1/\alpha$ and $\gamma \asymp \sqrt\alpha$, and that \cref{asmp:Lip-Hess} holds. Let $\Ren[w]{\nu}{\target}=O(d\log \kappa)$ for a constant $w\geq \frac{5}{2}$. Then, using
\begin{align}
    \wt{O}\Bigl(&\bigl(\kappa^{1/2}+\kappaF^{3/5}d^{1/5}
    +\kappa^{1/3}\,(\kappaF+\kappaF^{3/11})\,d^{2/11}\bigr)\log^2(1/\delta)\Bigr)
\end{align}
queries, we can generate a sample from a distribution $\nuhat$ with $\Ren*[w-1/2]{\nuhat}{\target}\leq \delta^2$.
\end{theorem}

The proof is presented in \cref{appdx:ULA2-cold}.

We remark that by using this result to implement the RGO in the proximal sampler and using the convergence under PI in~\citet{Chen+22ProxSampler}, it would yield a high-accuracy sampler in the PI setting with dimension dependence $d^{6/5}$, although we omit the details for brevity.

With this approach, it also seems possible to further improve the dimension dependence under additional smoothness, namely bounds on $\nabla^k V$ ($k\geq 3$). We leave this as an open question.

\section{Mirror Langevin}\label{sec:mirror}

Consider a mirror map $\phi: \mc X \to\RR\cup\crl{\infty}$. \sredit{Assume that $\phi$ is Legendre, so that $\nabla\phi$ and $\nabla\phi^\star$ are inverse maps on the interiors of their domains.} \screplace{Let $\phi^\star$ be the dual and $\cT=\nabla \phi^\star$.}{Let $\phi^\star$ be the dual, set $\mc X^\star\deq\operatorname{int}(\operatorname{dom}\phi^\star)$, and let $\cT=\nabla \phi^\star:\mc X^\star\to\mc X$.} We let $\norm{\cdot}_x$ denote the $\nabla^2 \phi(x)$-weighted norm at $x$, and $\norm{\cdot}_{x, *}$ be the dual norm, i.e., the $(\nabla^2 \phi(x))^{-1}$ weighted norm.

The mirror Langevin diffusion can be written as follows:
\begin{align}\label{eq:mirror}
    \d Y_t=-\nabla V(\cT(Y_t))\,\dt+\sqrt{2}\,\nabla \cT(Y_t)\isq\,\dBt\,.
\end{align}
\sredit{Throughout this section, $Y_t$ denotes the dual process, $X_t\deq\cT(Y_t)$ denotes its primal image, and $\nu_t\deq\operatorname{Law}(X_t)$. All divergences from $\pi$ are taken between primal laws.}

\subsection{Proposal}

Let $\bP$ be the measure corresponding to the process~\eqref{eq:mirror} above. As the proposal, we consider the measure $\bQ$ which describes the law of
\begin{align}
    \d Y_t=\sqrt{2}\,\nabla \cT(Y_t)\isq\,\dBt\,.
\end{align}
This proposal is the so-called mirror Brownian motion, and the assumption that it can be exactly simulated was used in~\citet{ahn2021efficient}. We will adopt this assumption in the present work.

Fix a step length $h>0$ and regard both path measures over the interval $[0,h]$.
Now, define $\mu_t = \frac{1}{\sqrt{2}}\,\nabla \cT(Y_t)^{1/2}\,\nabla V(\cT(Y_t))$. By Girsanov's theorem, we have
\begin{align}
    \frac{\d \bP}{\d \bQ}=\screplace{\exp\prn[\bigg]{ -\int_{0}^T \tri{\mu_t,\,\dBt}-\frac12\int_0^T \nrm{\mu_t}^2\,\dt  }}{\exp\prn[\bigg]{ -\int_{0}^h \tri{\mu_t,\,\dBt}-\frac12\int_0^h \nrm{\mu_t}^2\,\dt  }}\,.
\end{align}
To produce an estimator for $\int_{0}^{\screplace{T}{h}} \tri{\mu_t,\,\dBt}$, first consider the generator of the proposal:
\begin{align}
    (\msf L f)(y) \deq \tr\bigl(\nabla \cT(y)^{-1}\,\nabla^2 f(y)\bigr)\,.
\end{align}
We apply It\^o's formula to $f(Y)=V(\cT(Y))$:
\begin{align}
    \d f(Y_t)=&~ 2\inner{\mu_t, \,\d B_t} + \msf L f(Y_t) \, \d t\,.
\end{align}
Thus,
\begin{align}
    \int_0^h \inner{\mu_t, \, \d B_t} = \frac{1}{2}\,\Bigl( f(Y_h) - f(Y_0) - \int_0^h \msf L f(Y_t)\, \d t \Bigr)\,.
\end{align}
The natural estimator for the Girsanov tilt factor is to take $t \sim \operatorname{unif}[0, h]$ and then use
\begin{align}
    W(t, Y) = -\frac{1}{2}\,\bigl( f(Y_h) - f(Y_0)\bigr) + \frac{h}{2}\,\bigl(\msf L f(Y_t) - \norm{\mu_t}^2 \bigr)\,.
\end{align}
It is clearly unbiased. We will need to show that it has a bounded subexponential norm (which should scale with $d$ even when $T$ is the identity, because $\abs{\Delta V}\leq d\beta$). By the assumptions and argument of \citet{ahn2021efficient}, we should be able to show that $\abs{\tri{ \nabla^2 f(Y_t), \nabla \cT(Y_t)^{-1} }}\leq O(d)$, which in turn implies that \screplace{$\abs{f(Y_T)-f(Y_0)} \leq O(Td)$}{$\abs{f(Y_h)-f(Y_0)} \leq O(hd)$} in subexponential norm. It then suffices to choose \screplace{$T\ll \frac{1}{d}$}{$h\ll \frac{1}{d}$} to ensure the high-accuracy guarantee of FORS.

The generator term causes Hessians of $f$ to appear. If we do not want a Hessian oracle, we can apply the finite-difference machinery as with \screplace{ULA}{ULD}: we aim to simulate an ultrafine discretization of \screplace{ULA}{ULD}, replacing \screplace{ULA}{ULD} here with the standard mirror Langevin discretization.

\subsection{Analysis}

\begin{assumption}[Ahn--Chewi]\label{as:ahn-chewi}
    Suppose that $\phi$ is $M_\phi$-self-concordant: \sredit{ for all $x\in\mc X$ and $u,v,w\in\R^d$,}
    \begin{align}
        \sredit{\abs{\nabla^3\phi(x)[u,v,w]}
        \leq 2M_\phi\norm{u}_{x}\norm{v}_{x}\norm{w}_{x}.}
    \end{align}
    Furthermore, assume $V$ is convex and $\beta$-smooth relative to $\phi$, and $L$-relatively Lipschitz, \sredit{meaning in particular that
    \begin{align}
        \sup_{x\in\mc X}\norm{\nabla V(x)}_{x,*}\leq L.
    \end{align}}
\end{assumption}

We now quote the following continuous-time result from~\cite{chewi2020exponential}.
\begin{lemma}[{Adapted from~\cite[Theorem 1]{chewi2020exponential}}]\label{lem:mirror-chi2}
    If, in addition to convexity, we assume that $V$ is $\alpha$-strongly convex relative to $\phi$, \sredit{i.e., $\nabla^2V\succeq\alpha\nabla^2\phi$,} and if \sredit{$Y_t$ solves~\eqref{eq:mirror} and $\nu_t=\operatorname{Law}(\cT(Y_t))$}, then
    \begin{align}
        \screplace{\chi^2(\mu_t \, \lVert\, \pi)}{\Dchis{\nu_t}{\pi}} \leq e^{-2\alpha t} \screplace{\chi^2(\mu_0 \, \lVert\, \pi)}{\Dchis{\nu_0}{\pi}}\,.
    \end{align}
\end{lemma}
\scedit{
\begin{proof}
This is \citet[Theorem 1]{chewi2020exponential}, specialized to the relative strong convexity parameter $\alpha$. Their dissipation calculation gives
\begin{align}
    \frac{\d}{\d t}\Dchis*{\nu_t}{\pi}\leq-2\alpha\Dchis{\nu_t}{\pi}\,.
\end{align}
The claim follows from Gr\"onwall's inequality.
\end{proof}
}
As remarked in~\cite{ahn2021efficient} and subsequent works, one might hope for exponential convergence in R\'enyi divergences via an analogue of the log-Sobolev inequality. Unfortunately, such an inequality will not hold generally, even under the strong relative convexity assumptions.

An alternative convergence bound controls the \screplace{$\msf{KL}$}{$\Dklshort$} divergence in terms of the initial Bregman divergence, defined as follows.
\begin{definition}
    The $\phi$-Bregman divergence between two points is given by
    \begin{align}
        D_\phi(x \mmid y) = \phi(x) - \phi(y) - \screplace{\inner{\nabla f(y), x-y}}{\inner{\nabla \phi(y), x-y}}\,,
    \end{align}
    \screplace{\sccomment{The linear term in the $\phi$-Bregman divergence is $\inner{\nabla\phi(y),x-y}$, not $\inner{\nabla f(y),x-y}$. With $\nabla f(y)$ the object defined here is not $D_\phi$ and the identity $\nabla_yD_\phi(x\mmid y)=-\nabla^2\phi(y)(x-y)$ used in the proof need not hold.}}{}
    and the $\phi$-Bregman transport cost is
    \begin{align}
        \Breg_\phi(\mu \mmid \nu) = \inf_{\gamma \in \Gamma(\mu, \nu)} \E_{(x,y) \sim \gamma} D_\phi(x \mmid y)\,,
    \end{align}
    where $\Gamma(\mu, \nu)$ is as usual the set of admissible couplings between the two measures.
\end{definition}

We now state the following infinitesimal version of~\cite[]{ahn2021efficient}.
\begin{lemma}\label{lem:mirror-breg}
    Suppose that \screplace{$\mu_t$}{$\nu_t$} is the law of an iterate following the mirror Langevin diffusion in continuous time, that $\pi \propto \exp(-V)$ is the stationary measure in primal coordinates, and that $V$ is $\alpha$-strongly convex relative to $\phi$. Then,
    \begin{align}
        \screplace{\msf{KL}(\mu_t \mmid \pi)}{\Dkl{\nu_t}{\pi}} \leq -\partial_t \Breg_\phi(\pi \mmid \screplace{\mu_t}{\nu_t}) - \alpha \Breg_\phi(\pi \mmid \screplace{\mu_t}{\nu_t})\,,
    \end{align}
    for all $t \geq 0$. Integrating, we obtain the bounds
    \begin{align}
        \screplace{\msf{KL}(\mu_t \mmid \pi)}{\Dkl{\nu_T}{\pi}} \leq \frac{\Breg_\phi(\pi \mmid \screplace{\mu_0}{\nu_0})}{T}\,, \text{if $\alpha = 0$,} \qquad\screplace{\msf{KL}(\mu_t \mmid \pi)}{\Dkl{\nu_T}{\pi}} \leq \frac{\alpha \Breg_\phi(\pi \mmid \screplace{\mu_0}{\nu_0})}{e^{\alpha T} - 1}\,, \text{if $\alpha > 0$.}
    \end{align}
\end{lemma}
The proof is presented in \cref{appdx:proof-mirror-breg}.
Putting this together, we obtain our final result for this section.
\begin{theorem}\label{thm:mirror}
    Suppose that \cref{as:ahn-chewi} holds and that, additionally, $V$ is $\alpha$-strongly convex relative to $\phi$. \scedit{Set $\bar d\deq (\beta + 2LM_\phi)\,d+L^2$, and let $\widehat\nu$ denote the output of the FORS-corrected mirror Langevin process started at $\nu_0$.} \sredit{Assume also that the null mirror diffusion can be sampled at every finite collection of requested times and that the quantities within the estimator $W$ can be evaluated exactly. If $\alpha>0$,} then, with probability $1-\zeta$, the query complexity required to obtain \screplace{$\msf R_{3/2}(\widehat \mu \mmid \pi)$}{$\Ren*[3/2]{\widehat\nu}{\pi}$} $\leq \varepsilon^2$ is
    \begin{align}
        N_{\operatorname{query}} = \widetilde O\Bigl(\frac{\bar d}{\alpha}\,\operatorname{polylog}(\varepsilon^{-2}, \screplace{\chi^2(\mu_0 \mmid \pi)}{\Dchis{\nu_0}{\pi}}, \zeta^{-1}) \Bigr)\,.
    \end{align}

    Alternatively, if $\alpha>0$, a measure \screplace{$\widehat \mu$}{$\widehat\nu$} obtained using a procedure similar to the one above satisfies \screplace{$\msf{TV}^2(\widehat \mu \mmid \pi)$}{$\Dtv*{\widehat\nu}{\pi}^2$} $\leq \varepsilon^2$, with query complexity bounded, with probability $1-\zeta$, by
    \begin{align}
        N_{\operatorname{query}} = \widetilde O\Bigl( \frac{\bar d}{\alpha}\,\operatorname{polylog}(\varepsilon^{-2}, \Breg_\phi(\pi \mmid \screplace{\mu_0}{\nu_0}), \zeta^{-1})\Bigr)\,,
    \end{align}
    where $\bar d$ is defined as above. If $\alpha \equiv 0$, then the complexity is instead linear in the initialization:
    \begin{align}
        N_{\operatorname{query}} = \widetilde O\Bigl( \sredit{\frac{\bar d\,\Breg_\phi(\pi \mmid \nu_0)}{\varepsilon^2}}\,\operatorname{polylog}(\zeta^{-1})\Bigr)\,,
    \end{align}

    The FORS parameters---the step size, clipping parameter, and number of iterations---are given in the proof below.
\end{theorem}

The proof is presented in \cref{appdx:proof-mirror-FORS}. 
We briefly compare against the result of~\citet{SriWibWil24MAMLA}, which obtains a polylogarithmic $\varepsilon$-dependence under similar assumptions based on a simple Metropolization of a mirror Langevin discretization. They do not assume access to an exact simulator for the null potential mirror diffusion. Their dimension dependence is $d^3$, whereas we improve this to near-linear in our setting. Finally, their $\varepsilon$-dependence remains polylogarithmic in the $\alpha = 0$ setting, whereas we incur polynomial dependence.

\section{Fisher information bounds for non-log-concave sampling}\label{sec:fisher}

The complexity of \emph{non-convex} optimization is often measured in terms of the number of queries to reach an approximate stationary point. Analogously, as put forth in~\citet{Bal+22NonLogConcave}, we can view a measure $\mu$ with small $\FI{\mu}{\pi}$ as a notion of approximate first-order stationarity for sampling, and thereby develop a quantitative theory for non-log-concave sampling.

Recently, the work of~\cite{chewi2026complexity} reduced the problem of obtaining Fisher information guarantees to high-accuracy log-concave sampling.
Namely, they relied on the following assumption.

\begin{assumption}\label{as:alg-exist}
    There exists a sampling algorithm $\mathtt{Alg}$ with the following property: for any $1$-strongly log-concave and \sredit{$3$-log-smooth} distribution $\pi$ with mode at the origin and any \sredit{tolerance $\delta>0$}, the algorithm returns a sample with distribution $\widehat \pi$ such that \sredit{$\Ren*[3]{\widehat\pi}{\pi}\leq\delta^2$}.
\end{assumption}

They also assumed access to a proximal oracle for $f$.

\begin{definition}
    The proximal oracle for a function $f$ with step size $h$ outputs the minimizer:
    \begin{align}
        \operatorname{prox}_{f,h}(x) = \argmin_{y \in \mathbb R^d}{\bigl\{ f(y) + \frac{1}{2h}\norm{y-x}^2\bigr\}}\,.
    \end{align}
\end{definition}
When $f$ is $\beta$-smooth and we take $h = \frac{1}{2\beta}$, then the computation of $\operatorname{prox}_{f,h}$ is a strongly convex optimization problem, so this is not a restrictive assumption.

They established the following theorem.

\begin{theorem}[{Adapted from~\citet[Theorem 2]{chewi2026complexity}}]\label{thm:FI-reduction}
    Suppose that the target distribution $\pi \propto \exp(-V)$ is $\beta$-smooth. Suppose that we have access to an initial distribution $\mu_0$ with $\msf{KL}_0 = \Dkl{\mu_0}{\pi}$ and $\Delta_0 = \sredit{\Ren[2]{\mu_0}{\pi}}$. Under \cref{as:alg-exist}, there exists an algorithm that returns a sample from $\widehat \pi$ with \sredit{$\FI{\widehat\pi}{\pi}\leq\varepsilon^2$}, using $O(\beta\, \msf{KL}_0/\varepsilon^2)$ queries each to the proximal oracle with step size $\frac{1}{2\beta}$ and to $\mathtt{Alg}$. \sredit{Each call to $\mathtt{Alg}$ is made with a tolerance $\delta_{\mathrm{RGO}}$ satisfying
    \begin{align}
        \delta_{\mathrm{RGO}}^{-1}
        =\operatorname{poly}\prn*{1+\Delta_0,\,\beta/\varepsilon^2,\,d}.
    \end{align}
    The reduction uses the smoothed approximate RGO and returns an iterate from the proximal sampler.}
\end{theorem}

Since we have improved the complexity of high-accuracy log-concave sampling, this reduction can be combined directly with the bounds in \cref{thm:ULA-main} to give the following new state-of-the-art Fisher information guarantee for non-log-concave sampling.

\begin{theorem}\label{thm:FI-main}
Suppose that the target distribution $\pi\propto \exp(-V)$ is $\beta$-smooth. Suppose that we have access to an initial distribution $\mu_0$ with $\msf{KL}_0=\Dkl{\mu_0}{\pi}$ and $\Delta_0=\Ren[2]{\mu_0}{\pi}$. \sredit{Let $M\geq1$ denote the number of approximate RGO calls prescribed by \cref{thm:FI-reduction}, so that $M=O(1+\beta\,\msf{KL}_0/\varepsilon^2)$.} Then, there is an algorithm that returns a sample from $\widehat\pi$ with
\begin{align}
    \sredit{\FI{\widehat\pi}{\pi}}\leq \varepsilon^2
\end{align}
using
\begin{align}
    \sredit{\wt{O}\prn[\Big]{M\prn*{d^{1/3}+\log(M/p)}}}
\end{align}
first-order oracle calls with probability $1-p$ for \sredit{$p\in(0,1]$}, and $M$ calls to the proximal oracle. The $\wt{O}(\cdot)$ notation hides polylogarithmic factors in $d$, $\beta$, \sredit{$1+\Delta_0$}, $\msf{KL}_0$, and $1/\varepsilon$.
\end{theorem}
The proof is presented in \cref{appdx:FI}. If the additional Hessian Lipschitz condition in \cref{asmp:Lip-Hess} is available for the RGO subproblems, one can replace the use of \cref{thm:ULA-main} in the preceding paragraph by \cref{thm:ULA2-cold}, yielding the corresponding improved dimension dependence from \cref{thm:ULA2-cold}. \sredit{In particular, if the rescaled RGO potentials have uniformly bounded normalized parameter $\kappa_F$, the factor $d^{1/3}$ above can be replaced by $d^{1/5}$. For an original Hessian-Lipschitz constant $\beta_F$, this normalized condition is $\beta_F^{2/3}/\beta=O(1)$.}

\bibliography{ref}

\appendix

\section{Implementation}\label{app:implementation}
\subsection{Implementation of the refined ULD proposal}\label{app:refined-uld-implementation}

We write
\[
    g_0\deq \nabla V(X_0),\qquad
    H_0\deq \nabla^2V(X_0),\qquad
    U_t^\ell\deq X_t^\ell-X_0\,.
\]
The iterative proposal in~\eqref{eq:ULA2-iterative} is
\[
    \d U_t^\ell=P_t^\ell\,\d t,\qquad
    \d P_t^\ell=-(g_0+H_0U_t^{\ell-1}+\gamma P_t^\ell)\,\d t
        +\sqrt{2\gamma}\,\d B_t\,,
\]
with $U_t^0\equiv 0$, $U_0^\ell=0$, and $P_0^\ell=P_0$. We now give an
explicit formula that avoids exponentiating any matrix depending on $H_0$.

Let
\[
    a_t\deq \frac{1-e^{-\gamma t}}{\gamma}\,,
\]
and let $\cK$ denote the scalar convolution operator
\[
    (\cK f)(t)\deq \int_0^t a_{t-s}f(s)\,\d s\,.
\]
For $j\geq0$, define scalar functions
\[
    A_j(t)\deq (\cK^j a)(t),\qquad
    B_j(t)\deq \int_0^t A_j(s)\,\d s,\qquad
    C_j(t)\deq \dot A_j(t)\,.
\]
These coefficients are completely explicit, namely,
\begin{align}
    A_j(t)
    &=
    \sum_{k=1}^{j+1}
    \binom{2j+1-k}{j+1-k}
    \frac{
        (-1)^{j+1-k}+(-1)^{j+1}e^{-\gamma t}
    }{\gamma^{2j+2-k}}
    \,\frac{t^{k-1}}{(k-1)!}\,,
    \label{eq:Aj-explicit}
\end{align}
and
\begin{align}
    B_j(t)
    &=
    \sum_{k=1}^{j+2}
    (-1)^{j+2-k}
    \binom{2j+2-k}{j+2-k}
    \frac{t^{k-1}}{\gamma^{2j+3-k}(k-1)!}
    \notag\\
    &\qquad
    +
    e^{-\gamma t}
    \sum_{k=1}^{j+1}
    (-1)^{j+2}
    \binom{2j+2-k}{j+1-k}
    \frac{t^{k-1}}{\gamma^{2j+3-k}(k-1)!}\,.
    \label{eq:Bj-explicit}
\end{align}
Finally, define the Gaussian convolution variables
\begin{align}
    \Xi_j^X(t)
    &\deq
    \sqrt{2\gamma}\int_0^t A_j(t-s)\,\d B_s\,,\\
    \Xi_j^P(t)
    &\deq
    \sqrt{2\gamma}\int_0^t C_j(t-s)\,\d B_s\,.
\end{align}
Then, for every $\ell\geq1$,
\begin{align}
    X_t^\ell
    &=
    X_0+
    \sum_{j=0}^{\ell-1}
    (-H_0)^j
    \brk*{
        A_j(t)\,P_0
        -
        B_j(t)\,g_0
        +
        \Xi_j^X(t)
    }\,,
    \label{eq:Xell-explicit}\\
    P_t^\ell
    &=
    \sum_{j=0}^{\ell-1}
    (-H_0)^j
    \brk*{
        C_j(t)\,P_0
        -
        A_j(t)\,g_0
        +
        \Xi_j^P(t)
    }\,.
    \label{eq:Pell-explicit}
\end{align}
For any
finite set of times, the family
\[
    \bigl\{\Xi_j^X(t),\Xi_j^P(t):0\leq j\leq L-1\bigr\}
\]
is jointly Gaussian with covariances
\begin{align}
    \Cov\prn*{\Xi_i^X(t),\Xi_j^X(s)}
    &=
    2\gamma\int_0^{t\wedge s}
    A_i(t-r)\,A_j(s-r)\,\d r\,\Id\,,\\
    \Cov\prn*{\Xi_i^P(t),\Xi_j^P(s)}
    &=
    2\gamma\int_0^{t\wedge s}
    C_i(t-r)\,C_j(s-r)\,\d r\,\Id\,,\\
    \Cov\prn*{\Xi_i^X(t),\Xi_j^P(s)}
    &=
    2\gamma\int_0^{t\wedge s}
    A_i(t-r)\,C_j(s-r)\,\d r\,\Id\,.
\end{align}
These covariance matrices are scalar multiples of $\Id$ and depend only on the
one-dimensional kernels above, not on $H_0$. Hence, conditional sampling of any
finite collection of coordinates of the path reduces to Gaussian conditioning
for these scalar convolution variables, followed by the deterministic polynomial
map \eqref{eq:Xell-explicit}--\eqref{eq:Pell-explicit}.

\subsection{Examples of implementable mirror diffusions}\label{app:mirror-diffusion-examples}

In this section, we give examples of mirror maps $\phi$ for which we can implement the exact mirror diffusion.

\paragraph{One-dimensional examples}
For $1 \le m \le 2$, consider the family
\begin{align}
    \phi_m : \R_{>0} \to \R\,, \qquad \phi_m(x) \deq \begin{cases}
        x\log x - x\,, & m=1\,, \\
        -x^{2-m}/[(m-1)\,(2-m)]\,, & m \in (1,2)\,, \\
        -\log x\,, & m = 2\,.
    \end{cases}
\end{align}
Then, $\phi_m''(x) = x^{-m}$.
In primal coordinates, the mirror diffusion becomes
\begin{align}
    \ud X_t = m\,X_t^{m-1}\,\ud t + \sqrt 2\,X_t^{m/2}\,\ud B_t\,.
\end{align}
After the transformation \screplace{$W \deq 2\,X^{2-m}/(2-m)^2$}{$S \deq 2\,X^{2-m}/(2-m)^2$}, this becomes
\begin{align}
    \ud \screplace{W_t}{S_t} = \frac{2}{2-m}\,\ud t + 2\,\sqrt{\screplace{W_t}{S_t}}\,\ud B_t\,,
\end{align}
which is called a squared Bessel process.
Its transition densities are non-central $\chi^2$ distributions and are therefore explicit.
In the case $m=2$, we instead have
\begin{align}
    \ud X_t = 2\,X_t\,\ud t + \sqrt 2\,X_t\,\ud B_t\,,
\end{align}
which is a geometric Brownian motion with
\begin{align}
    X_t = X_0 \exp(t+\sqrt 2\,B_t)\,.
\end{align}

\paragraph{Direct sum}
If $\phi_1,\dotsc,\phi_k$ are mirror maps, then the mirror diffusion corresponding to $(x_1,\dotsc,x_k) \mapsto \sum_{i=1}^k \phi_i(x_i)$ consists of coordinatewise mirror diffusions corresponding to the maps $\phi_i$.

\paragraph{Invertible linear transformation}
If $A : \R^d\to\R^d$ is invertible and $b \in \R^d$, then the mirror
diffusion corresponding to $\phi(x)=\psi(Ax+b)$ is obtained by running
the mirror diffusion for $\psi$ in the coordinate $U=Ax+b$ (up to an
orthogonal rotation of Brownian motion) and returning
$X=A^{-1}(U-b)$.

Combined with the preceding cases, a constraint set
$\{x:Ax>b\}$, with $A$ square and invertible, is handled by taking
$\phi(x)=\psi(Ax-b)$ and returning $X=A^{-1}(U+b)$. %

\section{Technical tools}

\subsection{Bounds for underdamped Langevin dynamics}

The following lemma is standard~\citep[e.g.,][Lemma 6.2.7]{Chewi26Book}.
\begin{lemma}\label{lem:ULA-MGF}
Under the distribution $\target$, it holds that
\begin{align}
    &~ \En_{(x,p)\sim \target}\En_{\bZ\sim \bP_z} \exp\prn*{ \frac{1}{4}\nrm{P_t}^2 }\leq \exp\prn*{ d }, \quad
    \En_{(x,p)\sim \target}\exp\prn*{ \frac{1}{8\beta}\nrm{\nabla V(x)}^2 }\leq \exp\prn*{ d }.
\end{align}
In particular, for $R(z)=\nrm{p}^2+\beta^{-1}\nrm{\nabla V(x)}^2+d$, it holds that \screplace{for $p\in(0,1)$}{for $\delta\in(0,1)$},
\begin{align}
    \pi\prn*{ z: R(z)\geq 20(d+\log(1/\screplace{p}{\delta})) }\leq \screplace{p}{\delta}.
\end{align}
\end{lemma}

\scedit{
\begin{proof}
Since $\target$ is invariant for ULD, $P_t\sim\normal{0,\Id}$ under the stationary path law. Hence
\begin{align}
    \En_{z\sim\target}\En_{\bP_z}\exp\prn*{\frac14\nrm{P_t}^2}
    =2^{d/2}\leq e^d\,.
\end{align}
For the second estimate, \citet[Lemma 6.2.7]{Chewi26Book} shows that $\nabla V(X)$ is $\sqrt\beta$-sub-Gaussian when $X\sim\mu_V$. If $G\sim\normal{0,\Id}$ is independent of $X$, this sub-Gaussian estimate give
\begin{align}
    \En_{\mu_V}\exp\prn*{\frac{1}{8\beta}\nrm{\nabla V(X)}^2}
    &=\En_{\mu_V}\En_G\exp\prn*{\frac{1}{2\sqrt\beta}\tri{G,\nabla V(X)}}\\
    &\leq \En_G\exp\prn*{\frac18\nrm{G}^2}
    =\prn*{\frac43}^{d/2}\leq e^d\,.
\end{align}
Finally, $X$ and $P$ are independent under $\target$, and the estimates above imply
\begin{align}
    \En_\target\exp\prn*{\frac18\prn*{\nrm{P}^2+\beta^{-1}\nrm{\nabla V(X)}^2}}
    \leq e^{5d/4}\,.
\end{align}
If $A\deq d+\log(1/\delta)$, then $R\geq20A$ implies
$\nrm{P}^2+\beta^{-1}\nrm{\nabla V(X)}^2\geq20A-d$. Markov's inequality therefore bounds the probability of this event by
\begin{align}
    \exp\prn*{\frac54d-\frac{20A-d}{8}}
    \leq\delta\,.
\end{align}
\end{proof}
}

We note that under $z\sim \target$, we have $X_t\perp P_t$ marginally, and $P_t\sim \normal{0,\Id}$. Therefore, we may combine \cref{lem:Gaussian-chaos} with the standard change-of-measure argument to prove the following lemma.

\begin{lemma}\label{lem:Gaussian-chaos-chis}
There is an absolute constant $C>0$ such that, for every $t\geq0$, every $\nu\ll\target$ with $\Ren[1.5]{\nu}{\target}<\infty$, and every measurable $F:\RR^d\to\RR^{d\times d}$,
\begin{align}
    &~\En_{z\sim \nu} \En_{\bP_z}\exp\prn*{ \frac{1}{C\betaF}\nrm{\nabla^3 V(X_t)[P_t,P_t]} }
    \leq \exp\prn*{C\sqrt d+\frac13\Ren[1.5]{\nu}{\target}},\\
    &~\En_{z\sim \nu} \En_{\bP_z}\exp\prn*{ \frac{1}{C\betaF\beta}\nrm{\nabla^3 V(X_t)[P_t,\nabla^2 V(X_t)P_t]} }
    \leq \exp\prn*{C\sqrt d+\frac13\Ren[1.5]{\nu}{\target}},\\
    &~\En_{z\sim \nu} \En_{\bP_z}\exp\prn*{ \frac{1}{C\nrmF{F(X_t)}^2}\nrm{F(X_t)P_t}^2 }
    \leq \exp\prn*{C+\frac13\Ren[1.5]{\nu}{\target}}.
\end{align}
In the last display, the quotient is defined to be zero when $\nrmF{F(X_t)}=0$.
\end{lemma}

\scedit{
\begin{proof}
Let $\nu_t\deq\cPt{t}\nu$ and $r_t\deq\d\nu_t/\d\target$. By stationarity and data processing,
\begin{align}
    \Ren*[3/2]{\nu_t}{\target}\leq\Ren[3/2]{\nu}{\target}\,.
\end{align}
Consequently, for every non-negative measurable $H$, H\"older's inequality gives
\begin{align}\label{eq:chaos-change-measure}
    \En_{\nu_t}e^H
    \leq \prn*{\En_{\target}r_t^{3/2}}^{2/3}
        \prn*{\En_{\target}e^{3H}}^{1/3}
    \leq \exp\prn*{\frac13\Ren[3/2]{\nu}{\target}}
        \prn*{\En_{\target}e^{3H}}^{1/3}\,.
\end{align}
Under $\target$, $X_t$ and $P_t$ are independent and $P_t\sim\normal{0,\Id}$. Conditional on $X_t=x$, \cref{lem:Gaussian-chaos} with $A=\Id$ gives
\begin{align}
    \bbP\prn*{\nrm{\nabla^3V(x)[P_t,P_t]}
        \gtrsim\betaF\prn*{\sqrt d+u}}\leq e^{-u}\,,
    \qquad u\geq1\,.
\end{align}
The same result with $A=\nabla^2V(x)$, together with
$\nrmop{\nabla^2V(x)}\leq\beta$, gives the analogous bound with scale
$\betaF\beta$. Integrating these two tail bounds shows that the corresponding exponential moments under $\target$ are at most $\exp(C\sqrt d)$ after enlarging by a universal constant $C$. Finally, conditional on $X_t=x$, the standard Gaussian quadratic-form identity gives
\begin{align}
    \En\exp\prn*{\frac{\nrm{F(x)P_t}^2}{C\nrmF{F(x)}^2}}\leq e^C
\end{align}
for a universal $C$, with the stated convention when $F(x)=0$. Applying \eqref{eq:chaos-change-measure} in the three cases proves the claim.
\end{proof}
}

\subsection{Concentration inequalities}

\begin{lemma}\label{lem:max-B}
Let $(B_t)_{t\geq0}$ be standard $d$-dimensional Brownian motion. For every $T\geq0$ and $\delta\in(0,1)$, with probability at least $1-\delta$,
\begin{align}
    \max_{0\leq t\leq T}\nrm{B_t}^2\leqsim T(d+\log(1/\delta)).
\end{align}
\end{lemma}

\scedit{
\begin{proof}
This is the estimate in \citet[Lemma 6.2.4]{Chewi26Book}. Indeed, that result with $\lambda=1/(4T)$ gives
\begin{align}
    \En\exp\prn*{\frac{1}{4T}\max_{0\leq t\leq T}\nrm{B_t}^2}
    \leq 3^d\,.
\end{align}
Markov's inequality proves the result. The case $T=0$ is immediate.
\end{proof}
}

The following result is a special case of \citet[Theorem 3.5]{pinelis2012optimum}.
\begin{lemma}\label{lem:Ito-Freedman}
Fix $T>0$. Let $(B_t)_{t\geq0}$ be standard $d$-dimensional Brownian motion, and let $(J_t)_{t\geq0}$ be an $\RR^{d\times d}$-valued process that is predictable with respect to the filtration of $(B_t)_{t\geq0}$ and satisfies $\int_0^T\nrmF{J_s}^2\,\screplace{ds}{\d s}<\infty$ almost surely. Denote $Z_t\deq\int_0^tJ_s\,\screplace{dB_s}{\d B_s}$ and $M_t\deq\int_0^t\nrmF{J_s}^2\,\screplace{ds}{\d s}$. Then, for every $L>0$ and $\delta\in(0,1)$,
\begin{align}
    \bbP\prn*{ \max_{0\leq t\leq T}\nrm{Z_t}^2\geq 2L\log(2/\delta),\ M_T\leq L }\leq \delta\,.
\end{align}
\end{lemma}

\scedit{
\begin{proof}
Stop the Hilbert-valued continuous martingale $Z$ when $M_t$ first exceeds $L$. The stopped martingale has quadratic characteristic at most $L$, so \citet[Theorem 3.5]{pinelis2012optimum} gives
\begin{align}
    \bbP\prn*{\max_{0\leq t\leq T}\nrm{Z_t}\geq u,\ M_T\leq L}
    \leq 2\exp\prn*{-\frac{u^2}{2L}}\,.
\end{align}
Taking $u=\sqrt{2L\log(2/\delta)}$ proves the claim.
\end{proof}
}

\begin{lemma}[Hilbert-valued Gaussian Hanson--Wright inequality]
\label{lem:hilbert-HW}
Let $H$ be a Hilbert space and let $(a_{ij})_{i,j=1}^d$ be a symmetric matrix with
entries in $H$. Let $g \sim \mathcal N(0,\Id)$ and define
\[
    Z:=\left\|\sum_{i,j=1}^d a_{ij}(g_i g_j-\mathbf 1_{i=j})\right\|_H.
\]
Then, for all $\delta\in(0,1/2)$, with probability at least $1-\delta$,
\[
    Z \lesssim \mathbb E Z + R
        + U \sqrt{\log\frac1\delta}
        + V \log\frac1\delta ,
\]
where
\[
    R:=\mathbb E\left\|\sum_{i\ne j} a_{ij}G_{ij}\right\|_H ,
\]
with $(G_{ij})_{i\ne j}$ i.i.d. $\mathcal N(0,1)$, and
\[
    U
    :=
    \sup_{\|x\|\le 1}
    \mathbb E\left\|\sum_{i,j=1}^d a_{ij}x_i g_j\right\|_H
    +
    \sup_{\|B\|_F\le 1}
    \left\|\sum_{i,j=1}^d a_{ij}B_{ij}\right\|_H ,
\]
\[
    V
    :=
    \sup_{\|x\|\le 1,\ \|y\|\le 1}
    \left\|\sum_{i,j=1}^d a_{ij}x_i y_j\right\|_H.
\]
\end{lemma}

\begin{proof}
See \citet[Theorem 6 and Remark 9]{adamczak2020hanson}.
\end{proof}

The following lemma generalizes \citet[Lemma 3.2]{zhang2026algorithmic}.

\begin{lemma}[Gaussian chaos bound]\label{lem:Gaussian-chaos}
Let $\cT\in(\RR^d)^{\otimes 3}$ be a symmetric $3$-tensor and $A\in\RR^{d\times d}$ be a symmetric matrix. Assume $\nrm{\cT}_{\{1,2\},\{3\}}\leq \bar\beta$ and $\nrmop{A}\leq L$. Then for $\xi\sim\normal{0,\Id}$ and all $\delta\in(0,1/2)$, with probability at least $1-\delta$,
\begin{align}
    \nrm{\cT[\xi,A\xi]}
    \leqsim \bar\beta L\prn*{\sqrt d+\log\frac1\delta}.
\end{align}
More generally,
\begin{align}\label{eq:tensor-chaos-general-A}
    \nrm{\cT[\xi,A\xi]}
    \leqsim \bar\beta\prn*{\nrmF{A}+\nrmop{A}\log\frac1\delta}.
\end{align}
\end{lemma}

\begin{proof}
For $j,\ell\in[d]$, define the vector
\begin{align}
    d_{j\ell}\deq \cT[\cdot,e_j,Ae_\ell]\in\RR^d,
    \qquad
    c_{j\ell}\deq \frac12(d_{j\ell}+d_{\ell j}).
\end{align}
Since $\xi_j\xi_\ell$ is symmetric in $(j,\ell)$,
\begin{align}
    \cT[\xi,A\xi]=\sum_{j,\ell=1}^d c_{j\ell}\xi_j\xi_\ell.
\end{align}
Let $\bar c\deq \sum_{j=1}^d c_{jj}=\En\cT[\xi,A\xi]$. For any unit vector $u$,
\begin{align}
    \tri{u,\bar c}=\tr(\cT[u,\cdot,\cdot]A),
\end{align}
and hence $\nrm{\bar c}\leq \bar\beta\nrmF{A}$. It remains to control the centered chaos
\begin{align}
    Z\deq \nrm*{\sum_{j,\ell=1}^d c_{j\ell}(\xi_j\xi_\ell-\indic\crl{j=\ell})}.
\end{align}
We apply \cref{lem:hilbert-HW} with $H=\RR^d$ and coefficients $a_{j\ell}=c_{j\ell}$.

First, by symmetrization and the symmetry of $\cT$,
\begin{align}
    \sum_{j,\ell=1}^d \nrm{c_{j\ell}}^2
    \leq \sum_{j,\ell=1}^d \nrm{d_{j\ell}}^2
    =\sum_{\ell=1}^d \nrmF{\cT[Ae_\ell,\cdot,\cdot]}^2
    \leq \bar\beta^2\nrmF{A}^2.
\end{align}
Therefore $\En Z\leq (\En Z^2)^{1/2}\lesssim \bar\beta\nrmF{A}$ and $R\lesssim \bar\beta\nrmF{A}$ in the notation of \cref{lem:hilbert-HW}.

We next bound the two terms entering $U$. For the deterministic part, for any matrix $B$ with $\nrmF{B}\leq1$ and any unit vector $u$,
\begin{align}
    \tri*{u,\sum_{j,\ell}d_{j\ell}B_{j\ell}}
    =\tri*{\cT[u,\cdot,\cdot]A,B}_F
    \leq \bar\beta\nrmop{A}.
\end{align}
The same bound holds with $d_{j\ell}$ replaced by $c_{j\ell}$. For the Gaussian-linear part, for any $\nrm{x}\leq1$,
\begin{align}
    \En\nrm*{\sum_{j,\ell}c_{j\ell}x_j\xi_\ell}
    &\leq \prn*{\sum_{\ell=1}^d\nrm*{\sum_{j=1}^d c_{j\ell}x_j}^2}^{1/2}  \\
    &\leq \frac12\prn*{\nrmF{\cT[x,\cdot,\cdot]A}+\nrmF{\cT[Ax,\cdot,\cdot]}}
    \leq \bar\beta\nrmop{A}.
\end{align}
Thus $U\lesssim\bar\beta\nrmop{A}$.
Finally, for unit vectors $x,y$,
\begin{align}
    \nrm*{\sum_{j,\ell}c_{j\ell}x_jy_\ell}
    \leq \frac12\nrm{\cT[x,Ay]}+\frac12\nrm{\cT[y,Ax]}
    \leq \bar\beta\nrmop{A},
\end{align}
so $V\leq \bar\beta\nrmop{A}$.

Plugging these estimates into \cref{lem:hilbert-HW} gives
\begin{align}
    Z\leqsim \bar\beta\nrmF{A}
    +\bar\beta\nrmop{A}\sqrt{\log\frac1\delta}
    +\bar\beta\nrmop{A}\log\frac1\delta.
\end{align}
Since $\delta<1/2$, the middle term is absorbed into the last one. Adding the mean bound proves \cref{eq:tensor-chaos-general-A}, and the displayed corollary follows from $\nrmF{A}\leq \sqrt d\,\nrmop{A}\leq \sqrt d L$.
\end{proof}

\subsection{Information-theoretic inequalities}

\begin{lemma}\label{lem:Renyi-triangle}
Let $\nu\ll\mu\ll\pi$ and $\eps\in(0,1]$. For every $\lambda\geq1$,
\begin{align}
    \Ren[1+\lambda]{\nu}{\pi}
    \leq \frac{(1+\eps)(1+\lambda)}{1+(1+\eps)\lambda}
    \Ren[1+(1+\eps)\lambda]{\nu}{\mu}+\Ren[4\lambda/\eps]{\mu}{\pi}.
\end{align}
For every $\lambda>1$,
\begin{align}
    \Ren[\lambda]{\nu}{\pi}
    \leq \prn*{1+\frac{\eps}{\lambda-1}}\Ren[(1+\eps)\lambda]{\nu}{\mu}+\Ren[2\lambda/\eps]{\mu}{\pi}.
\end{align}
\end{lemma}

\scedit{
\begin{proof}
These are all instantiations of the R\'enyi weak triangle inequality \citep[see, e.g.,][Lemma 6.2.5]{Chewi26Book}.
\end{proof}
}

\begin{lemma}\label{lem:KL-triangle}
For probability distributions $\nu,\mu,\muhat$,
\begin{align}
    \Dkl*{\nu}{\muhat}\leq 2\Dkl{\nu}{\mu}+\log\prn*{1+\Dchis*{\mu}{\muhat}}\,.
\end{align}
\end{lemma}

\scedit{
\begin{proof}
If the right-hand side is finite, then $\nu\ll\mu\ll\muhat$; otherwise there is nothing to prove. Apply the Donsker--Varadhan variational principle, with reference measure $\mu$ and test function $\log(\d\mu/\d\muhat)$. It yields
\begin{align}
    \Dkl*{\nu}{\muhat}
    &=\Dkl{\nu}{\mu}+\En_\nu\log\frac{\d\mu}{\d\muhat}\\
    &\leq 2\Dkl{\nu}{\mu}
        +\log\En_\mu\frac{\d\mu}{\d\muhat}
    =2\Dkl{\nu}{\mu}+\log\prn*{1+\Dchis*{\mu}{\muhat}}\,.
\end{align}
\end{proof}
}

\begin{lemma}\label{lem:Renyi-to-diff}
Let $(\Omega,\mathcal F,\rho)$ be a $\sigma$-finite measure space, and let
$f,g:\Omega\to\RR$ be measurable functions such that
\begin{align}
    0<Z_f\deq\int e^{-f}\,\d\rho<\infty,
    \qquad
    0<Z_g\deq\int e^{-g}\,\d\rho<\infty.
\end{align}
Define
\begin{align}
    \d\mu_f\deq Z_f^{-1}e^{-f}\,\d\rho,
    \qquad
    \d\mu_g\deq Z_g^{-1}e^{-g}\,\d\rho.
\end{align}
Then, for every $q>1$,
\begin{align}
    \Dsys[q]{\mu_f}{\mu_g}
    \leq \En_{\mu_f}\brk{e^{2q\,\abs{f-g}}-1}\,.
\end{align}
\end{lemma}

\begin{proof}
Set $h\deq f-g$, $A\deq\En_{\mu_f}e^h=Z_g/Z_f$, and
$M\deq\En_{\mu_f}e^{2q|h|}$. The result is immediate if $M=\infty$.
Otherwise, H\"older's inequality and
$(\En_{\mu_f}e^h)(\En_{\mu_f}e^{-h})\geq1$ give
\begin{align}
    A\leq M^{1/(2q)},
    \qquad A^{-1}\leq\En_{\mu_f}e^{-h}\leq M^{1/(2q)}.
\end{align}
Since $\d\mu_g/\d\mu_f=e^h/A$, we have
\begin{align}
    1+\Dren[q]{\mu_g}{\mu_f}
    &=A^{-q}\En_{\mu_f}e^{qh}\leq M,\\
    1+\Dren[q]{\mu_f}{\mu_g}
    &=A^{q-1}\En_{\mu_f}e^{-(q-1)h}
      \leq M^{(q-1)/q}\leq M.
\end{align}
Taking the maximum proves the claim.
\end{proof}

\begin{lemma}\label{lem:Renyi-kernel}
Suppose that $\T$ and $\That$ are transition kernels on $\cZ$, and
$\nu,\target$ are probability distributions such that $\nu\ll\target$ and
$\That_z\ll\T_z$ for $\nu$-almost every $z$. Then, for $\lambda>0$ and
$\eps\in(0,\frac12]$, it holds that
\begin{align}
    \Ren*[1+\lambda]{\That\nu}{\T\target}
    \leq \Ren[1+(1+\eps)\lambda]{\nu}{\target}+\frac{\eps}{\lambda}\log\En_{z\sim \nu}\prn*{1+\Dren*[1+\lambda]{\That_z}{\T_z}}^{1/\eps+1}.
\end{align}
\end{lemma}

\scedit{
\begin{proof}
Let $r\deq\d\nu/\d\target$ and
\begin{align}
    a(z)\deq1+\Dren*[1+\lambda]{\That_z}{\T_z}\,.
\end{align}
Consider the joint laws $\nu(\d z)\That_z(\d z')$ and
$\target(\d z)\T_z(\d z')$. Data processing under the projection $(z,z')\mapsto z'$ and then H\"older's inequality with conjugate exponents $1+\eps$ and $(1+\eps)/\eps$ give
\begin{align}
    \Ren*[1+\lambda]{\That\nu}{\T\target}
    &\leq\frac1\lambda\log\En_{z\sim\nu}\brk{r(z)^\lambda a(z)}\\
    &\leq\frac{1}{\lambda(1+\eps)}
        \log\En_{z\sim\nu}r(z)^{\lambda(1+\eps)}
        +\frac{\eps}{\lambda(1+\eps)}
        \log\En_{z\sim\nu}a(z)^{1+1/\eps}\\
    &\leq\Ren*[1+(1+\eps)\lambda]{\nu}{\target}
        +\frac\eps\lambda\log\En_{z\sim\nu}a(z)^{1+1/\eps}\,.
\end{align}
This is the desired estimate.
\end{proof}
}

\subsection{Truncated algorithm}

\begin{theorem}\label{thm:trunc-alg}
Let $q\geq 2$, $0<\delta\leq \frac{\Delta}{2}$, $0<p\leq \min\crl*{1,\frac{\delta}{q}}$, \scedit{and $K\in\mathbb N$.} Suppose that $\T$ and $\That$ are transition kernels on $\cZ$, and $\target$ is stationary under $\T$. Let $\nu$ be an initial distribution over $\cZ$ with $\Ren[2]{\nu}{\target}\leq \Delta$, and let $\cE$ be an event such that $\That_z\ll\T_z$ for $\target$-almost every $z\in\cE$ and
\begin{align}
    1-\target(\cE)\leq \screplace{\prn*{\frac{p}{8K}}^3e^{-\Delta}}{\prn*{\frac{p}{64K}}^3e^{-2\Delta}}\,.
\end{align}
Further, we assume that for any $\nu'$ supported on $\cE$ with $\Ren[1.5]{\nu'}{\target}\leq 2\Delta$,
\begin{align}
    \En_{z\sim \nu'}\prn*{1+\Dren*[q]{\That_z}{\T_z}}^{5K}\leq e^\delta.
\end{align}

We consider the following algorithm:
\begin{itemize}
    \item Set $\nu_0=\nu(\cdot|\cE)$ and sample $Z^0\sim \nu_0$.
    \item Sample $\Zbar^k\sim \That_{Z^{k-1}}$. If $\Zbar^k\in \cE$, set $Z^k=\Zbar^k$; Otherwise, declare failure and set $Z^\ell=\perp$ for $\ell\geq k$.
\end{itemize}
Let $\pihat$ be the distribution of $Z^K$ conditioned on the algorithm succeeding.
Then it holds that $\Dren*[q/2]{\pihat}{\T^K\nu}\leq e^{\delta}-1$. Further, the failure probability of the algorithm is bounded by $p$.
\end{theorem}

\begin{proof}
Let $\nubar_k$ be the distribution of $\Zbar^k$ conditional on
$Z^{k-1}\neq\perp$, and let $\nuhat_k$ be the distribution of $Z^k$
conditional on $Z^k\neq\perp$. Thus $\nuhat_0=\nu(\cdot\mid\cE)$,
$\nubar_0=\nu$, and $\nuhat_k=\nubar_k(\cdot\mid\cE)$.
For every $r>1$ and probability measure $\varpi$,
\begin{align}\label{eq:conditioning-cost}
    \Ren*[r]{\nubar_k(\cdot\mid\cE)}{\varpi}
    \leq \Ren*[r]{\nubar_k}{\varpi}
    +
    \frac{r}{r-1}\log\frac{1}{\nubar_k(\cE)}.
\end{align}

If $\Ren[3/2]{\rho}{\target}\leq2\Delta$, H\"older's inequality and the
assumption on $\target(\cE^c)$ give
\begin{align}\label{eq:trunc-tail}
    \rho(\cE^c)
    \leq e^{2\Delta/3}\target(\cE^c)^{1/3}
    \leq \frac{p}{64K}.
\end{align}
The same conclusion holds for $\rho=\nu$, using
$\Ren[2]{\nu}{\target}\leq\Delta$. Put
\begin{align}
    a\deq1+\frac{1}{4K},\qquad
    \ell_k\deq a^{-k}.
\end{align}
Since $a^K\leq e^{1/4}<4/3$, we have $\ell_k>3/4$. In particular,
\eqref{eq:conditioning-cost} and \eqref{eq:trunc-tail} show that the cost of conditioning in the every instance below is at most $p/(20K)$.

We first prove by induction that
\begin{align}\label{eq:trunc-warm-invariant}
    \Ren*[1+\ell_k]{\nuhat_k}{\target}
    \leq \Delta+\frac{k+1}{2K}(p+\delta)<2\Delta.
\end{align}
The case $k=0$ follows from \eqref{eq:conditioning-cost}. If the claim
holds at $k$, the transition kernel inequality above, monotonicity of
the divergence in its order, and $4K+1\leq5K$ imply
\begin{align}
    \Ren*[1+\ell_{k+1}]{\nubar_{k+1}}{\target}
    &\leq \Ren*[1+\ell_k]{\nuhat_k}{\target}
      +\frac{\delta}{4K\ell_{k+1}}\\
    &\leq \Ren*[1+\ell_k]{\nuhat_k}{\target}
      +\frac{\delta}{3K}.
\end{align}
This bound is below $2\Delta$, so \eqref{eq:trunc-tail} applies to
$\nubar_{k+1}$. Adding the conditioning cost proves
\eqref{eq:trunc-warm-invariant}.

For the accuracy estimate, define
\begin{align}
    q_k\deq1+(q-1)a^{-k},\qquad
    A_k\deq\Ren*[q_k]{\nuhat_k}{\T^k\nu}.
\end{align}
Again $q_k-1>3(q-1)/4$, and
\begin{align}
    A_0=\log\frac{1}{\nu(\cE)}\leq\frac{p}{60K}.
\end{align}
The same transition calculation, now at order $q_{k+1}$, followed by
\eqref{eq:conditioning-cost}, gives
\begin{align}
    A_{k+1}
    \leq A_k+\frac{\delta}{3K(q-1)}+\frac{p}{20K}.
\end{align}
Consequently,
\begin{align}
    (q-1)A_K
    \leq \frac{p(q-1)}{60K}+\frac{\delta}{3}
      +\frac{p(q-1)}{20}<\delta,
\end{align}
where we used $p(q-1)\leq\delta$. Since $q_K\geq q/2$, monotonicity
of R\'enyi divergence yields, for $q>2$,
\begin{align}
    \Dren*[q/2]{\nuhat_K}{\T^K\nu}
    \leq \exp\prn*{\frac{q/2-1}{q-1}\delta}-1
    \leq e^\delta-1.
\end{align}
For $q=2$, the conclusion is immediate from the convention
$\Dren[1]{\rho}{\mu}=0$. Finally, \eqref{eq:trunc-tail} and a union
bound give
\begin{align}
    \PP(\mathrm{fail})\leq
    \sum_{k=1}^K\nubar_k(\cE^c)\leq\frac{p}{64}\leq p.
\end{align}
\end{proof}

\section{Proofs from \cref{sec:ULA}}\label{app:ULA}

To apply \cref{lem:FORS-error}, we define
\begin{align}
    \wstar(\bZ)\deq \log \frac{\d \bP}{\d \bQ}=-\int_{0}^T \tri{\mu(X_t), \dBt}-\frac12\int_0^T \nrm{\mu(X_t)}^2\dt,
\end{align}
and recall that we define $w(\bZ)\deq \En_{t\sim \unif([0,T])} W(t;\bZ)$.

We state the following lemmas.

\begin{lemma}\label{lem:ULA-bias}
There are universal constants $c,C>0$ such that, if
$T\leq c/\sqrt{\beta}$ and $\gamma T\leq 1$, then the following holds.
For
\begin{align}
    1\leq\lambda\leq c
    \min\crl*{\frac{1}{h\beta T\sqrt d},\frac{1}{h\beta^2T^3}},
\end{align}
it holds that
\begin{align}
    \En_{\bZ\sim \bP_z} \exp\prn*{ \lambda \abs{\wstar(\bZ)-w(\bZ)}}
    \leq \exp\prn*{
    \frac{C\lambda h\beta}{\gamma}
    \prn*{d^{-1/2}+\beta T^2}R(z)}.
\end{align}
\end{lemma}
Proof in \cref{appdx:proof-ULA-bias}.

\begin{lemma}\label{lem:ULA-clip}
There are universal constants $c,C>0$ such that, if
$T\leq c/\sqrt\beta$ and $\gamma T\leq 1$, then for any
$1\leq \lambda\leq c/(\beta T^2)$ and any $t\in[0,T)$,
\begin{align}
    \En_{\bZ\sim \bP_z} \exp\prn*{ \lambda \abs{W(t;\bZ)}}
    \leq 2\exp\prn*{ \frac{C\lambda^2\beta^2T^3}{\gamma}R(z) }.
\end{align}
\end{lemma}
The proof is contained in \cref{appdx:proof-ULA-clip}.

\scedit{
\begin{lemma}[Finite-time simulation of ULD]\label{lem:ULA-finite-time}
Suppose that $V$ is convex and $\beta$-smooth. Let $\alpha>0$, $\gamma\asymp\sqrt\alpha$, $\kappa\deq\max\{\beta/\alpha,1\}$, $q>1$, $\delta\in(0,\frac12]$, and $K\in\mathbb N$. Set
\begin{align}
    \Delta\deq d+\Ren[2]{\nu}{\target}\,,\qquad
    \vartheta_q\deq\delta\min\crl*{1,q-1}\,,\qquad
    \Lambda_q\deq q+\log\Biggl(\frac{Kq}{\vartheta_q}\Biggr)\,.
\end{align}
If
\begin{align}
    \frac{1}{T\gamma}\gg
    \kappa^{2/3}\Delta^{1/3}\Lambda_q^{2/3}\,,
\end{align}
then there is an algorithm with output distribution $\nuhat$ such that
\begin{align}
    \Ren*[q]{\nuhat}{\cPt{KT}\nu}\leq\delta\,.
\end{align}
Moreover, for every $p\in(0,1]$, with probability at least $1-p$, the number of first-order queries is bounded by $O(K+\log(1/p))$, with a universal implicit constant.
\end{lemma}

\begin{lemma}[Composition of R\'enyi errors]\label{lem:Renyi-composition-ULA}
Let $q>1$ and $r>q$, and set
\begin{align}
    s\deq 1+\frac{(q-1)r}{r-q}\,.
\end{align}
For probability distributions $\rho$, $\mu$, and $\pi$,
\begin{align}
    \Ren[q]{\rho}{\pi}
    \leq \frac{q(r-1)}{r(q-1)}\,\Ren[r]{\rho}{\mu}
    +\Ren[s]{\mu}{\pi}\,.
\end{align}
\end{lemma}

\begin{proof}
This is another special case of the R\'enyi weak triangle inequality.
\end{proof}
}

\subsection{Proof of \cref{prop:ULA-error-prop} and \scedit{\cref{lem:ULA-finite-time,thm:ULA-main}}}

\begin{proof}[\pfref{prop:ULA-error-prop}]
By \cref{lem:FORS-error}, it holds that
\begin{align}
    2\log\prn*{1+\Dsys*[q]{\That_z}{\T_z}}\leq \log\En_{\bP_z}\exp\prn*{ 4q\abs{\wstar(\bZ)-w(\bZ)} }+\log \En_{\bP_z,t\sim \Unif([0,T])}\exp\prn*{ 4q\trunc[B]{W(t;\bZ)} }.
\end{align} 
By \cref{lem:ULA-bias}, as long as
\begin{align}
    h\leq \frac{\delta\gamma}{C_0q\beta R(z)\prn*{d^{-1/2}+\beta T^2}},
\end{align}
it holds that
\begin{align}
    \log\En_{\bP_z}\exp\prn*{ 4q\abs{\wstar(\bZ)-w(\bZ)} }\leq \frac{\delta}{2}.
\end{align}
Fix any $\ell\geq 1$. By \cref{lem:ULA-clip}, it holds that as long as $q\ell \leq \frac{1}{C_1\beta T^2}$,
\begin{align}
    \En_{\bP_z}\exp\prn*{ 4q\trunc[B]{W(t;\bZ)} }
    \leq&~ 1+e^{-4q\ell B}\En_{\bP_z}\exp\prn*{ 4q\ell\abs{W(t;\bZ)}} \\
    \leq&~1+2\exp\prn*{-4q\ell B+\frac{C_1q^2\ell^2\beta^2T^3}{\gamma}R(z)}.
\end{align}
Then, we can choose
\begin{align}
    \ell=\frac{1}{C_1}\min\crl*{ \frac{1}{C_1q\beta T^2}, \frac{\gamma B}{C_1q\beta^2T^3R(z)}}, 
\end{align}
so that as long as $\ell\geq 1+\frac{\log(4/\delta)}{qB}$ holds,
it holds that
\begin{align}
    \En_{\bP_z}\exp\prn*{ 4q\trunc[B]{W(t;\bZ)} }\leq 1+\frac{\delta}{2}.
\end{align}
Therefore, under \cref{eq:ULA-T-cond} and \cref{eq:ULA-N-cond}, we have
\begin{align}
    \Dsys*[q]{\That_z}{\T_z}\le e^{\delta/2}-1\leq \delta.
\end{align}
\end{proof}

\begin{proof}[\screplace{\pfref{thm:ULA-main}}{\pfref{lem:ULA-finite-time}}]
\screplace{Fix any $q\geq 2$ and $\delta\in(0,\frac{1}{2}]$.}{Fix $q>1$ and $\delta\in(0,\frac{1}{2}]$, and recall $\vartheta_q=\delta\min\crl*{1,q-1}$. Set $p_0\deq\vartheta_q/(4q)$.}
Fix the FORS threshold $B\in[1,2]$ throughout the construction.
We define
\begin{align}
    \cE=\crl*{z=(x,p): \nrm{p}^2+\beta^{-1}\nrm{\nabla V(x)}^2\leq C_0(\Delta+\screplace{\log(K/\delta)}{\log(Kq/\vartheta_q)})},
\end{align}
\screplace{where $C_0>0$ is a sufficiently large constant so that $1-\target(\cE)\leq \prn*{\frac{\delta}{8K}}^3e^{-\Delta}$.}{Increasing $C_0$ if necessary, we have
\begin{align}
    1-\target(\cE)\leq\prn*{\frac{p_0}{64K}}^3e^{-2\Delta}\,.
\end{align}}

\screplace{In the following, we define $\Rbar=(C_0+1)(\Delta+\log(K/\delta))$, so that $R(z)\leq \Rbar$ for $z\in\cE$.}{In the following, we define $\Rbar=(C_0+1)(\Delta+\log(Kq/\vartheta_q))$, so that $R(z)\leq \Rbar$ for $z\in\cE$.}
Further, by \cref{prop:ULA-error-prop}, as long as
\begin{align}\label{eq:ULA-T-cond-spec}
    T^3\leq \frac{\gamma}{C_1\beta^2 \Rbar(q+\log(Kq/\vartheta_q))}, \qquad T^2\leq \frac{1}{C\beta (q+\log(Kq/\vartheta_q))}, \qquad
    h=\frac{T}{N}\leq
    \frac{\vartheta_q\gamma}{CKq\beta \Rbar\prn*{d^{-1/2}+\beta T^2}},
\end{align}
\screplace{\sccomment{Here the target one-step error is $\delta/(4K)$ and the R\'enyi order used below is $2q$. Substituting these values into \eqref{eq:ULA-N-cond} requires $h\lesssim \delta\gamma\sqrt d/(Kq\beta\Rbar)$. The displayed condition instead has $K/\delta$, so it does not imply the one-step estimate claimed in the next sentence.}}{}
\screplace{we can guarantee that $\Dsys*[2q]{\That_z}{\T_z}\leq \frac{\delta}{4K}$ for any $z\in\cE$. Applying \cref{thm:trunc-alg} then gives the desired upper bound.}{we can guarantee that
\begin{align}
    \Dsys*[2q]{\That_z}{\T_z}\leq\frac{\vartheta_q}{10K}
\end{align}
for any $z\in\cE$. To see that the hypothesis of
\cref{lem:ULA-finite-time} implies \eqref{eq:ULA-T-cond-spec}, use
$\gamma^2\asymp\alpha$, $\beta/\alpha\leq\kappa$, and
$\Rbar\lesssim\Delta\Lambda_q$. The first detailed condition is
equivalent, up to constants, to
$(T\gamma)^{-3}\gtrsim\kappa^2\Rbar\Lambda_q$, while the second follows
from $(T\gamma)^{-2}\gtrsim\kappa\Lambda_q$; both follow from the
displayed $\kappa^{2/3}\Delta^{1/3}\Lambda_q^{2/3}$ bound because
$\kappa,\Delta,\Lambda_q\geq1$. Apply \cref{thm:trunc-alg} at order
$2q$, with accuracy parameter $\vartheta_q$ and failure parameter
$p_0$.
We obtain
\begin{align}
    \Dren*[q]{\nuhat}{\cPt{KT}\nu}\leq e^{\vartheta_q}-1\,.
\end{align}
Consequently,
\begin{align}
    \Ren*[q]{\nuhat}{\cPt{KT}\nu}
    \leq\frac{\vartheta_q}{q-1}\leq\delta\,.
\end{align}
If we restart after declaring failure, the output law remains preserved. The exponential tail bound for the sum of the FORS query counts over the $K$ steps and over the restarts gives $O(K+\log(1/p))$ queries with probability at least $1-p$.}
\end{proof}

\scedit{
\begin{proof}[\pfref{thm:ULA-main}]
Let $\nu_K\deq\cPt{KT}\nu$.

First suppose that PI holds. Apply \cref{lem:ULA-finite-time} at order $3$ and accuracy $\delta/4$. Since
\begin{align}
    3+\log(12K/\delta)\lesssim\log(K/\delta)\,,
\end{align}
the assumed step size condition gives
\begin{align}
    \Ren*[3]{\nuhat}{\nu_K}\leq\frac{\delta}{4}\,.
\end{align}
By \cref{thm:ST-PI} and the identity $\Dchis{\nu}{\target}=\exp(\Ren[2]{\nu}{\target})-1$, the assumed lower bound on $K$ ensures that
\begin{align}
    \Ren[2]{\nu_K}{\target}
    \leq\Dchis{\nu_K}{\target}\leq\frac{\delta}{2}\,.
\end{align}
Taking $q=3/2$ and $r=3$ in \cref{lem:Renyi-composition-ULA} gives
\begin{align}
    \Ren*[3/2]{\nuhat}{\target}
    \leq2\Ren*[3]{\nuhat}{\nu_K}+\Ren[2]{\nu_K}{\target}
    \leq\delta\,.
\end{align}

Now suppose that LSI holds. Apply \cref{lem:ULA-finite-time} at order $2q$ and accuracy $\delta/3$, giving
\begin{align}
    \Ren*[2q]{\nuhat}{\nu_K}\leq\frac{\delta}{3}\,.
\end{align}
Applying \cref{thm:ST-LSI} at order $2q$, and using $\Ren[2q]{\nu}{\target}\leq\Delta_q$, the assumed lower bound on $K$ yields
\begin{align}
    \Ren[2q]{\nu_K}{\target}\leq\frac{\delta}{2}\,.
\end{align}
For $q\geq2$, \cref{lem:Renyi-composition-ULA} and monotonicity of R\'enyi divergence give the simpler bound
\begin{align}
    \Ren*[q]{\nuhat}{\target}
    \leq\frac32\,\Ren*[2q]{\nuhat}{\nu_K}+\Ren[2q]{\nu_K}{\target}
    \leq\delta\,.
\end{align}
In both cases, the query bound follows from \cref{lem:ULA-finite-time}.
\end{proof}
}

\subsection{\pfref{lem:ULA-bias}}\label{appdx:proof-ULA-bias}

For simplicity, write $m(t)=\floor{t/h}h$ and
\begin{align}
    a_t\deq\mu(X_t)-\mu(X_{m(t)}).
\end{align}
Since $\d B_t=\d\Bbar_t-\mu(X_t)\,\dt$ under $\bP_z$, the discretization error is
\begin{align}
    b\deq \wstar(\bZ)-w(\bZ)
    =-\int_0^T\tri{a_t,\d\Bbar_t}
      +\int_0^T\tri{a_t,\mu(X_t)}\,\dt.
\end{align}
Set
\begin{align}
    R_0\deq\nrm{P_0}+T\nrm{\nabla V(X_0)},\qquad
    R_{\bar B}\deq\max_{0\leq t\leq T}\nrm{\Bbar_t},\qquad
    \mathcal R\deq R_{\bar B}+\frac{R_0}{\sqrt{2\gamma}}.
\end{align}
The formula for the exact ULD equation, followed by the $\beta$-Lipschitz property of $\nabla V$, gives
\begin{align}
    \max_{0\leq t\leq T}\nrm{P_t}
    \leq C\prn*{R_0+\sqrt{2\gamma}R_{\bar B}}
\end{align}
when $C\beta T^2\leq1$ and $\gamma T\leq1$. Consequently,
\begin{align}
    \nrm{a_t}\leq Ch\beta\mathcal R,\qquad
    \nrm{\mu(X_t)}\leq CT\beta\mathcal R,\qquad
    \abs*{\int_0^T\tri{a_t,\mu(X_t)}\dt}
    \leq Ch\beta^2T^2\mathcal R^2.
\end{align}

Let $M=-\int_0^T\tri{a_t,\d\Bbar_t}$. The
exponential martingale identity and the bound on $a_t$ give
\begin{align}
    \En_{\bP_z}\exp\prn*{\lambda\abs M
      -C\lambda^2h^2\beta^2T\mathcal R^2}\leq2.
\end{align}
For $k\geq2$, H\"older's inequality therefore yields
\begin{align}
    \En_{\bP_z}e^{\lambda|M|}
    \leq2^{1/k}
    \prn*{\En_{\bP_z}
      e^{C\lambda^2kh^2\beta^2T\mathcal R^2}}^{1-1/k}.
\end{align}
Choose $k=1/(C_0\lambda h\beta T\sqrt d)$; the constant in the
statement is small enough that $k\geq2$. Since
$\mathcal R^2\leq2R_{\bar B}^2+\gamma^{-1}R_0^2$, the
Brownian estimate in \cref{lem:max-B}, together with
$R_0^2\leq2R(z)$, $d\leq R(z)$, and $\gamma T\leq1$, gives
\begin{align}
    \log\En_{\bP_z}e^{\lambda|M|}
    \leq \frac{C\lambda h\beta}{\gamma\sqrt d}R(z)
\end{align}
provided $C\lambda h\beta T\sqrt d\leq1$. For the other
term, the same Brownian estimate gives
\begin{align}
    \log\En_{\bP_z}\exp\prn*{C\lambda h\beta^2T^2\mathcal R^2}
    \leq \frac{C\lambda h\beta^2T^2}{\gamma}R(z)
\end{align}
provided $C\lambda h\beta^2T^3\leq1$. Cauchy--Schwarz, with constants absorbed into $C$, proves the claim.
\qed

\subsection{\pfref{lem:ULA-clip}}\label{appdx:proof-ULA-clip}

Fix $t\in[0,T)$ and set $j\deq\floor{t/h}+1$. Define
\begin{align}
    u_j\deq \frac{\nabla V(X_{jh})-\nabla V(X_{(j-1)h})}{h\sqrt{2\gamma}}.
\end{align}
Under $\bP_z$,
\begin{align}
    B_T-B_{jh}=\Bbar_T-\Bbar_{jh}-\int_{jh}^T\mu(X_s)\d s.
\end{align}
The increment $\Bbar_T-\Bbar_{jh}$ is independent of $u_j$. With $\mathcal R$ as in the preceding proof,
\begin{align}
    \nrm{u_j}\leq C\beta\mathcal R,\qquad
    \nrm{\mu(X_t)}\leq CT\beta\mathcal R.
\end{align}
It follows that
\begin{align}
    \abs{W(t;\bZ)}
    \leq T\abs*{\tri{\Bbar_T-\Bbar_{jh},u_j}}
    +C\beta^2T^3\mathcal R^2.
\end{align}
Conditioning at time $jh$ and using the exponential martingale bound for the first term, followed by Cauchy--Schwarz, yields
\begin{align}
    \log \En_{\bP_z}\exp\prn*{\lambda\abs{W(t;\bZ)}}
    \leq \frac{\log2}{2}
    +\frac12\log\En_{\bP_z}
    \exp\prn*{C(\lambda^2+\lambda)\beta^2T^3\mathcal R^2}.
\end{align}
Since $\lambda\geq1$, the linear term is absorbed by $\lambda^2$. The condition $\lambda\leq(C\beta T^2)^{-1}$ and \cref{lem:max-B} therefore imply
\begin{align}
    \log\En_{\bP_z}
    \exp\prn*{C\lambda^2\beta^2T^3\mathcal R^2}
    \leq \frac{C\lambda^2\beta^2T^3}{\gamma}R(z),
\end{align}
which proves the result.
\qed

\begin{remark}\label{rmk:Girsanov}
Under $\bQ_z$, write
\begin{align}
    a_t\deq\frac{1-e^{-\gamma t}}{\gamma},
    \qquad b_t\deq\int_0^t a_s\d s.
\end{align}
The proposal has the explicit representation
\begin{align}
    X_t-X_0=a_t p-b_t\nabla V(x)
    +\sqrt{2\gamma}\int_0^t a_{t-s}\d B_s.
\end{align}
Here $a_t\leq T$, $b_t\leq T^2/2$, and the Gaussian term has covariance
at most $C\gamma T^3I_d$. Since
$\nrm{\mu(X_t)}\leq\beta\nrm{X_t-X_0}/\sqrt{2\gamma}$, the standard
Gaussian-square MGF bound, with $s=\lambda\beta^2T/(2\gamma)$, gives
\begin{align}
    \log\En_{\bQ_z}\exp\prn*{\lambda T\nrm{\mu(X_t)}^2}
    \leq C\lambda\gamma^{-1}\beta^2T
    \prn*{T^2\nrm p^2+T^4\nrm{\nabla V(x)}^2+\gamma T^3d},
\end{align}
provided $\lambda\beta^2T^4\leq c$. Consequently, if
$T\leq c/\sqrt\beta$ and $\gamma T\leq1$, then
\begin{align}
    \sup_{0\leq t<T}\log \En_{\bQ_z} \exp\prn*{ \lambda T\nrm{\mu(X_t)}^2 }\leq C\lambda\gamma^{-1}\beta^2T^3R(z)
\end{align}
for $0\leq\lambda\leq c/(\beta^2T^4)$. In particular, if
$T\leq c/\sqrt\beta$ and $\gamma T\leq1$, Jensen's inequality gives
\begin{align}
    \En_{\bQ_z}\exp\prn*{\int_0^T\nrm{\mu(X_t)}^2\dt}
    \leq \frac{1}{T}\int_0^T
    \En_{\bQ_z}\exp\prn*{T\nrm{\mu(X_t)}^2}\dt<\infty.
\end{align}
Thus Novikov's condition holds,
\begin{align}
    \En_{\bQ_z} \exp\prn*{ \int_0^T \nrm{\mu(X_t)}^2\dt }<+\infty,
\end{align}
and Girsanov's theorem can be applied.
\end{remark}

\subsection{Limitation of the exponential Euler proposal}\label{app:exponential-euler-limitation}\label{appdx:proof-lower-marginal}

We consider the example of a Gaussian target. Let $\Tbar$ denote the proposal transition kernel and $\T$ the exact underdamped Langevin transition kernel; thus, $\Tbar_z$ and $\T_z$ are their respective output distributions at time $T$ from $z$.

\paragraph{1D calculation}
Consider $d=1$ and $V(x)=\frac{w}{2}x^2$. The target distribution is then $\pi=\normal{0,w^{-1}\Id}\otimes\normal{0,\Id}$. Under $\bQ_z$, the dynamics are affine and hence $\Tbar_z=\normal{\bar m_T, \bar\Sigma_T}$ is Gaussian. More explicitly, set
\begin{align}
    a_T\deq \frac{1-e^{-\gamma T}}{\gamma},\qquad
    b_T\deq \frac{T}{\gamma}-\frac{1-e^{-\gamma T}}{\gamma^2}\,.
\end{align}
Then
\begin{align}
    \bar m_T
    \deq
    \En_{\Tbar_z}\binom{X_T}{P_T}
    =
    \binom{x+a_Tp-wb_Tx}{e^{-\gamma T}p-wa_Tx}\,,
\end{align}
and
\begin{align}
    \bar\Sigma_T=\begin{pmatrix}
        \frac{2}{\gamma}\,\prn*{T-\frac{2(1-e^{-\gamma T})}{\gamma}+\frac{1-e^{-2\gamma T}}{2\gamma}}
        & \frac{(1-e^{-\gamma T})^2}{\gamma} \\
        \frac{(1-e^{-\gamma T})^2}{\gamma}
        & 1-e^{-2\gamma T}
    \end{pmatrix}\,.
\end{align}
The following lower bound shows that the exponential Euler proposal can already incur a non-negligible one-step error. Note that the error below is measured \emph{marginally}, not conditionally.

\begin{lemma}[Marginal lower bound]\label{lem:warm-oscillatory-TV-min}
There are universal constants $c,\tau_0>0$ such that, if $0<\gamma\leq \frac{3}{2}\sqrt{\beta}$ and $\sqrt{\beta}T\leq\tau_0$, then there exists a probability distribution $\pi'$ satisfying
\[
    \frac12\leq \frac{\d\pi'}{\d\pi}\leq 2
\]
such that
\[
    \Dtv*{\T\pi'}{\Tbar\pi'}^2
    \geq
    c\,\min\crl*{\frac{\beta^2T^3}{\gamma},1}\,.
\]
\end{lemma}

\paragraph{Proof of \cref{lem:warm-oscillatory-TV-min}}
Set
\[
    \eps\deq \frac{\gamma}{\sqrt\beta},
    \qquad
    \tau\deq \sqrt\beta T .
\]
We take $V(x)=\frac{\beta}{2}x^2$. In the rescaled coordinates $Z=(Y,P)=(\sqrt\beta X,P)$, the invariant law is
$\normal{0,\Id_2}$, and the exact dynamics are
\[
    \d Y_t=P_t\d t,\qquad
    \d P_t=-(Y_t+\eps P_t)\d t+\sqrt{2\eps}\d B_t .
\]
Let
\[
    \theta\deq \sqrt{1-\eps^2/4}.
\]
For $0<\eps<2$, the exact transition satisfies
\[
    P_\tau\mid Z_0\sim \normal{\ell^\top Z_0,q},
\]
where
\[
    \ell
    =
    \binom{
        -e^{-\eps\tau/2}\frac{\sin(\theta\tau)}{\theta}
    }{
        e^{-\eps\tau/2}
        \prn*{\cos(\theta\tau)-\frac{\eps}{2\theta}\sin(\theta\tau)}
    }
\]
and
$q=1-\norm{\ell}^2$.

For the exponential Euler proposal,
\[
    \bar P_\tau\mid Z_0\sim\normal{\bar\ell^\top Z_0,\bar q},
\]
where
\[
    \bar\ell
    =
    \binom{
        -\frac{1-e^{-\eps\tau}}{\eps}
    }{
        e^{-\eps\tau}
    },
    \qquad
    \bar q=1-e^{-2\eps\tau}.
\]

We first record the elementary estimates used below, which can be shown by applying Taylor's expansion to $\tau$.

\begin{lemma}\label{claim:ell-q-estimates}
There exist universal constants $c,C,\tau_0>0$ such that, for all
$0<\eps\leq \frac32$ and $0<\tau\leq\tau_0$,
\begin{align}
    \frac12\leq \norm{\bar\ell}\leq 2,                                      \label{eq:ellbar-norm}\\
    c\tau^2
    \leq
    -\tri*{\bar\ell,\ell-\bar\ell}
    \leq C\tau^2,                                                           \label{eq:lbar-dot-delta}\\
    \norm{\ell-\bar\ell}\leq C\tau^2,                                        \label{eq:l-delta-bound}\\
    c\eps\tau\leq q,\bar q\leq C\eps\tau,
    \qquad
    |q-\bar q|\leq C\eps\tau^3 .                                             \label{eq:q-bounds}
\end{align}
\end{lemma}

\begin{proof}
All remainders below are uniform over $0<\eps\leq\frac32$. Indeed,
$\theta=\sqrt{1-\eps^2/4}$ is bounded away from zero on this interval, and the functions involved are smooth after defining $(1-e^{-\eps\tau})/\eps=\tau$ at $\eps=0$. Taylor's theorem therefore gives
\begin{align}
    \bar\ell
    &=\binom{-\tau+\frac{\eps}{2}\tau^2+O(\tau^3)}
                   {1-\eps\tau+\frac{\eps^2}{2}\tau^2+O(\tau^3)},\\
    \ell-\bar\ell
    &=\binom{\frac16\tau^3+O(\tau^4)}
                   {-\frac12\tau^2+O(\tau^3)}.
\end{align}
Consequently,
\begin{align}
    \nrm{\bar\ell}=1+O(\tau),\qquad
    -\tri{\bar\ell,\ell-\bar\ell}=\frac12\tau^2+O(\tau^3),\qquad
    \nrm{\ell-\bar\ell}=O(\tau^2).
\end{align}
Choosing a sufficiently small universal $\tau_0$ proves
\eqref{eq:ellbar-norm}--\eqref{eq:l-delta-bound}.

For the variance bounds, the exact covariance formula gives
\begin{align}
    q=2\eps\int_0^\tau e^{-\eps s}
    \prn*{\cos(\theta s)-\frac{\eps}{2\theta}\sin(\theta s)}^2\d s,
    \qquad
    \bar q=2\eps\int_0^\tau e^{-2\eps s}\d s.
\end{align}
Both integrands are bounded above and below by positive universal constants for $0\leq s\leq\tau_0$, which proves $q,\bar q\asymp\eps\tau$. The two integrands have the same value and first derivative at $s=0$; their second derivatives are uniformly bounded. Taylor's theorem thus bounds their difference by $Cs^2$, and integration gives $|q-\bar q|\leq C\eps\tau^3$.
\end{proof}

Now fix a parameter $K\geq 2$ to be decided later, and $\lambda\deq K-1$.
Define
\[
    h(z)\deq 1+\frac{1}{2}\sin(K\bar\ell^\top z),
    \qquad
    \pi'\deq h\pi .
\]
Since $\bar\ell^\top Z$ is centered Gaussian under $\pi$, we have
$\En_\pi\sin(K\bar\ell^\top Z)=0$, so $\pi'$ is a probability measure.

Let $\phi(p)\deq \sin(\lambda p)$.
Since $\norm{\phi}_\infty\leq1$,
\begin{align}
    2\Dtv*{\T\pi'}{\Tbar\pi'}
    \geq&~
    \abs*{
    \En_{\T\pi'}\phi(P)
    -
    \En_{\Tbar\pi'}\phi(P)} \\
    =&~\frac12\abs*{\En_\pi\brk*{
        \sin(K\bar\ell^\top Z_0)
        \prn*{\phi(P_\tau)-\phi(\bar P_\tau)}
    } }.
\end{align}
Here we used $\En_\pi\phi(P_\tau)=\En_\pi\phi(\bar P_\tau)=0$.

For a centered Gaussian $Z\sim\normal{0,\Id_2}$ and an independent
$\xi\sim\normal{0,s^2}$, the identity
\begin{align}
    \En\sin(K\bar\ell^\top Z)\sin(\lambda(\ell_0^\top Z+\xi))
    &=
    \frac12
    \exp\prn*{
        -\frac{\norm{K\bar\ell-\lambda\ell_0}^2+\lambda^2s^2}{2}
    }
    -
    \frac12
    \exp\prn*{
        -\frac{\norm{K\bar\ell+\lambda\ell_0}^2+\lambda^2s^2}{2}
    }
\end{align}
holds for every vector $\ell_0$ and variance $s^2$. Applying this with
$(\ell_0,s^2)=(\ell,q)$ and $(\ell_0,s^2)=(\bar\ell,\bar q)$, we obtain
\[
    A_1\deq \En_\pi\brk*{
        \sin(K\bar\ell^\top Z_0)
        \prn*{\phi(P_\tau)-\phi(\bar P_\tau)}
    }
    =
    \frac12\prn*{e^{-E_-/2}-e^{-\bar E_-/2}}
    -
    \frac12\prn*{e^{-E_+/2}-e^{-\bar E_+/2}},
\]
where
\begin{align}
    E_\pm\deq \norm{K\bar\ell\pm\lambda\ell}^2+\lambda^2q,
    \qquad
    \bar E_\pm\deq \norm{K\bar\ell\pm\lambda\bar\ell}^2+\lambda^2\bar q .
\end{align}

We first analyze the ``minus'' exponents. Since $K-\lambda=1$, it holds that $\bar E_- = \norm{\bar\ell}^2+\lambda^2\bar q \leq 5$ as long as $K^2\bar q\leq 1$. 
Let $\delta\deq \ell-\bar\ell$. Then
\begin{align}
    E_- - \bar E_-
    &=
    \norm{\bar\ell-\lambda\delta}^2-\norm{\bar\ell}^2
    +\lambda^2(q-\bar q)
    \notag\\
    &=
    -2\lambda\tri{\bar\ell,\delta}
    +\lambda^2\norm{\delta}^2+\lambda^2(q-\bar q).
\end{align}
Using \eqref{eq:lbar-dot-delta}--\eqref{eq:q-bounds}, as long as $K\leq \frac{c}{2C\eps\tau}$,
\begin{align}
    E_- - \bar E_-
    &\geq
    c K\tau^2
    -
    CK^2\eps\tau^3
    \geq \frac{c}{2}K\tau^2,
\end{align}
Consequently, assuming $K\leq \min\crl*{\frac{1}{\sqrt{C\eps\tau}},\frac{c}{2C\eps\tau}}$,
\[
    \abs*{e^{-E_-/2}-e^{-\bar E_-/2}}
    =
    e^{-\bar E_-/2}\prn*{1-e^{-(E_--\bar E_-)/2}}
    \geq
    c_1\min\prn*{K\tau^2,1}.
\]

It remains to show that the ``plus'' exponents are negligible. Since
$K+\lambda=2K-1$ and $\lambda\asymp K$,
\[
    \bar E_+
    =
    (K+\lambda)^2\norm{\bar\ell}^2+\lambda^2\bar q
    \geq \frac14K^2 .
\]
Similarly, using \eqref{eq:l-delta-bound}, as long as $4C\tau^2\leq 1$,
\[
    E_+
    =
    \norm{(K+\lambda)\bar\ell+\lambda\delta}^2+\lambda^2q
    \geq (K\norm{\bar\ell}-C\lambda\tau^2)_+^2\geq \frac{1}{16}K^2.
\]
 Therefore, $e^{-E_+/2}+e^{-\bar E_+/2}\leq 2e^{-K^2/32}$.

In summary, if $2\leq K\leq \min\crl*{\frac{1}{2\sqrt{C\eps\tau}},\frac{c}{2C\eps\tau}}$, then
\[
    2\Dtv*{\T\pi'}{\Tbar\pi'}\geq \frac12|A_1|
    \geq
    c_1\min\prn*{K\tau^2,1}-2e^{-K^2/32}.
\]
Therefore, there is a constant $\tau_1>0$ such that for $\tau\leq \tau_1$, we can choose $K=\frac{1}{2\sqrt{C\eps\tau}}\in[2,\frac{c}{2C\eps\tau}]$, and it holds that
\[
    2\Dtv*{\T\pi'}{\Tbar\pi'}
    \geq
    c_2\min\prn*{\sqrt{\tau^3/\eps},1}-2e^{-c_2/(\eps\tau)}.
\]
For a sufficiently small universal constant $\tau_2>0$ and $\tau\in(0,\tau_2]$, the exponential term is absorbed uniformly over $0<\eps\leq\frac32$. Thus, for $\tau\leq\tau_2$,
\begin{align}
    \Dtv*{\T\pi'}{\Tbar\pi'}
    \geq c_3\min\prn*{\sqrt{\tau^3/\eps},1}.
\end{align}
Squaring this inequality and using $\tau^3/\eps=\beta^2T^3/\gamma$ proves the claim.
Finally, push the constructed law forward under $(Y,P)\mapsto(Y/\sqrt\beta,P)$ to conclude. %
\qed

\section{Proofs from \cref{sec:ULA2}}\label{app:ULA2}

\subsection{Analysis of FORS}

We state the following results.

\begin{lemma}\label{lem:ULA2-clip}
For $\lambda\geq1$, it holds that
\begin{align}
    \En_{\bP_z}\En_{t\sim \Unif([0,T])} \exp\prn*{\lambda \abs{W(t;\bZ)}}
    \leq&~ 2\En_{\bP_z} \En_{t\sim \Unif([0,T])}\exp\prn*{6\lambda^2T^3\nrm{\partial_t \mu_t}^2}.
\end{align}
\end{lemma}

Proof in \cref{appdx:proof-ULA2-clip}.

\begin{lemma}\label{lem:ULA-fine-pathwise}
Suppose $T\sqrt{\beta}\leq \frac14$, $\gamma T\leq\frac14$, and $\gamma\leq\sqrt\beta$. Let
\begin{align}
    R_B\deq \sup_{0\leq t\leq T}\nrm{\Bbar_t},
    \qquad
    R\deq 4\prn*{\nrm{P_0}+T\nrm{\nabla V(X_0)}+\sqrt{\gamma}R_B}.
\end{align}
Then, under \cref{eq:ULA2-iterative}, $\nrm{P^\ell_t}\leq R$ and $\nrm{X^\ell_t-X_0}\leq TR$ for every $\ell\geq1$ and $t\in[0,T]$. For every $\ell\geq2$,
\begin{align}\label{eq:fine-P-increment}
    \sup_{0\leq t\leq T}\nrm{P^\ell_t-P^{\ell-1}_t}
    \leq (T\sqrt{\beta})^{2\ell-2} R,
\end{align}
For every $\ell\geq1$,
\begin{align}\label{eq:fine-X-increment}
    \sup_{0\leq t\leq T}\nrm{X^\ell_t-X^{\ell-1}_t}
    \leq T(T\sqrt{\beta})^{2\ell-2} R.
\end{align}
\end{lemma}

Proof in \cref{appdx:ULA-fine-pathwise}.

\begin{proposition}\label{prop:ULA2-clip-full}
Suppose that $\nu$ is supported on the set $\cE_{\Rbar}=\crl*{z: R(z)\leq \Rbar}$. Let $\rho=\beta T^2$, $\Delta=1+\Ren[1.5]{\nu}{\target}$, $r>0$, and $K,L\in\mathbb N$. Suppose that $T\leq \frac{1}{4\sqrt{\beta}}$, $\gamma T\leq\frac14$, $\gamma\leq\sqrt\beta$, $T^3\leq \frac{1}{\betaF\sqrt{\Rbar}}$, and
\begin{align}
    \frac{1}{Cr}\geq \beta^2 T+\gamma^{-1}\betaF \beta KT\sqrt{\Rbar}+\gamma^{-1}\betaF^2T^3\Rbar,
\end{align}
then
\begin{align}
    &~\log \En_{z\sim \nu}\prn*{ \En_{\bZ\sim \bP_z} \exp\prn*{ r\max_{0\leq t\leq T} \nrm{\partial_t \mu_t}^2 }}^K \\
    \leq&~ CrK\beta^2T+CrK\gamma^{-1}\prn*{\betaF^2T^2(d+\Delta^2)+ \beta^2 \rho^{2L-2} \Rbar+\betaF \beta T\sqrt{\Rbar}+\betaF^2T^3\Rbar+\beta^{3}\betaF^2T^8\Rbar^2}  \\
    &~ + CrK\gamma^{-1}\min\crl*{\beta^3T^2\Rbar,\,\betaF^2\beta T^4\Rbar \Delta} \\
    \leq&~ K+CrK\gamma^{-1}\prn*{\betaF^2T^2(d+\Delta^2)+ \beta^2 \rho^{2L-2} \Rbar+\beta^{3}\betaF^2T^8\Rbar^2+\min\crl*{\beta^3T^2\Rbar,\,\betaF^2\beta T^4\Rbar\Delta}},
\end{align}
\end{proposition}

Proof in \cref{appdx:proof-ULA2-clip-full}.

\begin{corollary}\label{cor:ULA2-Renyi-onestep}
Suppose that $\nu$ is supported on the set $\cE_{\Rbar}=\crl*{z: R(z)\leq \Rbar}$. Let $\rho=\beta T^2$ and $\Delta=1+\Ren[1.5]{\nu}{\target}$.

Suppose also that $T\leq\frac{1}{4\sqrt\beta}$, $\gamma T\leq\frac14$, $\gamma\leq\sqrt\beta$, and $\betaF T^3\sqrt{\Rbar}\leq1$. Let $\delta\in(0,\frac12]$, $B=\Theta(1)$, $K\geq1$, $q\geq 2$,
\begin{align}
    &~\frac{1}{C(q^2+\log^2(K/\delta))}\geq \beta^2T^4+\gamma^{-1}\betaF \beta KT^4\sqrt{\Rbar}+\gamma^{-1}\betaF^2T^6\Rbar, \\
    &~\frac{1}{C(q^2+q\log(K/\delta))}\geq \gamma^{-1}\Bigl(\betaF^2T^5(d+\Delta^2)+ \beta^2 T^3 \rho^{2L-2} \Rbar+\beta^{3}\betaF^2T^{11}\Rbar^2 
    +\min\crl*{\beta^3T^5\Rbar,\,\betaF^2\beta T^7\Rbar\Delta}\Bigr).
\end{align}
Then 
\begin{align}
    \En_{z\sim \nu}\prn*{1+\Dsys*[q]{\That_z}{\T_z}}^K\leq1+\delta.
\end{align}
\end{corollary}

Proof in \cref{appdx:proof-ULA2-Renyi-onestep}.

\subsection{\pfref{thm:ULA2-warm}}\label{appdx:ULA2-warm}

\begin{proof}[\pfref{thm:ULA2-warm}]
Similar to the proof of \cref{thm:ULA-main}, we consider the set $\cE=\crl{z:R(z)\leq C_0(\Delta+d+\log(K/\delta))}$,
where $C_0>0$ is a sufficiently large constant so that $1-\target(\cE)\leq \prn*{\frac{p}{64K}}^3e^{-2\Delta}$.
Then, by \cref{cor:ULA2-Renyi-onestep}, we can guarantee that as long as $\nu'$ satisfies $\Ren[1.5]{\nu'}{\target}\leq 2\Delta$, we have
$$\En_{z\sim \nu'}\prn*{1+\Dsys*[2q]{\That_z}{\T_z}}^K\leq1+\delta.$$
Applying \cref{thm:trunc-alg} gives the desired result.
\end{proof}

\subsection{Proof of \cref{thm:ULA2-cold}}\label{appdx:ULA2-cold}

Let $\Tbar_z$ be the distribution of $(X_T,P_T)$ under \cref{eq:ULA2-iterative}. The following guarantee for $\Tbar$ follows from \cref{prop:ULA2-clip-full}.

\begin{corollary}\label{cor:ULA2-Renyi-onestep-crude}
Let $q\geq2$. Suppose that the hypotheses of \cref{prop:ULA2-clip-full} hold with $r=q^2T^3$, after adjusting the universal constant $C$. Then
\begin{align}
    \En_{z\sim \nu}\prn*{1+\Dsys*[q]{\Tbar_z}{\T_z}}^K\leq&~ \exp\Bigl(C+CKq^2\gamma^{-1}T^3\bigl(\betaF^2T^2(d+\Delta^2)+ \beta^2 \rho^{2L-2} \Rbar+\gamma\beta^2T \\
    &~+\betaF^2T^3\Rbar+\beta^{3}\betaF^2T^{8}\Rbar^2+\min\crl*{\beta^3T^2\Rbar,\,\betaF^2\beta T^4\Rbar\Delta}\bigr)\Bigr).
\end{align}
Further, we also have the crude bound $1+\Dsys*[q]{\Tbar_z}{\T_z}\leq e^{Cq^2\gamma^{-1}\beta^2T^3\Rbar}$ for any $z\in\cE_{\Rbar}$.
\end{corollary}

Proof in \cref{appdx:proof-ULA2-Renyi-onestep-crude}.

\begin{proposition}\label{prop:ULA2-warm-gen}
\screplace{Suppose that \cref{asmp:LSI,asmp:Lip-Hess} hold,
$w$ is a constant, and $L\in\mathbb N$.}{Suppose that $V$ is convex and $\beta$-smooth, that $\mu_V$ is log-concave and satisfies LSI with constant $1/\alpha$, and that \cref{asmp:Lip-Hess} holds. Set $\gamma\asymp\sqrt\alpha$, $\kappa\deq\beta/\alpha$, and $\kappaF\deq\betaF^{2/3}/\alpha$. Let $w$ be a constant and $L\in\mathbb N$.} Suppose that a constant $c>0$
satisfies $w\geq 2+c$, and set
\begin{align}
    \Delta\deq d+2+\Ren[w]{\nu}{\target},
    \qquad \Lambda\deq\log \Delta.
\end{align}
Then as long as
\begin{align}
    M\deq \frac{1}{\gamma T}\gg \sqrt{\kappa \Lambda}+\kappa^{1/3}(\kappaF+\kappaF^{\frac{3}{10}})(\Delta+\log M)^{1/6}\Lambda,
\end{align}
there is $N=O(M\Lambda^2)$ and an algorithm from \cref{thm:trunc-alg} that uses at most $N$ queries to $\Tbar$, succeeds with probability at least $1-e^{-\Delta}$, and, conditional on success, has output distribution $\nuhat$ satisfying
\begin{align}
    \Ren*[w-c]{\nuhat}{\target}
    \leqsim&~ 1+\gamma^{-2}\bigl(\betaF^2T^4d+ (\beta^2 T^2\rho^{2L-2}+\betaF^2T^5+\betaF^2\beta T^6)(\Delta+\log M) \\
    &~\hspace{31mm}+\beta^{3}\betaF^2T^{10}(\Delta+\log M)^2+\gamma\beta^2T^3\bigr)\Lambda^3 \\
    \leqsim&~ (\Delta+\log M)^{1/3}+\kappa^{1/2}+\gamma^{-2}\beta^2T^2\rho^{2L-2}(\Delta+\log M)\Lambda^3.
\end{align}
\end{proposition}

Proof in \cref{appdx:proof-ULA2-warm-gen}. 

We finish the proof of \cref{thm:ULA2-cold}. Let
\begin{align}
    \iota&\deq 1+\log\prn*{\frac{\kappa(1+\kappaF)d}{\delta}}, \\
    A_1&\deq \sqrt\kappa,\qquad
    A_2\deq \kappaF^{3/5}d^{1/5},\qquad
    A_3\deq \kappa^{1/3}(\kappaF+\kappaF^{3/11})d^{2/11}.
\end{align}
Choose
\begin{align}\label{eq:ULA2-cold-choice}
    M\deq \frac{1}{\gamma T}=C\iota^2(A_1+A_2+A_3),
    \qquad L=\left\lceil C\log(2+\kappa d)\right\rceil.
\end{align}

Set $\Delta=d+2+\Ren[w]{\nu}{\target}=O(d\log(2\kappa))$ and apply \cref{prop:ULA2-warm-gen} with order $w$ and $c=1/4$. The step size condition holds because
\begin{align}
    \sqrt{\kappa\log \Delta}
    +\kappa^{1/3}(\kappaF+\kappaF^{3/10})(\Delta+\log M)^{1/6}\log \Delta
    \leq C\iota^{7/6}(A_1+A_3)\leq M/C.
\end{align}
Here we used $d^{1/6}\leq d^{2/11}$ and
$\kappaF+\kappaF^{3/10}\leq2(\kappaF+\kappaF^{3/11})$.
Let $\nu_\star$ be the resulting distribution and write
\begin{align}
    D_\star\deq1+\Ren*[w-1/4]{\nu_\star}{\target}.
\end{align}
Repeat this stage upon failure; its success probability gives only a constant expected overhead.
We substitute $T=(\gamma M)^{-1}$ to obtain
\begin{align}\label{eq:ULA2-cold-Dstar}
    D_\star\leq1+C\iota^5\prn*{
    \frac{\kappaF^3d}{M^4}
    +\frac{\kappa\kappaF^3d}{M^6}
    +\frac{\kappa^3\kappaF^3d^2}{M^{10}}
    +\frac{\kappa^2}{M^3}},
\end{align}
and also
\begin{align}\label{eq:ULA2-cold-Dstar-coarse}
    D_\star\leq C\iota^3(d^{1/3}+\sqrt\kappa).
\end{align}
Indeed, the omitted $M^{-5}$ term is bounded by the $M^{-4}$ term, while $M\geq C\iota^2\sqrt\kappa$ and the choice of $L$ make the iterative term's contribution to $D_\star$ at most a constant.

For the final stage, set $q=w-1/2$, $s=w-1/4$,
\begin{align}
    r\deq \frac{q(s-1)}{s-q}=4(w-1/2)(w-5/4),
\end{align}
and set
\begin{align}
    U=C\log(D_\star/\delta),\qquad
    K=\left\lceil CM\log(D_\star/\delta)\right\rceil.
\end{align}
Then $K\leq UM$ and $U+\log(K/\delta)\leq C\iota$. We apply \cref{thm:ULA2-warm} at order $r$ and accuracy $\delta^2/3$, with initial law $\nu_\star$, and take its second branch. The bound \eqref{eq:ULA2-cold-Dstar-coarse} shows directly that all terms outside this branch are at most $C\iota^2(A_1+A_2+A_3)$. For the last term, substitute \eqref{eq:ULA2-cold-Dstar} and use
\begin{align}
    ((d+D_\star)D_\star)^{1/7}\leq(dD_\star)^{1/7}+D_\star^{2/7}.
\end{align}
Balancing the four powers of $M$ in \eqref{eq:ULA2-cold-Dstar} yields
\begin{align}
    \max\Bigl\{&\kappa^{1/7}\kappaF^{3/7}d^{1/7},
    \kappa^{1/11}\kappaF^{6/11}d^{2/11},
    \kappa^{2/13}\kappaF^{6/13}d^{2/13},
    \kappa^{4/17}\kappaF^{6/17}d^{3/17},
    \kappa^{3/10}\kappaF^{3/10}d^{1/10},\\
    &\kappa^{1/7}\kappaF^{3/7},
    \kappa^{1/15}\kappaF^{3/5}d^{2/15},
    \kappa^{3/19}\kappaF^{9/19}d^{2/19},
    \kappa^{7/27}\kappaF^{1/3}d^{4/27},
    \kappa^{5/13}\kappaF^{3/13}\Bigr\}
    \leq C(A_1+A_2+A_3).\label{eq:ULA2-cold-balances}
\end{align}
To see the last inequality, the first term is $A_1^{2/7}A_2^{5/7}$; the second through fifth, as well as the eighth and ninth, are bounded by $A_3$ after splitting at $\kappaF=1$. The sixth is bounded by the first. Finally, the seventh and tenth are bounded by $A_1+A_2+A_3$ after splitting at $\kappaF=\sqrt\kappa$. The largest logarithmic factor arising from this substitution is $\iota^{12/7}$, which is absorbed by the $\iota^2$ in \eqref{eq:ULA2-cold-choice}. Thus \eqref{eq:ULA2-T-cond-warm} holds.

The exact law $\mu$ at time $KT$ satisfies $\Ren[s]{\mu}{\target}\leq\delta^2/3$ by \cref{thm:ST-LSI}, and the simulated law $\nuhat$ satisfies $\Ren*[r]{\nuhat}{\mu}\leq\delta^2/3$. Since $1+(q-1)r/(r-q)=s$ and the prefactor in \cref{lem:Renyi-composition-ULA} is at most $2$, the lemma gives $\Ren*[q]{\nuhat}{\target}\leq\delta^2$. The two stages use at most $C\iota^4(A_1+A_2+A_3)$ queries, which is the claimed bound.
\qed

\subsection{\pfref{cor:ULA2-Renyi-onestep}}\label{appdx:proof-ULA2-Renyi-onestep}
By \cref{lem:FORS-error} and \cref{lem:ULA2-clip}, we can take $\ell\geq1$ and bound
\begin{align}
    \Dsys*[q]{\That_z}{\T_z}\leq&~ \En_{\bP_z} \En_{t\sim \Unif([0,T])} \exp\prn*{2q(\abs{W(t;\bZ)}-B)_+}-1 \\
    \leq&~ e^{-2Bq\ell}\En_{\bP_z} \En_{t\sim \Unif([0,T])} \exp\prn*{2q\ell \abs{W(t;\bZ)}} \\
    \leq&~ 2e^{-2Bq\ell}\En_{\bP_z} \En_{t\sim \Unif([0,T])}\exp\prn*{24\ell^2q^2T^3\nrm{\partial_t \mu_t}^2}.
\end{align}
Then, by \cref{prop:ULA2-clip-full}, there are constants $C_0,C_2>0$ such that as long as
\begin{align}
    \frac{1}{C_0\ell^2q^2}\geq \beta^2T^4+\gamma^{-1}\betaF \beta KT^4\sqrt{\Rbar}+\gamma^{-1}\betaF^2T^6\Rbar,
\end{align}
it holds that for $1\leq j\leq K$,
\begin{align}
    \En_{z\sim \nu}\Dsys*[q]{\That_z}{\T_z}^j
    \leq&~ C_2^j\exp\Bigl(-2jBq\ell+C_0j\ell^2q^2\gamma^{-1}T^3\bigl(\betaF^2T^2(d+\Delta^2)+ \beta^2 \rho^{2L-2} \Rbar \\
    &~\hspace{38mm}+\beta^{3}\betaF^2T^{8}\Rbar^2+\min\crl*{\beta^3T^2\Rbar,\,\betaF^2\beta T^4\Rbar\Delta}\bigr)\Bigr).
\end{align}
Therefore, as long as
\begin{align}
    \frac{1}{C_0\ell q^2}\geq&~ \gamma^{-1}\Bigl(\betaF^2T^5(d+\Delta^2)+ \beta^2 T^3 \rho^{2L-2} \Rbar+\beta^{3}\betaF^2T^{11}\Rbar^2 \\
    &~\hspace{35mm}+\min\crl*{\beta^3T^5\Rbar,\,\betaF^2\beta T^7\Rbar\Delta}\Bigr),
\end{align}
it holds that for $1\leq j\leq K$, $\En_{z\sim \nu}\Dsys*[q]{\That_z}{\T_z}^j\leq C_2^je^{-j\ell qB}$. In particular, as long as $\ell qB\geq \log(2C_2K/\delta)$, we can use $(1+x)^K\leq \sum_{j=0}^K (Kx)^j$ to bound
\begin{align}
    \En_{z\sim \nu}\prn*{1+\Dsys*[q]{\That_z}{\T_z}}^K\leq \sum_{j=0}^K K^j\En_{z\sim \nu}\Dsys*[q]{\That_z}{\T_z}^j\leq \sum_{j=0}^K (\delta/2)^j\leq 1+\delta.
\end{align}
\qed

\subsection{\pfref{prop:ULA2-clip-full}}\label{appdx:proof-ULA2-clip-full}

We denote $\rho=T^2\beta$, $H_t\deq \nabla^2 V(X_t)$, $Y_t\deq (H_t-H_0)P_t$ for $t\geq 0$. Then by \cref{lem:ULA-fine-pathwise}, we have shown $\nrm{Y_t-\sqrt{2\gamma}\partial_t \mu_t}\leq \beta \nrm{P_t-P^{L-1}_t}\leq \beta \rho^{L-1} R_0$,
where
\begin{align}
    R_0\deq 4(\nrm{p}+\beta^{-1/2}\nrm{\nabla V(x)}+\sqrt{\gamma}R_B), \qquad
    R_B\deq \max_{0\leq t\leq T}\nrm{\Bbar_t}.
\end{align}
Further, for every $t\in[0,T]$, we have $\nrm{X_t-X_0}\leq TR_0$ and
\begin{align}
    \nrm{\nabla V(X_t)}\leq \nrm{\nabla V(X_0)}+\beta TR_0\leq 2\sqrt{\beta}R_0.
\end{align}

\paragraph{Step 1: Upper bound on $\nrm{Y_t}$}
Note that
\begin{align}
    \d Y_t= \nabla^3 V(X_t)[P_t,P_t]\dt+(H_t-H_0)\d P_t.
\end{align}
We denote $A_t\deq \nabla^3 V(X_t)[P_t,P_t]$, and then
\begin{align}
    \d Y_t= (A_t-\gamma Y_t-(H_t-H_0)\nabla V(X_t))\dt+\sqrt{2\gamma}(H_t-H_0)\d \Bbar_t.
\end{align}
Now, we define
\begin{align}
    S_t \deq \int_0^t e^{\gamma s}(H_s-H_0)\d \Bbar_s, \qquad
    Z_t\deq (H_t-H_0)\nabla V(X_t),
\end{align}
and by definition, we have
\begin{align}
    \nrm{Y_t}\leq \int_0^t (\nrm{A_s}+\nrm{Z_s})\d s+\sqrt{2\gamma} e^{-\gamma t} \nrm{S_t}.
\end{align}
Now, by taking the derivative recursively, we have
\begin{align}
    Z_t
    =\int_0^t \nabla^3 V(X_s)[P_s, \nabla V(X_s)] \d s+
    \int_0^t  (H_s-H_0)H_sP_s \d s,
\end{align}
and we denote $U_t\deq (H_t-H_0)H_tP_t$, so that
\begin{align}
    \d U_t=(\nabla^3 V(X_t)[P_t,H_tP_t]+(H_t-H_0)(A_t-H_t\nabla V(X_t))-\gamma U_t)\dt+\sqrt{2\gamma}(H_t-H_0)H_t\d \Bbar_t.
\end{align}
Now, define 
\begin{align}
    S_t'\deq \int_0^t e^{\gamma s}(H_s-H_0)H_s\d \Bbar_s, \qquad
    A_t'\deq \nabla^3 V(X_t)[P_t, \nabla V(X_t)], \qquad
    A_t''\deq \nabla^3 V(X_t)[P_t,H_tP_t].
\end{align}
Then, using $\nrm{\nabla V(X_t)}\leq 2\sqrt{\beta}R_0$, $\nrmop{H_t}\leq \beta$, and $\nrmop{H_t-H_0}\leq \betaF\nrm{X_t-X_0}\leq \betaF TR_0$, we can bound
\begin{align}
    \nrm{U_t}\leq \int_0^t (\nrm{A_s''}+2\beta\nrm{A_s}+2\beta^{3/2}\betaF TR_0^2)\d s+\sqrt{2\gamma} e^{-\gamma t}\nrm{S_t'},
\end{align}
and hence
\begin{align}
    \nrm{Z_t}\leq \int_0^t (\nrm{A_s'}+T\nrm{A_s''}+2\beta T\nrm{A_s}+\sqrt{2\gamma} \nrm{S_s'})\d s+2\beta^{3/2}\betaF T^3R_0^2.
\end{align}
Combining the above bounds, we can bound
\begin{align}
    \nrm{Y_t}\leq 2\int_0^t (\nrm{A_s}+T\nrm{A_s'}+T^2\nrm{A_s''}+\sqrt{\gamma} T\nrm{S_s'})\d s
    +2\beta^{3/2}\betaF T^4R_0^2+\sqrt{2\gamma} \nrm{S_t}.
\end{align}
Alternatively, $\nrm{Z_t}\leq 2\beta^{3/2}R_0$, and the first bound on $Y_t$ gives
\begin{align}
    \nrm{Y_t}\leq \int_0^t\nrm{A_s}\d s+2\beta^{3/2}TR_0+\sqrt{2\gamma}\nrm{S_t}.
\end{align}

\paragraph{Step 2: High probability upper bounds}
Fix any $t\in[0,T]$. 

(i) By \cref{lem:max-B}, it holds that \whp~under $\bP_z$, $R_B\leqsim \sqrt{T(d+\log(1/\delta))}$, and hence
\begin{align}
    R_0^2\leqsim R(z)+\gamma T(d+\log(1/\delta))
    \leqsim \Rbar+\gamma T\log(1/\delta)=:R_+^2.
\end{align}

(ii) Note that under the event of (i), it holds that
\begin{align}
    \int_0^T \nrmF{H_s-H_0}^2 \d s\leq T(\betaF TR_0)^2\leqsim \betaF^2 T^3R_+^2.
\end{align}
Since $\gamma T\leq1/4$, the factors in the quadratic variations satisfy $e^{2\gamma s}\leq e^{1/2}$. Thus, by \cref{lem:Ito-Freedman}, it holds that \whp[2\delta]~under $\bP_z$,
\begin{align}
    \max_{0\leq t\leq T}\prn*{\nrm{S_t}^2+\beta^{-2}\nrm{S_t'}^2} \leqsim
    \betaF^2 T^3R_+^2\log(1/\delta).
\end{align}

(iii) We define $F_t\deq \nabla^3 V(X_t)[\nabla V(X_t)]$ and
\begin{align}
    E_t\deq&~ \brk*{ \betaF^{-1}\nrm{A_t}+\beta^{-1}\betaF^{-1}\nrm{A_t''}-C_1(\sqrt{d}+\Delta)}_++\brk*{ \frac{\nrm{F_tP_t}^2}{\nrmF{F_t}^2}-C_1\Delta}_+,  \\
    R_E(z)\deq&~ \log \En_{\bZ\sim \bP_z, t\sim \Unif([0,T])}\exp\prn*{c_1E_t},
\end{align}
where the last quotient is defined to be zero when $\nrmF{F_t}=0$, and $C_1,c_1>0$ are appropriate constants so that \cref{lem:Gaussian-chaos-chis} applies, i.e.,
\begin{align}
    \En_{z\sim \nu} \exp\prn*{R_E(z)}\leq 2.
\end{align}
Now, by Markov's inequality, it holds that \whp~under $\bP_z$,
\begin{align}
    &~\int_0^T (\nrm{A_t}+\beta^{-1}\nrm{A_t''})\dt\leqsim \betaF T(\sqrt{d}+\Delta+R_E(z)+\log(1/\delta)), \\
    &~\int_0^T \frac{\nrm{F_tP_t}^2}{\nrmF{F_t}^2}\dt \leqsim T(1+\Delta+R_E(z)+\log(1/\delta)).
\end{align}
Note that under this event, we can choose $\eta=\frac{\betaF\sqrt{\beta}R_0}{\sqrt{1+\Delta+R_E(z)+\log(1/\delta)}}$ and
further bound (using $\nrmF{F_t}\leq \betaF\nrm{\nabla V(X_t)}\leq 2\betaF\sqrt{\beta}R_0$)
\begin{align}
    \frac{1}{T}\int_0^T \nrm{A_t'}\dt=\frac{1}{T}\int_0^T \nrm{F_tP_t}\dt \leq&~ \frac{1}{T}\int_0^T \prn*{ \frac{\eta\nrm{F_tP_t}^2}{\nrmF{F_t}^2} + \eta^{-1}\nrmF{F_t}^2}\d t \\
    \leqsim&~ \eta(1+\Delta+R_E(z)+\log(1/\delta))+\eta^{-1}(\betaF\sqrt{\beta}R_0)^2 \\
    \leqsim&~ \betaF \sqrt{\beta}R_0 \sqrt{1+\Delta+R_E(z)+\log(1/\delta)}.
\end{align}

\paragraph{Step 3: Finalizing the proof}
Recall that we denote $R_+^2\deq \Rbar+\gamma T\log(1/\delta)$ so $R_0^2\leqsim R_+^2$. 
In the following, we condition on the event of (i)--(iii), which holds \whp[4\delta]~under $\bP_z$, and denote $\iota\deq R_E(z)+\log(1/\delta)$. The refined bound from Step 1 gives
\begin{align}
    \frac{1}{\betaF T}\max_{0\leq t\leq T}\nrm{Y_t}
    \leqsim &~ \sqrt d+\Delta+\iota
    +T\sqrt{\beta R_+^2(1+\Delta+\iota)}
    +\sqrt{\gamma TR_+^2\log(1/\delta)}+\beta^{3/2}T^3R_+^2 \\
    \leqsim &~\sqrt d+\Delta+\iota+\sqrt{\beta T^2\Rbar(1+\Delta+R_E(z))} \\
    &~+\sqrt{(\gamma T +\beta T^2)\Rbar \log(1/\delta)}+\beta^{3/2}T^3\Rbar,
\end{align}
where we use $\gamma T\leq\sqrt{\beta}T\leq 1$. Therefore,
\begin{align}
    \max_{0\leq t\leq T}\nrm{Y_t}^2
    \leqsim &~\betaF^2T^2(\iota^2+d+\Delta^2)+\betaF^2\beta T^4\Rbar(1+\Delta+\iota) \\
    &~+\gamma\betaF^2T^3\Rbar\log(1/\delta)+\beta^{3}\betaF^2T^8\Rbar^2.
\end{align}

The alternative bound from Step 1 gives, on the same event,
\begin{align}
    \max_{0\leq t\leq T}\nrm{Y_t}^2
    \leqsim&~ \betaF^2T^2(\iota^2+d+\Delta^2)+\beta^3T^2\Rbar \\
    &~+\gamma\betaF^2T^3\Rbar\log(1/\delta)+\beta^2\gamma T\log(1/\delta).
\end{align}
We also have the trivial bound
\begin{align}
    \nrm{Y_t}^2\leq 4\beta^2\nrm{P_t}^2\leq 4\beta^2 R_0^2\leqsim \beta^2(\Rbar+\gamma T\log(1/\delta)).
\end{align}

We use $\min\crl{a_1+a_2,b_1+b_2}\leq \sqrt{a_1b_1}+a_2+b_2$ only to replace the common $\iota^2$ term by a term linear in $\iota$. Since $\betaF T^3\sqrt{\Rbar}\leq1$, the term $\betaF^2\beta T^4\Rbar\iota$ is absorbed by $\betaF\beta T\sqrt{\Rbar}\iota$. Applying this argument to the two estimates separately, and then retaining the better one, gives
\begin{align}
    \max_{0\leq t\leq T}\nrm{Y_t}^2
    \leqsim&~ \gamma (\beta^2T+\betaF^2T^3\Rbar)\log(1/\delta)+\betaF\beta T\sqrt{\Rbar}\,\iota+\betaF^2T^2(d+\Delta^2)+\beta^{3}\betaF^2T^8\Rbar^2 \\
    &~+\min\crl*{\beta^3T^2\Rbar,\,\betaF^2\beta T^4\Rbar\Delta}.
\end{align}
Further, we recall that $\sqrt{2\gamma}\nrm{\partial_t \mu_t}\leq \nrm{Y_t}+ \beta \rho^{L-1} R_0$, and hence
\begin{align}
    \max_{0\leq t\leq T}\nrm{\partial_t \mu_t}^2\leqsim&~ (\beta^2T+\gamma^{-1}\betaF^2T^3\Rbar)\log(1/\delta)+\gamma^{-1}\betaF \beta T\sqrt{\Rbar}\iota \\ 
    &~ +\gamma^{-1}\prn*{\betaF^2T^2(d+\Delta^2)+\beta^{3}\betaF^2T^8\Rbar^2+\beta^2 \rho^{2L-2} \Rbar} \\
    &~+\gamma^{-1}\min\crl*{\beta^3T^2\Rbar,\,\betaF^2\beta T^4\Rbar\Delta}.
\end{align}
In the following, we denote
\begin{align}
    M_0&\deq \beta^2T+\gamma^{-1}\betaF^2T^3\Rbar, \qquad M_1\deq \gamma^{-1}\betaF \beta T\sqrt{\Rbar}, \\
    G&\deq \gamma^{-1}\prn*{\betaF^2T^2(d+\Delta^2)+ \beta^{3}\betaF^2T^8\Rbar^2+\beta^2 \rho^{2L-2} \Rbar+\min\crl*{\beta^3T^2\Rbar,\,\betaF^2\beta T^4\Rbar\Delta}}.
\end{align}
Then, for $0\leq r\leq \frac{1}{C_2(M_0+M_1)}$, integrating gives
\begin{align}
    \En_{\bZ\sim \bP_z} \exp\prn*{ r\max_{0\leq t\leq T}\nrm{\partial_t \mu_t}^2 }\leq \exp\prn*{C_2r(G+M_1R_E(z)+M_0+M_1)}
\end{align}
for an absolute constant $C_2>0$. Taking expectation over $z\sim \nu$ gives
\begin{align}
    \En_{z\sim \nu}\prn*{ \En_{\bZ\sim \bP_z} \exp\prn*{ r\max_{0\leq t\leq T}\nrm{\partial_t \mu_t}^2 } }^K \leq \exp\prn*{C_3rK(G+M_0+M_1)}
\end{align}
as long as $r\leq \frac{1}{C_3KM_1}$, and
the desired upper bound follows.
\qed

\subsection{Proof of \cref{cor:ULA2-Renyi-onestep-crude}}\label{appdx:proof-ULA2-Renyi-onestep-crude}

By Girsanov's theorem, it is straightforward to verify that
\begin{align}
    1+\Dsys*[q]{\Tbar_z}{\T_z}
    \leq&~ \En_{\bP_z}\exp\prn*{q^2\int_0^T \nrm{\mu_t}^2\dt}
    \leq \En_{\bP_z}\exp\prn*{q^2T^3\max_{0\leq t\leq T}\nrm{\partial_t \mu_t}^2}.
\end{align}
The first upper bound follows from \cref{prop:ULA2-clip-full}. For the second, the proof of that proposition gives
\begin{align}
    \sqrt{2\gamma}\nrm{\partial_t \mu_t}
    \leq \nrm{Y_t}+\beta\rho^{L-1}R_0
    \leq 3\beta R_0
    \leqsim \beta\bigl(\sqrt{R(z)}+\sqrt{\gamma}R_B\bigr).
\end{align}
The result then follows from \cref{lem:max-B} and the preceding Girsanov bound.
\qed

\subsection{\pfref{lem:ULA2-clip}}\label{appdx:proof-ULA2-clip}
For any $t\in[0,T]$, we can bound
\begin{align}
    \frac{1}{T}\abs{W(t;\bZ)}
    \leq \abs{\tri{\partial_t \mu_t, B_T-B_t}}+\frac{1}{2}\nrm{\mu_t}^2
    \leq \abs{\tri{\partial_t \mu_t, \Bbar_T-\Bbar_t}}+ \nrm{\partial_t \mu_t}\int_0^T \nrm{\mu_s}\d s+ \frac{1}{2}\nrm{\mu_t}^2.
\end{align}
Note that
\begin{align}
    \En_{\bP_z}\exp\prn*{\lambda \abs{\tri{\partial_t \mu_t, \Bbar_T-\Bbar_t}}-\frac{\lambda^2T}{2}\nrm{\partial_t \mu_t}^2}\leq 2, \qquad \forall \lambda\geq 0.
\end{align}
Therefore, by the Cauchy--Schwarz inequality,
\begin{align}
    \En_{\bP_z} \exp\prn*{\frac{\lambda}{T}\abs{W(t;\bZ)}}
    \leq&~ \En_{\bP_z} \exp\prn*{\lambda \abs{\tri{\partial_t \mu_t, \Bbar_T-\Bbar_t}}+\frac{\lambda^2T}{2}\nrm{\partial_t \mu_t}^2+\frac{\lambda}{2}\nrm{\mu_t}^2+\frac{\lambda}{2T}\int_0^T \nrm{\mu_s}^2\d s} \\
    \leq&~ 2\En_{\bP_z} \exp\prn*{3\lambda^2T\nrm{\partial_t \mu_t}^2+\lambda\nrm{\mu_t}^2+\frac{\lambda}{T}\int_0^T \nrm{\mu_s}^2\d s} \\
    \leq&~ 2\En_{\bP_z} \exp\prn*{3\lambda^2T\nrm{\partial_t \mu_t}^2+2\lambda T\int_0^T \nrm{\partial_s \mu_s}^2\d s},
\end{align}
where we also used $\nrm{\mu_t}\leq \int_0^t \nrm{\partial_s \mu_s}\d s$. In the preceding displays, take the internal parameter to be $\lambda T$, where the rescaled parameter $\lambda\geq1$. Taking $t\sim \Unif([0,T])$ then gives
\begin{align}\label{pfeq:ULA2-weight-to-mu}
\begin{aligned}
    \En_{\bP_z}\En_{t\sim \Unif([0,T])} \exp\prn*{\lambda \abs{W(t;\bZ)}}
    \leq&~ 2\En_{\bP_z} \En_{t,s\sim \Unif([0,T])}\exp\prn*{3\lambda^2T^3\nrm{\partial_t \mu_t}^2+2\lambda T^3\nrm{\partial_s \mu_s}^2} \\
    \leq&~ 2\En_{\bP_z} \En_{t\sim \Unif([0,T])}\exp\prn*{6\lambda^2T^3\nrm{\partial_t \mu_t}^2}.
\end{aligned}
\end{align}
\qed

\subsection{\pfref{lem:ULA-fine-pathwise}}\label{appdx:ULA-fine-pathwise}
We regard $P^0_t\equiv0$. For $\ell\geq1$, define
\begin{align}
    \wt{X}^\ell_t\deq X^\ell_t-X^{\ell-1}_t.
\end{align}
For $\ell\geq2$, define
\begin{align}
    \wt{P}^\ell_t\deq P^\ell_t-P^{\ell-1}_t.
\end{align}
For $\ell\geq2$, the Brownian motion and the gradient terms cancel, so
\begin{align}
    \d \wt{X}^\ell_t=\wt{P}^\ell_t\d t,
    \qquad
    \d \wt{P}^\ell_t
    =
    -\gamma\wt{P}^\ell_t\d t
    -H_0\wt{X}^{\ell-1}_t\d t,
\end{align}
with $\wt{X}^\ell_0=\wt{P}^\ell_0=0$. Hence, for $\ell\geq2$,
\begin{align}
    \wt{P}^\ell_t
    =
    -\int_0^t e^{-\gamma(t-s)}H_0\wt{X}^{\ell-1}_s\d s.
\end{align}
Let
\begin{align}
    D^\ell_X\deq \sup_{0\leq t\leq T}\nrm{\wt{X}^\ell_t}\quad(\ell\geq1),
    \qquad
    D^\ell_P\deq \sup_{0\leq t\leq T}\nrm{\wt{P}^\ell_t}\quad(\ell\geq2).
\end{align}
Then, for $\ell\geq 2$,
\begin{align}\label{eq:fine-recursion}
    D^\ell_P\leq \beta T D^{\ell-1}_X,
    \qquad
    D^\ell_X\leq TD^\ell_P.
\end{align}
For the base position increment, we define $R_P^1=\sup_{0\leq t\leq T}\nrm{P^1_t}$, and
\begin{align}
    D^1_X
    =
    \sup_{0\leq t\leq T}\nrm{X^1_t-x}
    \leq
    T\sup_{0\leq t\leq T}\nrm{P^1_t}
    =TR_P^1.
\end{align}
Combining this with \eqref{eq:fine-recursion} yields, by induction, the following position bound for $\ell\geq1$ and momentum bound for $\ell\geq2$:
\begin{align}
    D^\ell_X
    \leq 
    T(T\sqrt{\beta})^{2\ell-2}R_P^1, \qquad
    D^\ell_P
    \leq 
    (T\sqrt{\beta})^{2\ell-2}R_P^1.
\end{align}
Summing the geometric series and using $\beta T^2\leq1/16$ gives
\begin{align}
    \nrm{X^\ell_t-X_0}\leq \frac{16}{15}TR_P^1,
    \qquad
    \nrm{P^\ell_t}\leq \frac{16}{15}R_P^1
\end{align}
for every $\ell\geq1$.

We next bound $R_P^1$.  
For any continuous Brownian path,
integration by parts gives, for all $t\leq T$,
\begin{align}
    \int_0^t e^{-\gamma(t-s)}\d B_s
    =B_t-\gamma\int_0^t e^{-\gamma(t-s)}B_s\d s.
\end{align}
Therefore, we denote $R_B'=\max_{0\leq t\leq T}\nrm{B_t}$, and
\begin{align}\label{eq:OU-conv-RB}
    \sup_{0\leq t\leq T}
    \nrm*{\int_0^t e^{-\gamma(t-s)}\d B_s}
    \leq (1+\gamma T)R_B'.
\end{align}
Then,
\begin{align}
    P^1_t
    =
    e^{-\gamma t}p
    -\int_0^t e^{-\gamma(t-s)}\nabla V(x)\d s
    +\sqrt{2\gamma}\int_0^t e^{-\gamma(t-s)}\d B_s,
\end{align}
and hence, by \eqref{eq:OU-conv-RB},
\begin{align}
    R_P^1=\sup_{0\leq t\leq T}\nrm{P^1_t}\leq \nrm{p}+T\nrm{\nabla V(x)}+\sqrt{2\gamma}(1+\gamma T)R_B'.
\end{align}
Next, we can bound (using $T\leq \frac{1}{4\sqrt{\beta}}\leq \frac{1}{4\gamma}$)
\begin{align}
    \sqrt{2\gamma}(R_B'-R_B)\leq&~ T\max_{t} \sqrt{2\gamma}\nrm{\mu_t}\leq T\beta\max_{t}(\nrm{X^L_t-X_0}+\nrm{X^{L-1}_t-X_0})
    \leq \frac{8T^2\beta}{3}R_P^1 \\
    \leq&~ 3T^2\beta(\nrm{p}+T\nrm{\nabla V(x)})+4T^2\beta\sqrt{2\gamma}R_B'.
\end{align}
This immediately implies
\begin{align}
    \sqrt{2\gamma}R_B'\leq 2\sqrt{2\gamma}R_B+\nrm{p}+T\nrm{\nabla V(x)}.
\end{align}
Combining this estimate with \eqref{eq:OU-conv-RB}, $\gamma T\leq1/4$, and \eqref{eq:fine-recursion} yields
\begin{align}
    R_P^1\leq \frac{15}{4}\prn*{\nrm p+T\nrm{\nabla V(x)}+\sqrt\gamma R_B}=\frac{15}{16}R.
\end{align}
Hence $\nrm{P_t^\ell}\leq R$, $\nrm{X_t^\ell-X_0}\leq TR$, and \eqref{eq:fine-P-increment}--\eqref{eq:fine-X-increment} follow.
\qed

\subsection{\pfref{prop:ULA2-warm-gen}}\label{appdx:proof-ULA2-warm-gen}
After the standard rescaling, we may assume $\alpha=1$ and
$\gamma\asymp1$. By \cref{thm:ST-LSI}, we can choose
$K=\ceil{C_0M\Lambda}$ so that
\begin{align}\label{eq:warm-gen-contraction-short}
    \Ren*[\ell]{\T^K\nu'}{\target}\leq
    \frac{1}{100\Delta}\Ren[\ell]{\nu'}{\target},
    \qquad 2\leq\ell\leq w.
\end{align}

Fix $I=\ceil{C_0\Lambda}$ and
$\eps=\log(w/(w-c))/I$, so that $(1+\eps)^I\leq w/(w-c)$.
Set $q=\max\{4,2w/\eps\}$, $p=e^{-\Delta}/(2I)$, and
\begin{align}
    \Rbar=20\brk*{d+5\Delta+3\log(128KI)},
    \qquad \cE=\crl{z:R(z)\leq\Rbar}.
\end{align}
Then \cref{lem:ULA-MGF} gives
\begin{align}\label{eq:warm-gen-cutoff-short}
    1-\target(\cE)\leq\prn*{\frac{p}{64K}}^3e^{-2\Delta}.
\end{align}

Let $C_1$ absorb the constants in
\cref{cor:ULA2-Renyi-onestep-crude}, and define
\begin{align}
    U\deq200\Bigl[qp+C_1+C_1Kq^2\gamma^{-1}\bigl(&\betaF^2T^5d
    +\beta^2T^3\rho^{2L-2}\Rbar+\gamma\beta^2T^4
    +\betaF^2T^6\Rbar\\
    &+\beta^3\betaF^2T^{11}\Rbar^2
    +\betaF^2\beta T^7\Rbar\bigr)\Bigr].
\end{align}
When $U\geq\Delta$, there is nothing to prove: We simply output $\nu$.
Therefore, assume $U<\Delta$.

Let $\nu_0=\nu$, let $\nu_i$ be the conditional success distribution
after $i$ stages of \cref{thm:trunc-alg}, set
$w_i=w(1+\eps)^{-i}$, and let $S_i$ be a deterministic upper bound on
$\Ren[w_i]{\nu_i}{\target}$.
Set
\begin{align}
    \eta_0\deq qp+C_1Kq^2\gamma^{-1}\beta^2T^3\Rbar,
    \qquad U_0\deq1+2\eta_0.
\end{align}
If $U_0\leq\Delta$, apply \cref{thm:trunc-alg} using the crude bound in
\cref{cor:ULA2-Renyi-onestep-crude}, with warmness $\Delta$, accuracy
$\eta_0$, and failure probability $p$. The composition below gives
$\Ren[w_1]{\nu_1}{\target}\leq U_0$, so choose $i_0=1$ and
$S_{i_0}=U_0$. Otherwise choose $i_0=0$ and $S_{i_0}=\Delta$. In either case,
$S_{i_0}\leq\min\{\Delta,U_0\}$.

For $i\geq i_0$, use failure probability $p$ and truncation accuracy
\begin{align}
    \eta_i\deq qp+C_1+C_1Kq^2\gamma^{-1}
    \Bigl(\betaF^2T^5(d+S_i^2)
    +\beta^2T^3\rho^{2L-2}\Rbar+\gamma\beta^2T^4
    +\betaF^2T^6\Rbar+\beta^3\betaF^2T^{11}\Rbar^2
    +\betaF^2\beta T^7\Rbar S_i\Bigr).
\end{align}
As long as $S_i\leq\Delta$ and $S_i>4U$, we have
$p\leq\eta_i/q$ and $\eta_i\leq S_i/2$, while
\eqref{eq:warm-gen-cutoff-short} supplies the required cutoff. Thus
\cref{thm:trunc-alg} gives
$\Ren*[q/2]{\nu_{i+1}}{\T^K\nu_i}\leq\eta_i$.

For
\begin{align}
    r_i\deq\frac{w_{i+1}(w_i-1)}{w_i-w_{i+1}}
    =\frac{w_i-1}{\eps}\leq\frac q2,
\end{align}
by \cref{lem:Renyi-composition-ULA}, the R\'enyi triangle inequality gives
\begin{align}
    \Ren[w_{i+1}]{\nu_{i+1}}{\target}
    \leq\frac{S_i}{100\Delta}
    +\frac{w_i}{w_i-1}\Ren*[r_i]{\nu_{i+1}}{\T^K\nu_i}
    \leq\frac{S_i}{100\Delta}+2\eta_i\deq S_{i+1}.
\end{align}
Substituting $K\asymp M\Lambda$, $q\asymp\Lambda$, and
$\Rbar\lesssim\Delta+\log M$ into the stated condition gives
\begin{align}
    2C_1Kq^2\gamma^{-1}\prn*{\betaF^2T^5U_0
    +\betaF^2\beta T^7\Rbar}\leq\frac{1}{100}.
\end{align}
Hence $S_{i+1}\leq S_i/2$ whenever $S_i>4U$.
We therefore stop at the first deterministic index with $S_i\leq4U$,
which occurs by stage $I$. Since
$w_I\geq w-c$, the total failure probability is at most
$Ip\leq e^{-\Delta}$ and $N\leq KI=O(M\Lambda^2)$. Substituting the
definitions of $K,q,\Rbar$ into $U$ gives the two stated bounds.

\section{Proofs from \cref{sec:mirror}}\label{appdx:mirror}

\subsection{Sub-exponentiality of the estimator}

\begin{lemma}\label{lem:mirror-subexp}
    Define
    \begin{align}
        \mathfrak M \deq \sup_x \norm{\nabla V(x)}_{x, *} < \infty \,.
    \end{align}
Then we have the following bound, where $*$ in the expectation denotes either the proposal measure $\mathbf Q$ or $\mathbf P$, and $C > 0$ is a sufficiently large absolute constant, for any $\lambda > 0$ and $t \in [0, h]$:
\begin{align}
    \E^* \exp \Bigl\{\lambda \abs{W(t, Y)}\} \leq 2\,\exp\bigl[C\lambda h\,\{(\beta + 2\,M_\phi\,\mathfrak M)\,d + \mathfrak M^2\} + C \lambda^2 h \mathfrak M^2\bigr]\,.
\end{align}
\end{lemma}
\begin{proof}
    Consider a point $y \in \mc X^*$, and let $x \in \mc X$ denote $\cT(y)$. Then,
    \begin{align}
        \nabla f(y) = \bigl(\nabla^2 \phi(x)\bigr)^{-1}\,\nabla V(x)\,.
    \end{align}
    Applying the chain rule,
    \begin{align}
        \partial_{ij}^2 f = \sum_{k, \ell} \partial_{k, \ell}^2 V(x) \cdot (\nabla^2 \phi(x))^{-1}_{ki} \cdot (\nabla^2 \phi(x))^{-1}_{\ell j} + \sum_k \partial_k V(x)\,\partial_{ij}^2 \cT_k(y)\,.
    \end{align}
    We now evaluate $\tr\bigl(\nabla^2 \phi(x)\,\nabla^2 f(y)\bigr)$. Using the expression for $\partial_{ij}^2 f$, we have
    \begin{align}
        \tr\bigl(\nabla^2 \phi(x)\,\nabla^2 f(y)\bigr) = \inner{\nabla^2 V(x), (\nabla^2 \phi(x))^{-1}} +  \underset{R(x)}{\underbrace{\sum_{i,j,k} \partial_k V(x)\,\partial_{ij}^2 \cT_k(y)\,(\nabla^2 \phi(x))_{i,j}}}\,,
    \end{align}
    where $R(x)$ will correspond to the second term, analyzed below. For this, differentiating $(\nabla^2 \phi(x))^{-1}\,\nabla^2 \phi(x) = (\nabla^2 \phi(x))^{-1}\,\nabla \cT(y) = \Id$,
    \begin{align}
        \partial_j (\nabla \cT(y))_{k, i} = -\sum_{\ell, m, n} (\nabla \cT(y))_{k\ell}\,\partial^3_{\ell, m, n} \phi(x)\,(\nabla \cT(y))_{im}\,(\nabla \cT(y))_{jn}\,,
    \end{align}
    Thus, we can write for a $\nabla^2 \phi(x)$-orthogonal basis $\{e_k\}_{k \in [d]}$,
    \begin{align}
        R(x) = -\sum_{k=1}^d \nabla^3 \phi(x) [(\nabla^2 \phi(x))^{-1}\,\nabla V(x), e_k, e_k]\,.
    \end{align}
    We can now apply self-concordance. Since the $e_k$ have unit length in the local norm,
    \begin{align}
        \abs{R(x)} \leq 2\,\sum_{k=1}^d M_\phi\,\norm{(\nabla^2 \phi(x))^{-1}\,\nabla V(x)}_{x} = 2\,\sum_{k=1}^d M_\phi\,\norm{\nabla V(x)}_{x, *}\,.
    \end{align}

    Now, since $V$ is convex and $\beta$-relatively smooth,
    \begin{align}
        \inner{\nabla^2 V(x), (\nabla^2 \phi(x))^{-1}} \leq \beta\,d\,.
    \end{align}
    Thus, we can bound
    \begin{align}
        \abs{\tr(\nabla^2 \phi(x)\,\nabla^2 f)}\leq \beta\,d + 2\,d\,M_\phi \cdot \mathfrak M\,,
    \end{align}
    where $\mathfrak M$ is given in the lemma statement. Therefore, under the proposal measure,
    \begin{align}
        \abs{\msf L f} \leq \abs{\tr(\nabla^2 \phi(x) \cdot \nabla^2 f)} \leq \beta\,d + 2\,d\,M_\phi \cdot \mathfrak M\,.
    \end{align}

    Repeat this calculation along the true path $X_s = \cT(Y_s)$, with
    \begin{align}
        \d f(Y_s) = 2 \inner{\mu_s,\,\d B_s} + \overline{\msf L} f(Y_s) \, \d s\,,
    \end{align}
    where $\overline{\msf L}$ is the generator of the true mirror Langevin diffusion. Using Cauchy--Schwarz in the dual norm, we have
    \begin{align}
        \abs{\inner{\nabla V(x), \nabla f(y)}}
        = \abs{\inner{\nabla V(x), (\nabla^2 \phi(x))^{-1}\,\nabla V(x)}}
        \leq \mathfrak M^2\,.
    \end{align}
    Consequently,
    \begin{align}
        \abs{\overline{\msf L} f} \leq \abs{\inner{\nabla V, \nabla f}} + \abs{\tr(\nabla^2 \phi(x) \cdot \nabla^2 f)} \leq \beta\,d + 2\,d\,M_\phi \cdot \mathfrak M + \mathfrak M^2\,.
    \end{align}

    Furthermore, we have, for $C_t \deq 2\,\int_0^t \inner{\mu_s,\,\d B_s}$, the quadratic variation
    \begin{align}
        \inner{C}_h = 4\,\int_0^h \norm{\mu_s}^2 \, \d s \leq 2h\mathfrak M^2\,.
    \end{align}
    Finally, it follows that
    \begin{align}
        \abs{f(Y_h) - f(Y_0)} \leq \abs{C_h} + h\,\bigl\{ (\beta + 2\,M_\phi\,\mathfrak M)\,d + \mathfrak M^2\bigr\}\,,
    \end{align}
    and combining this with the bounds on $\norm{\mu_t}^2$ and $\msf L f$, we obtain for some constant $C > 0$
    \begin{align}
        \abs{W(t, Y)} \leq \frac{1}{2}\,\abs{C_h} + Ch\,\bigl\{ (\beta + 2\,M_\phi\,\mathfrak M)\,d + \mathfrak M^2\bigr\}\,.
    \end{align}
    Indeed, as we have bounded the quadratic variation, the exponential martingale inequality also gives
    \begin{align}
        \E \exp\Bigl( \frac{\lambda}{2}\,\abs{C_h}\Bigr) \leq 2\,\exp \Bigl(C \lambda^2 h \mathfrak M^2\Bigr)\,.
    \end{align}
    Putting this all together gives the bound.
\end{proof}

We then apply this result within the framework of \cref{lem:FORS-error}.
\begin{lemma}\label{lem:mirror-one-step-clipping}
    For any order $q \geq 2$, $\varepsilon \in (0, 1]$, under the same assumptions and notation as \cref{lem:mirror-subexp}, if we choose
    \begin{align}
        B \gtrsim h\,\bigl[(\beta + 2\,M_\phi\,\mathfrak M)\,d+ \mathfrak M^2\bigr] + hq \mathfrak M^2 + \mathfrak M \sqrt{h \log (2/\varepsilon^2)}\,.
    \end{align}
    as the clipping parameter in FORS, then
    {letting $p_y$ and $\widehat p_y$ denote the exact and FORS-corrected one-step transition laws from $Y_0=y$, respectively,}
    \begin{align}
        \bigl(1 + \overline D_q(p_y \lVert \widehat p_y)\bigr)^2 \leq 1+\varepsilon^2\,.
    \end{align}
\end{lemma}
\begin{proof}
    Recall that the estimator constructed above is unbiased. The resulting error is therefore determined solely by clipping. Using \cref{lem:mirror-subexp}, we have
    \begin{align}
        \E \exp\{ \lambda \abs{W(t, Y)}\} \leq 2\,\exp\bigl\{C \lambda h\,\bigl[(\beta + 2\,M_\phi\,\mathfrak M)\,d+ \mathfrak M^2\bigr] + C\lambda^2 h \mathfrak M^2\bigr\}\,,
    \end{align}
    under the reference path measure. For $\abs{W - \pclip W}$, we have the following simple inequality for any $k \geq 1$:
    \begin{align}
        \exp(4q\,\abs{W - \pclip W}) \leq 1 + \exp(-4q k B)\,\exp(4 q k\,\abs{W})\,.
    \end{align}
    Applying the preceding lemma to bound the expectation of the second term, we obtain
    \begin{align}
        (1+{\overline{D}_q(p^* \, \lVert \, \widehat p)})^2 \leq 1 + 2\,\exp(-4kqB + Chkq\,\bigl[(\beta + 2\,M_\phi\,\mathfrak M)\,d+ \mathfrak M^2\bigr]+Chk^2 q^2\mathfrak M^2\bigr)\,,
    \end{align}
    and equivalently, it suffices to choose
    \begin{align}
        B \gtrsim h\,\bigl[(\beta + 2\,M_\phi\,\mathfrak M)\,d+ \mathfrak M^2\bigr] + hkq \mathfrak M^2 + \frac{\log (2/\varepsilon)}{4kq}\,.
    \end{align}
    Take
    \begin{align}
        k = \max\Bigl\{1, \frac{\sqrt{\log (2/\varepsilon)}}{q \mathfrak M \sqrt h}\Bigr\}\,,
    \end{align}
    to obtain the bound in the result.
\end{proof}

We provide another composition lemma which will be used in the proof below.

\sredit{
\begin{lemma}[Fixed-order composition of uniformly close kernels]\label{lem:mirror-kernel-composition}
Let $P_1,\ldots,P_K$ and $\widehat P_1,\ldots,\widehat P_K$ be Markov kernels on a common measurable state space, and suppose that, for some fixed $q \geq 2$ and $\eta\geq0$,
\begin{align}
    \sup_{k\in[K]}\sup_x\Dsys*[q]{P_k(x,\cdot)}{\widehat P_k(x,\cdot)}\leq\eta.
\end{align}
If $\nu_K=\nu_0P_1\cdots P_K$ and
$\widehat\nu_K=\nu_0\widehat P_1\cdots\widehat P_K$, then
\begin{align}
    1+\Dsys*[q]{\nu_K}{\widehat\nu_K}\leq(1+\eta)^K.
\end{align}
In particular, if $\eta\leq c\varepsilon^2/K$ for a sufficiently small absolute constant $c$, then
$\Ren*[q]{\widehat\nu_K}{\nu_K}\lesssim\varepsilon^2$.
\end{lemma}
\begin{proof}
Put the exact and approximate chains on path space, with common initial law $\nu_0$. In either direction, the likelihood ratio of one path law with respect to the other is the product of the $K$ conditional likelihood ratios. Successive conditioning, starting from the last transition, and the uniform one-step assumption show that its order-$q$ moment is at most $(1+\eta)^K$. The same argument with the two path laws interchanged gives the symmetric bound. Projecting a path to its final coordinate and using data processing proves the display. Finally,
$(1+\eta)^K\leq e^{K\eta}$ and
$\Ren[q]{\widehat\nu_K}{\nu_K}=(q-1)^{-1}\log(1+\Dren*[q]{\widehat\nu_K}{\nu_K})$.
\end{proof}
}

\subsection{\pfref{lem:mirror-breg}}\label{appdx:proof-mirror-breg}

    Consider an optimal coupling $\gamma$ with $X \sim \pi$ and $Y \sim \screplace{\mu_t}{\nu_t}$ for $\Breg_\phi(\pi \mmid \screplace{\mu_t}{\nu_t})$. We note that the mirror Langevin diffusion for $Y$ can be written in terms of a continuity equation with vector field
    \begin{align}
        v_t(y) = -[\nabla^2 \phi(y)^{-1}]\,\bigl\{\nabla V(y) + \nabla \log \screplace{\mu_t(y)}{\nu_t(y)} \bigr\}\,.
    \end{align}
    Then, differentiating the Bregman transport cost using these characteristics, we obtain
    \begin{align}
        \partial_t \Breg_\phi(\pi \mmid \screplace{\mu_t}{\nu_t}) &\leq \E\bigl[ \inner{v_t(y), \nabla_y D_\phi(x \mmid y)}] \\
        &= \E\bigl[ \inner{\nabla V(y) + \nabla \log \screplace{\mu_t(y)}{\nu_t(y)}, x-y}] \,,
    \end{align}
    where we used that $\nabla_y D_\phi(x \mmid y) = -\nabla^2 \phi(y)\,(x-y)$.
    
    Now, split the RHS into the potential energy and the entropy parts. Recalling that $V$ is $\alpha$-convex relative to $\phi$, integrating with respect to the optimal coupling $\gamma$ gives
    \begin{align}
        \E\inner{\nabla V(y), x-y} \leq \E_\pi V - \E_{\screplace{\mu_T}{\nu_t}} V-\alpha \Breg_\phi(\pi \mmid \screplace{\mu_t}{\nu_t})\,.
    \end{align}
    As for the entropy term, we use~\citet[Theorem 4]{ahn2021efficient}, which states for any measures $\mu, \pi$, coupling $y \sim \mu$, $x \sim \pi$ optimally according to the Bregman transport cost,
    \begin{align}
        \msf{Ent}(\pi) \geq \msf{Ent}(\mu)+ \E \inner{[\nabla_{W_2}\msf{Ent}(\mu)](y), x-y}\,,
    \end{align}
    where $\nabla_{W_2}$ denotes the Wasserstein-2 gradient of a functional over probability measures, which for the entropy $\msf{Ent}(\mu)$ is simply $\nabla \log \mu$. We take \screplace{$\mu_t$}{$\nu_t$} in place of $\mu$ above, and this shows that
    \begin{align}
        \E \inner{\nabla \log \screplace{\mu_t(y)}{\nu_t(y)}, x-y} \leq \msf{Ent}(\pi) - \msf{Ent}(\screplace{\mu_t}{\nu_t})\,.
    \end{align}
    As a result, we obtain
    \begin{align}
        \partial_t \Breg_\phi(\pi \mmid \screplace{\mu_t}{\nu_t}) \leq \E_\pi V - \E_{\screplace{\mu_t}{\nu_t}} V + \msf{Ent}(\pi) - \msf{Ent}(\screplace{\mu_t}{\nu_t}) - \alpha \Breg_\phi(\pi \mmid \screplace{\mu_t}{\nu_t})\,.
    \end{align}
    Using the definition of the \screplace{$\msf{KL}$}{$\Dklshort$} and $\msf{Ent}(\pi) =- \E_\pi V + C$, it follows that
    \begin{align}
        \partial_t \Breg_\phi(\pi \mmid \screplace{\mu_t}{\nu_t}) \leq -\screplace{\msf{KL}(\mu_t \mmid \pi)}{\Dkl{\nu_t}{\pi}} - \alpha \Breg_\phi(\pi \mmid \screplace{\mu_t}{\nu_t})\,.
    \end{align}
    This can be rearranged to yield the first result.

    For the integrated bounds, we use the fact that \screplace{$\msf{KL}$}{$\Dklshort$} is nonincreasing along the mirror Langevin flow by the data-processing inequality. \sredit{Writing $F(t)\deq\Breg_\phi(\pi\mmid\nu_t)$ and multiplying the differential inequality by $e^{\alpha t}$ gives
    \begin{align}
        e^{\alpha t}\Dkl{\nu_t}{\pi}
        \leq-\partial_t\bigl(e^{\alpha t}F(t)\bigr).
    \end{align}
    Hence, for $\alpha>0$,
    \begin{align}
        \frac{e^{\alpha T}-1}{\alpha}\Dkl{\nu_T}{\pi}
        \leq\int_0^T e^{\alpha t}\Dkl{\nu_t}{\pi}\,\d t
        \leq F(0),
    \end{align}
    where the first inequality uses monotonicity of $\Dkl{\nu_t}{\pi}$. For $\alpha=0$, the same calculation gives
    \begin{align}
        T\Dkl{\nu_T}{\pi}\leq F(0).
    \end{align}
    These are the two claimed bounds.}
\qed

\subsection{\pfref{thm:mirror}}\label{appdx:proof-mirror-FORS}
\sredit{Fix the R\'enyi order $q=3$ throughout.} To bound the R\'enyi divergence by $\varepsilon^2$, we take
    \begin{align}
        N_{\operatorname{iter}} \asymp \frac{1}{\alpha h} \log \frac{C\screplace{\chi^2(\mu_0 \mmid \pi)}{\Dchis{\nu_0}{\pi}}}{\varepsilon^2}\,.
    \end{align}
    Then, let \screplace{$\widehat \mu_{N_{\operatorname{iter} }h}$}{$\widehat\nu_{N_{\operatorname{iter}}h}$} be the law obtained by running FORS for $N_{\operatorname{iter}}$ iterations with step size $h$, and let \screplace{$\mu_{Nh}$}{$\nu_{N_{\operatorname{iter}}h}$} be the mirror Langevin law after the same duration. It follows that
    \begin{align}
        \screplace{\msf{R}_{3/2}(\widehat \mu_{N_{\operatorname{iter}}h} \mmid \pi)}{\Ren*[3/2]{\widehat\nu_{N_{\operatorname{iter}}h}}{\pi}} \lesssim \screplace{\msf{R}_3(\widehat \mu_{N_{\operatorname{iter}}h} \mmid \mu_{N_{\operatorname{iter}}h})}{\Ren*[3]{\widehat\nu_{N_{\operatorname{iter}}h}}{\nu_{N_{\operatorname{iter}}h}}} + \screplace{\msf{R}_2( \mu_{N_{\operatorname{iter}}h} \mmid \pi)}{\Ren[2]{\nu_{N_{\operatorname{iter}}h}}{\pi}}\,.
    \end{align}
    Use Lemma~\ref{lem:mirror-chi2} to bound the second term; we use the bound \screplace{$\msf R_2(\nu \, \lVert\, \pi) \leq \chi^2(\nu \, \lVert \, \pi)$}{$\Ren[2]{\nu}{\pi}\leq \Dchis{\nu}{\pi}$}.
    
    \sredit{For the first term, the Lipschitz condition in \cref{as:ahn-chewi} allows us to take $\mathfrak M\leq L$ in \cref{lem:mirror-one-step-clipping}. Apply that lemma at the fixed order $q=3$, with its accuracy parameter chosen so that $\varepsilon_{\mathrm{step}}^2\asymp\varepsilon^2/N_{\operatorname{iter}}$. It gives a uniform one-step bound $\eta\lesssim\varepsilon^2/N_{\operatorname{iter}}$. Apply \cref{lem:mirror-kernel-composition} to the dual transition kernels and then push forward by $\cT$; data processing gives
    \begin{align}
        \Ren*[3]{\widehat\nu_{N_{\operatorname{iter}}h}}{\nu_{N_{\operatorname{iter}}h}}\lesssim\varepsilon^2.
    \end{align}
    Since $\log(2/\varepsilon_{\mathrm{step}}^2)=O(\log(N_{\operatorname{iter}}/\varepsilon^2))$, the clipping threshold remains $O(1)$ provided}
    \begin{align}
        h \asymp \frac{1}{\bar d + L^2 \Lambda}\,, \qquad \Lambda \asymp \operatorname{polylog}(\varepsilon^{-2}, \screplace{\chi^2(\mu_0 \mmid \pi)}{\Dchis{\nu_0}{\pi}}, \zeta^{-1}) \,.
    \end{align}
    Then, it suffices to take $B = O(1)$.
    The query complexity then follows from \cref{thm:fors} by applying a union bound over $N_{\operatorname{iter}}$ iterations (taking the failure probability to be $\zeta/N_{\operatorname{iter}}$) and multiplying the per-iteration cost by $N_{\operatorname{iter}}$.

    Alternatively, using the bound from \cref{lem:mirror-breg} to control the \screplace{$\msf{KL}$}{$\Dklshort$} divergence in continuous time, we can obtain a guarantee in total variation distance by combining the triangle inequality with Pinsker's inequality; namely,
    \begin{align}
        \screplace{\msf{TV}(\widehat \mu_{Nh} \mmid \pi)}{\Dtv*{\widehat\nu_{N_{\operatorname{iter}}h}}{\pi}} \leq \screplace{\msf{TV}(\widehat \mu_{Nh} \mmid \mu_{Nh})}{\Dtv*{\widehat\nu_{N_{\operatorname{iter}}h}}{\nu_{N_{\operatorname{iter}}h}}} + \screplace{\msf{TV}(\mu_{Nh} \mmid \pi)}{\Dtv{\nu_{N_{\operatorname{iter}}h}}{\pi}}\,.
    \end{align}
    \sredit{The composition lemma and Pinsker's inequality bound the first term, while \cref{lem:mirror-breg} and Pinsker's inequality bound the second.} In the strongly convex case, we obtain
    \begin{align}
        N_{\operatorname{iter}} \asymp \frac{1}{\alpha h} \log \Bigl(1 + \frac{\Breg_\phi(\pi \mmid \screplace{\mu_0}{\nu_0})}{\varepsilon^2}\Bigr)\,,
    \end{align}
    
    Using the same union bounds as before, it then suffices to take
    \begin{align}
        h \asymp \frac{1}{\bar d + L^2 \Lambda}\,, \qquad \Lambda = \widetilde O\Bigl(\operatorname{polylog}(\bar d, \alpha^{-1}, \varepsilon^{-2}, \Breg_\phi(\pi \mmid \screplace{\mu_0}{\nu_0}))\Bigr)\,.
    \end{align}
    The query complexity again follows from Theorem~\ref{thm:fors}.

    \sredit{It remains to analyze the convex case $\alpha=0$. Write $R_0\deq\Breg_\phi(\pi\mmid\nu_0)$. By \cref{lem:mirror-breg} and Pinsker's inequality,
    \begin{align}
        \Dtv*{\nu_T}{\pi}^2\leq\frac12\Dkl{\nu_T}{\pi}
        \leq\frac{R_0}{2T}.
    \end{align}
    Thus it suffices to take $T=N_{\operatorname{iter}}h\asymp R_0/\varepsilon^2$, or
    \begin{align}
        N_{\operatorname{iter}}\asymp\frac{R_0}{h\varepsilon^2}.
    \end{align}
    Applying \cref{lem:mirror-one-step-clipping} with
    $\varepsilon_{\mathrm{step}}^2\asymp\varepsilon^2/N_{\operatorname{iter}}$ and then \cref{lem:mirror-kernel-composition}, exactly as above, gives
    \begin{align}
        \Dtv*{\widehat\nu_{N_{\operatorname{iter}}h}}{\nu_{N_{\operatorname{iter}}h}}\lesssim\varepsilon.
    \end{align}
    Taking $h^{-1}=\widetilde O(\bar d)$ and targeting a failure probability of $\zeta/N_{\operatorname{iter}}$ to each call of FORS gives
    \begin{align}
        N_{\operatorname{query}}
        =\widetilde O\prn*{\frac{\bar d\,R_0}{\varepsilon^2}},
    \end{align}
    where the suppressed logarithms include the dependence on
    $\varepsilon^{-2}$, $R_0$, and $\zeta^{-1}$.}
\qed

\section{\pfref{thm:FI-main}}\label{appdx:FI}
By \cref{thm:FI-reduction}, it suffices to instantiate \cref{as:alg-exist}. Let $\varpi\propto e^{-U}$ be any $1$-strongly log-concave and \sredit{$3$-log-smooth} distribution with mode at the origin. Initialize the underdamped Langevin sampler from $\nu_0=\normal{0,\Id/3}\otimes\normal{0,\Id}$. By the quadratic sandwich
\begin{align}
    U(0)+\frac12\,\nrm{x}^2\leq U(x)\leq U(0)+\frac32\,\nrm{x}^2\,,
\end{align}
\sredit{the density ratio is bounded by $3^{d/2}$. Since the momentum laws agree, this yields
\begin{align}
    \Ren*[\infty]{\nu_0}{\varpi\otimes\normal{0,\Id}}
    \leq\frac d2\log 3,
\end{align}
and therefore the order-$2$ and order-$6$ R\'enyi divergences are both $O(d)$.} For this auxiliary target, $\kappa=O(1)$ and hence \cref{thm:ULA-main}\sredit{(ii)}, with $q=3$ and \sredit{its accuracy parameter set to $\min\{1/2,\delta_{\mathrm{RGO}}^2\}$}, produces a phase-space law $\widehat\nu$ satisfying
\begin{align}
    \Ren*[3]{\widehat\nu}{\varpi\otimes\normal{0,\Id}}
    \leq\delta_{\mathrm{RGO}}^2.
\end{align}
\sredit{By data processing, its position marginal implements $\mathtt{Alg}$ at the accuracy required in \cref{as:alg-exist}. One such call uses}
\begin{align}
    \sredit{\wt{O}\prn[\big]{d^{1/3}+\log(M/p)}}
\end{align}
    queries to the first-order oracle for $U$ with \sredit{failure probability at most $p/M$}. Plugging this implementation of $\mathtt{Alg}$ into \cref{thm:FI-reduction} gives the stated complexity. The RGO subproblems appearing in the proximal sampler have condition number bounded by a universal constant, so the preceding implementation applies uniformly to every call. \sredit{Each first-order query to a translated and rescaled RGO potential requires one query to $\nabla V$. Finally, a union bound over the $M$ calls gives total failure probability at most $p$ and total first-order query complexity
\begin{align}
    \wt O\prn*{M\prn*{d^{1/3}+\log(M/p)}}.
\end{align}
For the final claim under higher-order smoothness, apply \cref{thm:ULA2-cold} with $w=7/2$. The $R_\infty$ bound above verifies its initialization requirement, and when $\kappa,\kappa_F=O(1)$ its displayed dimension dependence reduces to $\wt O(d^{1/5})$.}
\qed

\end{document}